\pdfoutput=1
\documentclass[11pt]{article}

\usepackage{fullpage}
\usepackage[english]{babel}
\usepackage[utf8x]{inputenc}
\usepackage[T1]{fontenc}
\usepackage{amsmath}
\usepackage{amssymb}
\usepackage{amsthm}
\usepackage{bbm}
\usepackage{mdframed}
\usepackage{bm}
\usepackage{bbold}
\usepackage{parskip}
\usepackage{thm-restate}
\usepackage{booktabs}
\usepackage{array}
\usepackage[lined,boxed,ruled,norelsize,linesnumbered]{algorithm2e}
\usepackage{hyperref}
\hypersetup{
	colorlinks=true,
	linkcolor=blue!70!black,
	citecolor=blue!70!black,
	urlcolor=blue!70!black
}
\usepackage{cleveref}
\usepackage{graphicx}
\usepackage{subcaption}

\usepackage{mathtools}
\usepackage{xcolor}

\newtheorem{theorem}{Theorem}
\newtheorem{lemma}{Lemma}
\newtheorem{fact}{Fact}
\newtheorem{corollary}{Corollary}
\newtheorem{proposition}{Proposition}
\newtheorem{definition}{Definition}
\newtheorem{model}{Model}

\newtheorem{remark}{Remark}

\newcommand{\defeq}{:=}

\newcommand{\sech}{\operatorname{sech}}
\newcommand{\norm}[1]{\left\lVert#1\right\rVert}
\newcommand{\norms}[1]{\lVert#1\rVert}
\newcommand{\normop}[1]{\left\lVert#1\right\rVert_{\textup{op}}}
\newcommand{\normf}[1]{\left\lVert#1\right\rVert_{\textup{F}}}

\newcommand{\normsop}[1]{\lVert#1\rVert_{\textup{op}}}

\newcommand{\inprod}[2]{\left\langle#1, #2\right\rangle}

\newcommand{\bK}{\mathbf K}
\newcommand{\br}{\mathbf r}

\newcommand{\eps}{\epsilon}
\newcommand{\lam}{\lambda}
\newcommand{\R}{\mathbb{R}}

\newcommand{\N}{\mathbb{N}}

\newcommand{\diag}[1]{\textbf{\textup{diag}}\left(#1\right)}
\newcommand{\half}{\frac{1}{2}}

\newcommand{\1}{\mathbf{1}}
\newcommand{\0}{\mathbf{0}}
\newcommand{\ind}{\mathbb{I}}
\newcommand{\E}{\mathbb{E}}

\newcommand{\Var}{\textup{Var}}

\newcommand{\Nor}{\mathcal{N}}

\newcommand{\dham}{\Delta_{\textup{Ham}}}

\newcommand{\id}{\mathbf{I}}

\newcommand{\dd}{\textup{d}}
\newcommand{\tO}{\widetilde{O}}

\newcommand{\Par}[1]{\left(#1\right)}
\newcommand{\Brack}[1]{\left[#1\right]}
\newcommand{\Brace}[1]{\left\{#1\right\}}
\newcommand{\Abs}[1]{\left|#1\right|}
\newcommand{\Cov}{\mathbf{Cov}}

\newcommand{\alg}{\mathcal{A}}
\newcommand{\mzero}{\mathbf{0}}
\newcommand{\Sym}{\mathbb{S}}
\newcommand{\PSD}{\Sym_{\succeq \mzero}}
\newcommand{\nnz}{\mathrm{nnz}}
\newcommand{\poly}{\mathrm{poly}}

\newcommand{\bk}{\color{black}}

\newcommand{\supp}{\mathrm{supp}}

\newcommand{\vth}{\boldsymbol{\theta}}

\newcommand{\vxi}{\boldsymbol{\xi}}
\newcommand{\sig}{\sigma}
\newcommand{\vX}{\bm{X}}
\newcommand{\vY}{\bm{Y}}
\newcommand{\sCov}{\mathrm{Cov}}

\newcommand{\twonorm}[1]{\norm{#1}_2}
\newcommand{\onenorm}[1]{\norm{#1}_1}

\newcommand{\msig}{\mathbf{\Sigma}}
\newcommand{\vtau}{\boldsymbol{\tau}}
\newcommand{\vmu}{\boldsymbol{\mu}}
\newcommand{\vlam}{\boldsymbol{\lambda}}
\newcommand{\vzero}{\mathbf{0}}
\newcommand{\vone}{\mathbf{1}}
\newcommand{\vths}{\vth^\star}
\newcommand{\vhth}{\widehat{\vth}}

\newcommand{\pisupp}{\pi_{\supp}}

\newcommand{\TV}[1]{\mathrm{TV}\left(#1\right)}
\newcommand{\TVs}[1]{\mathrm{TV}(#1)}

\newcommand{\CS}[2]{\chi^2\left(#1\|#2\right)}

\newcommand{\lazy}{\mathsf{lazy}}

\newcommand{\Bern}{\mathrm{Bern}}

\newcommand{\MN}{\mathrm{Multinomial}}
\newcommand{\Law}{\textup{Law}}

\newcommand{\BMGS}{\mathsf{BMGibbsSampler}}
\newcommand{\SAS}{\mathsf{SASPosteriorSampler}}

\newcommand{\simiid}{\sim_{\textup{i.i.d.}}}

\newcommand{\simu}{\sim_{\textup{unif.}}}

\newcommand{\Prob}{\mathbb{P}}
\renewcommand{\Pr}{\Prob}

\newcommand{\met}{\mathrm{m}}

\newcommand{\bark}{\bar{k}}
\newcommand{\hpi}{\hat{\pi}}
\newcommand{\hpisupp}{\hpi_{\textup{supp}}}

\newcommand{\ma}{{\mathbf{A}}}

\newcommand{\md}{{\mathbf{D}}}
\newcommand{\me}{{\mathbf{E}}}

\newcommand{\mg}{{\mathbf{G}}}
\newcommand{\mh}{{\mathbf{H}}}

\newcommand{\mj}{{\mathbf{J}}}
\newcommand{\mk}{{\mathbf{K}}}
\newcommand{\ml}{{\mathbf{L}}}
\newcommand{\mm}{{\mathbf{M}}}
\newcommand{\mn}{{\mathbf{N}}}

\newcommand{\mpp}{{\mathbf{P}}}
\newcommand{\mq}{{\mathbf{Q}}}
\newcommand{\mr}{{\mathbf{R}}}

\newcommand{\mw}{{\mathbf{W}}}
\newcommand{\mx}{{\mathbf{X}}}

\newcommand{\mxtx}{\mx^\top \mx}

\newcommand{\normkop}[2]{{\left\lVert#1\right\rVert_{#2, \textup{op}}}}

\newcommand{\down}{\mathrm{down}}
\newcommand{\up}{\mathrm{up}}

\newcommand{\va}{{\mathbf{a}}}
\newcommand{\vb}{{\mathbf{b}}}
\newcommand{\vc}{{\mathbf{c}}}
\newcommand{\vd}{{\mathbf{d}}}
\newcommand{\ve}{{\mathbf{e}}}

\newcommand{\vh}{{\mathbf{h}}}

\newcommand{\vm}{{\mathbf{m}}}

\newcommand{\vq}{{\mathbf{q}}}
\newcommand{\vr}{{\mathbf{r}}}
\newcommand{\vs}{{\mathbf{s}}}

\newcommand{\vu}{{\mathbf{u}}}
\newcommand{\vv}{{\mathbf{v}}}
\newcommand{\vw}{{\mathbf{w}}}
\newcommand{\vx}{{\mathbf{x}}}
\newcommand{\vy}{{\mathbf{y}}}
\newcommand{\vz}{{\mathbf{z}}}
\newcommand{\vsig}{\boldsymbol{\sigma}}

\newcommand{\calA}{{\mathcal{A}}}

\newcommand{\calC}{{\mathcal{C}}}

\newcommand{\calE}{{\mathcal{E}}}

\newcommand{\calI}{{\mathcal{I}}}

\newcommand{\calN}{{\mathcal{N}}}

\newcommand{\calP}{{\mathcal{P}}}

\newcommand{\calT}{{\mathcal{T}}}
\newcommand{\calU}{{\mathcal{U}}}

\newcommand{\calX}{{\mathcal{X}}}

\newcommand{\bbE}{{\mathbb{E}}}

\newcommand{\ba}[1]{\begin{align}#1\end{align}}

\newcommand{\cbra}[1]{\left\{#1\right\}}

\newcommand{\hZ}{\widehat{Z}}
\newcommand{\DU}{\mathsf{DU}}

\newcommand{\Lam}{\Lambda}
\newcommand{\EstZ}{\mathsf{EstimateZ}}

\crefname{assumption}{assumption}{assumptions}

\newcommand{\sign}{\mathrm{sign}}

\definecolor{burntorange}{rgb}{0.8, 0.33, 0.0}

\newcommand{\TDU}[1]{\calT^{\mathsf{DU}, #1}}

\newcommand{\set}{\textup{set}}
\newcommand{\vvec}{\mathbf{vec}}

\newcommand{\wAT}{_{\rm wAT}}
\newcommand{\AT}{_{\rm AT}}
\renewcommand{\emptyset}{\varnothing}
\newcommand{\OC}{_{\rm OC}}
\newcommand{\HM}{_{\rm HM}}

\title{High-Magnetization Sampling at Low Temperatures: \\ Ising Models and Bayesian Sparse Linear Regression}
\date{}
\author{Syamantak Kumar\thanks{University of Texas at Austin, \texttt{syamantak@utexas.edu}} \and Purnamrita Sarkar\thanks{University of Texas at Austin, \texttt{purna.sarkar@austin.utexas.edu}} \and Kevin Tian\thanks{University of Texas at Austin, \texttt{kjtian@cs.utexas.edu}} \and Yusong Zhu\thanks{University of Texas at Austin, \texttt{zhuys@utexas.edu}}}
\begin{document}
\allowdisplaybreaks
\pagenumbering{gobble}
\maketitle

\begin{abstract}
Sparsity is a powerful structural resource in optimization and statistics.
We develop frameworks for leveraging sparsity in sampling problems over the Hamming slice $\calX^d_k \defeq \{\vx \in \{\pm 1\}^d: |\{i: \vx_i = 1\}| = k\}$, in high-dimensional regimes where $k \ll d$ (i.e., where $\calX^d_k$ is \emph{highly-magnetized}). We use our frameworks to design improved samplers for canonical problems in the study of \emph{Ising models} and \emph{Bayesian sparse linear regression}.

Our first main result considers the \emph{Sherrington-Kirkpatrick} (SK) model, restricted to fixed-magnetization slices $\calX^d_k$. We give a polynomial-time sampler for fixed-magnetization SK models at any inverse temperature $\beta > 0$, under arbitrary external fields, provided $k \le c_\beta d$ for an appropriate constant $c_\beta$. By combining this result with an annealing strategy for estimating normalizing constants, this yields polynomial-time samplers for the SK model at arbitrarily low temperatures, under a sufficiently strong external field strength $h$. In the large $\beta$ limit, our framework permits sampling at field strengths within constant factors of the \emph{Almeida-Thouless line} delineating the replica symmetric and replica symmetry breaking regions \cite{deAlmeidaThouless78}, improving polynomially over the $h(\beta)$ required by recent work of \cite{BandeiraElAlaouiRodder26}. 

Our second main result concerns the measurement complexity of polynomial-time Bayesian sparse linear regression. Recent work by \cite{KumarSTZ25} shows how to sample from the canonical \emph{Gaussian spike-and-slab posterior} model with expected sparsity $k$, at any signal-to-noise ratio, given $n \gtrsim k^3 \log^3 d$ Gaussian measurements. We improve this to $n \gtrsim k^{1.5} \log^2 d + k\log^3 d$, using a common sparsity-aware framework underlying our results on Ising models.
\end{abstract}
\thispagestyle{empty}
\newpage
\tableofcontents
\thispagestyle{empty}
\newpage

\pagenumbering{arabic}

\section{Introduction}\label{sec:intro}
Harnessing sparsity is a central theme in modern high-dimensional optimization and statistics. 
From a \emph{sample complexity} perspective, sparsity is often a blessing:
classical results on Gelfand widths~\cite{Kashin77,GarnaevG84} imply that,
in the well-studied \emph{sparse linear regression (SLR)} problem, only \(n \approx k\log d\) noisy
measurements $
\vy=\mx\vths+\vxi$ 
are information-theoretically sufficient to estimate a \(k\)-sparse signal
\(\vths\in\R^d\) up to the noise level \(\norms{\vxi}_2\). From an \emph{algorithm design} perspective, however, imposing sparsity can create highly complex, nonconvex landscapes. Seminal work in compressed sensing overcame this tension by identifying structural conditions, such as the restricted isometry property, under which sparse recovery is possible in polynomial time \cite{CandesT05,CandesT06, CandesRT06,Donoho06}, leading to a broad theory of sparsity-aware optimization \cite{Wainwright19}.

Our goal is to develop an analogous framework for exploiting sparse structure in high-dimensional sampling. In the problems we study, sparsity restricts the number of simultaneously ``active'' coordinates without fixing their locations. A central motivating problem in this work is sampling from the \emph{Sherrington-Kirkpatrick (SK) model} (cf.\ Model~\ref{model:sk}) 
under sparsity constraints. 
Interestingly, understanding sampling algorithms for the SK model under sparsity has other consequences, including improved algorithms for posterior sampling in Bayesian sparse linear regression.

\textbf{SK model.} The SK model is a well-studied special case of the \emph{Ising model}, a measure over vectors of spins from the hypercube $\vx \in \calX^d \defeq \{\pm 1\}^d$. Such $\vx$ can be identified with a set $S \subseteq [d]$, the locations of positive spins $\vx_i = 1$. In Ising models, the underlying measure is
\begin{equation}\label{eq:ising_intro}
    \mu(\vx)
    \propto
    \exp\Par{
        \beta\Par{
            \half \vx^\top \mj \vx
            +
            \vh^\top \vx
        }
    },\quad \vx \in \calX^d,
\end{equation}
governed by an interaction matrix \(\mj \in \R^{d \times d}\), an external field \(\vh \in \R^d\), and an inverse temperature \(\beta>0\). Ising models are a canonical testbed across statistical physics, machine learning, and theoretical computer science \cite{WainwrightJ08,LevinPW09,talagrand2010mean}. For suitable $\mj$, the Gibbs landscape in \eqref{eq:ising_intro} is known to undergo qualitative phase transitions as  \(\beta\) varies \cite{talagrand2010mean}.  
Efficient algorithms exist for sampling under the SK model, where \(\mj\) is drawn from the Gaussian orthogonal ensemble (GOE),  for $\beta$ up to a universal constant $c$ \cite{EldanKZ22, AnariJKPV22, AlaouiMS22, AnariKV24, DaviesLSS26}, while conditional hardness has been demonstrated at a higher constant threshold $c' > c$ \cite{AlaouiMS22}.

\textbf{Sparsity from field strength.} As intuition for our central \emph{high-magnetization SK model}, to be introduced next, we describe a well-studied analog: the SK model under a strong positive external
field
\(  \mathbf{h}=\theta \mathbf{1}_d
\).
When $\theta$ is large,
\eqref{eq:ising_intro} is biased toward $\1_d$, driving the minority-spin set to become sparse.  The high-field strength model thus exhibits a soft form of sparsity. 

On the
algorithmic side,~\cite{BandeiraElAlaouiRodder26} proved polynomial-time
mixing of the Glauber dynamics for the SK model at \emph{every inverse temperature} $\beta > 0$, under a high enough field strength $h \defeq \beta\theta$. They derive their sampler as a consequence of a more general result that proves mixing in Ising models \eqref{eq:ising_intro},
under a condition on the sparse operator norm of $\mj$ at a constant sparsity scale $k = \Theta(n)$.

On the geometric side, the Almeida--Thouless (AT) line is the
standard benchmark for replica symmetry in the
\((\beta,h)\)-plane. This line, derived by \cite{deAlmeidaThouless78} using the replica method (and described in Appendix~\ref{sec:sk_at_line}), predicts a qualitative shift in the behavior of \eqref{eq:ising_intro} as $h$ grows compared to $\beta$. This line provides a
natural field strength scale against which to compare \cite{BandeiraElAlaouiRodder26}
to our results.

The role of sparsity in \cite{BandeiraElAlaouiRodder26} is implicit: a large field strength $h$ causes
independent replicas to have high overlap, so their disagreement set is
typically sparse. The relevant
interaction depends on a principal submatrix of $\mj$, giving a dependence on sparse operator norms. This suggests a complementary question: can high magnetization itself, imposed as a hard constraint rather than induced by an external field, make sampling in the SK model tractable? 

\textbf{Sparsity from high magnetization.}
High-magnetization regimes  have long played an important role in statistical physics, e.g., the study of spontaneous magnetization and large deviations \cite{yang1952,Bonati2014,Aizenman_2014,Ellis12}. For measures on $\calX^d$, high magnetization is a natural analog of sparsity. Under our convention, the positive spins are the active coordinates: an element $\vx$ of
\ba{\label{eq:Xdk}
    \calX^d_k
    \defeq
    \Brace{
        \vx\in\calX^d:
        \bigl|\{i:\vx_i=1\}\bigr|=k
    }
}
has exactly \(k\) positive spins and, for \(k\leq \frac d 2\), magnetization magnitude \(d-2k\). Taking \(k\ll d\) imposes a sparsity constraint, while maintaining an outcome space of exponential size \(\binom{d}{k} \approx \exp(k\log \tfrac{d}{k})\). We study the SK model restricted to both the fixed-magnetization slice $\calX^d_k$, and its bounded-magnetization counterpart
    $\calX^d_{\leq k}
    \defeq
    \bigcup_{i=0}^k \calX^d_i.$ 
We now state our first central problem.
\begin{gather*}
\textit{For every fixed } \beta < \infty, \textit{ is there a polynomial-time sampler from the} \\
\textit{fixed-magnetization SK model  whenever } k\le c_\beta d \textit{ for a fixed } c_\beta > 0?
\end{gather*}
Theorem~\ref{thm:informal-trickledown} answers this question affirmatively, taking the algorithm as the \emph{down-up walk}. 

\textbf{Bayesian SLR.}
Our second motivating example is a Bayesian variant of SLR. SLR is often phrased as an optimization problem: given noisy measurements $(\mx, \vy = \mx \vths + \vxi)$, return a $k$-sparse $\vhth \in \R^d$ (approximately) minimizing the residual norm $\norms{\mx\vhth - \vy}_2$. When $\mx$ is RIP, low residual error implies accurate estimation \cite{CandesRT06}, so optimization produces a good \emph{point estimate} of $\vths$. 

In many applications, however, it is preferable to sample $\vhth$ from a \emph{distribution} over plausible signals, e.g., to model uncertainty in support estimation (variable selection) \cite{MitchellB88, GeorgeM93}. In the well-established \emph{Bayesian SLR} model, the noise $\vxi$ is Gaussian with coordinatewise variance $\sig^2$, where $\sig^{-1} > 0$ is a \emph{signal-to-noise ratio}. For an appropriate prior $\pi$ over sparse signals $\vths$, the goal is then to sample from the posterior induced by the observations:
\begin{equation}\label{eq:post_intro}\vhth \sim \pi\Par{\cdot \mid \mx, \vy}, \text{ where } \vy = \mx \vths + \vxi,\quad \vths \sim \pi,\quad \vxi \sim \Nor\Par{\0_n, \sig^2 \id_n}.\end{equation}
Samples from the posterior density can then be used in downstream tasks, e.g., constructing credible intervals.
In the statistics literature, $\pi$ is often taken to be the \emph{spike-and-slab} prior, 
\begin{equation}\label{eq:sas_intro}
\pi = \bigotimes_{i \in [d]} \Par{1 - \frac k d} \delta_0 + \frac k d \Nor(0, 1),
\end{equation}
where each signal coordinate $\vths_i$ is independently set to $0$ except with low probability (so the expected sparsity is $k$). The induced \emph{spike-and-slab posterior sampling} problem is often called the ``theoretical gold standard'' for modeling uncertainty in variable selection \cite{JohnstoneS04, CarvalhoPS09, IshwaranS11, castillo2012needles, Roc18, PolsonS19}. Unfortunately, this task poses a notorious computational challenge \cite{CSHVdV15}, and many heuristics have been developed as approximations \cite{BaiRG21}. 

Recently, several works developed provable methods for this sampling problem \cite{YangWJ16, MontanariW26}, including an algorithm by \cite{KumarSTZ25} which samples from the posterior density \eqref{eq:post_intro}, \eqref{eq:sas_intro} given a sublinear-in-$d$, $n \gtrsim k^3 \log^3 d$ measurements. The counterpart optimization problem is known to be feasible even when $n \approx k\log d$, motivating our second central problem.
\begin{gather*}
\textit{What is the measurement threshold $n$ at which spike-and-slab posterior sampling} \\
\textit{admits polynomial-time algorithms, for any signal-to-noise ratio?}
\end{gather*}
We show that the complex SLR posterior density admits enough structure to be captured by our analysis framework for the high-magnetization SK model, and give a polynomial-time posterior sampling algorithm at $n\gtrsim  k^{1.5}\,\textup{polylog}(d)$ measurements (Theorem~\ref{thm:informal-slr}).

\subsection{Our results}\label{ssec:results}
In this section, we overview our main results. To obtain these results, we
develop a suite of technical tools for exploiting sparsity in sampling, described at more depth in Section~\ref{ssec:techniques}.

\textbf{High-magnetization SK models.} Our first main result concerns sampling in high-magnetization SK models, where $\mj$ is GOE and the external field $\vh$ is arbitrary. As a benchmark, \cite{AlaouiMS22} showed conditional hardness for sampling in SK models at large inverse temperatures $\beta > c'$ for a constant $c'$. We show that after fixing the magnetization $k \le c_\beta d$ of the spin vector $\vx \in \calX^d_k$, the SK model admits polynomial-time sampling at any temperature.

\begin{theorem}[Informal; see Theorem~\ref{thm:sk_trickle}]\label{thm:informal-trickledown}
For any fixed $\beta > 0$, $\vh \in \R^d$, and $\mj \sim \textup{GOE}(d)$, there is a constant $c_\beta > 0$ such that if $k \le c_\beta d$, we can sample from the SK model \eqref{eq:ising_intro} restricted to the Hamming slice $\calX^d_k$ in polynomial time, with high probability over $\mj$.
\end{theorem}

The algorithm in Theorem~\ref{thm:informal-trickledown} is the canonical \emph{down-up walk} over the Hamming slice $\calX^d_k$, and we bound its runtime by a relatively mild $\tO(d^2 + dk^2\max(1, \beta\norm{\vh}_\infty))$.\footnote{In this introduction only, we use $\tO$ to suppress logarithmic factors.} Beyond Theorem~\ref{thm:informal-trickledown}, Section~\ref{sec:trickledown} develops a more general framework for proving Poincar\'e inequalities on the down-up walk for Ising models, using the trickle down (``local-to-global'') theorem of \cite{oppenheim2018local, alev2020improved}. For example, we state an analogous result for \emph{Gaussian Hopfield} models in Corollary~\ref{cor:trickle_ghop}.

Due to Theorem~\ref{thm:informal-trickledown}, we recover a variant of the main result of \cite{BandeiraElAlaouiRodder26} by exploiting the relationship between high field strength and high magnetization.\footnote{Formally, this follows by combining Corollary~\ref{cor:fast_mixing_sparse_set_dobrushin_gen} and Lemma~\ref{lem:gaussian_cut_concentration}.} We do note that \cite{BandeiraElAlaouiRodder26} directly analyze the \emph{Glauber dynamics}, perhaps the simplest sampling algorithm over $\calX^d$, whereas we use an \emph{annealing scheme} (Section~\ref{sec:bm_sampling_anneal}) to reduce bounded-magnetization sampling to fixed-magnetization sampling.

Our framework yields other interesting consequences beyond the linear sparsity regime in Theorem~\ref{thm:informal-trickledown}. For example, for $k = O(1)$ independent of $d \to \infty$, our results imply polynomial-time sampling in high-magnetization SK models for polynomial inverse temperatures $\beta$, up to $O(\sqrt{d/\log(d)})$ (cf.\ Corollary~\ref{cor:near_proportional_sk}). Prior results exploiting notions of sparsity (albeit different from ours) to sample from Ising models at subconstant temperatures only tolerated $\beta \approx \log d$ \cite{Carlson2022, KuchukovaPappikPerkinsYap2025}. This comparison is discussed at more length in Section~\ref{ssec:related} and Appendix~\ref{app:critical}.

\textbf{The Almeida-Thouless line.} From a quantitative perspective, the AT line is a useful benchmark for comparing our sampling result with \cite{BandeiraElAlaouiRodder26}. The AT line predicts a phase transition in the geometry of the SK model under an external field $\vh = \frac h \beta \1_d$, once the field strength $h \ge h\AT(\beta)$ is large enough as a function of $\beta$, where $h\AT(\beta) = (1 + o(1))\beta\sqrt{2\log\beta}$ (Lemma~\ref{lem:at_line}).

Leveraging a simple reduction from high field strength sampling to high-magnetization sampling (Lemma~\ref{lem:gaussian_cut_concentration}), we show in Theorem~\ref{thm:unrestricted_high_field_sk} that Theorem~\ref{thm:informal-trickledown} implies sampling from the SK model whenever $h \ge h_{\rm{HM}}(\beta)$ for a threshold $h_{\rm{HM}}(\beta) = \sqrt{2}(1 + o(1)) h\AT(\beta)$, i.e., within constant factors of the AT line. We also show in Theorem~\ref{thm:oracle_centered_sampling} that, if given access to the signs of the mean vector $\E[\vx]$ under the SK model, a modification of our sampler succeeds at any field strength $h \ge (1 + o(1)) h\AT(\beta)$. By comparison, the result of \cite{BandeiraElAlaouiRodder26} applies whenever $h = \Omega(\beta^2 \sqrt{\log\beta})$, a polynomial factor larger (see discussion in Appendix~\ref{sec:sk_at_line}). The tighter range of $h$ tolerated by our framework is a result of our basic sampler in Theorem~\ref{thm:informal-trickledown} applying for an arbitrary external field $\vh$.
\bk

\textbf{Bayesian SLR.} 
Finally, we apply our sparse sampling frameworks to spike-and-slab posterior sampling. We obtain a state-of-the-art measurement complexity for a canonical variant of the problem, stated in \eqref{eq:post_intro}, \eqref{eq:sas_intro}, and studied by \cite{KumarSTZ25, MontanariW26}.

\begin{theorem}[Informal; see Theorem~\ref{thm:sas_post_sampler}]\label{thm:informal-slr}
Let $\mx\in\R^{n\times d}$ have i.i.d.\ entries distributed as
$\calN(0,\tfrac{1}{n})$. If
\begin{equation*}
    n
    =
    \Omega\left(
        \Par{k + \log\Par{\frac 1 \delta}}^{1.5}
        \log^2\left(\frac{d}{\delta}\right)
        +
        \Par{k + \log\Par{\frac 1 \delta}}
        \log^3\left(\frac{d}{\delta}\right)
    \right),
\end{equation*}
then, for every $\sig>0$, there is a polynomial-time algorithm that samples
from $\pi(\cdot\mid\mx,\vy)$, defined in~\eqref{eq:post_intro}
and~\eqref{eq:sas_intro}, within total variation distance $\delta$, with
probability at least $1-\delta$ over the model.
\end{theorem}
As in prior work, there are two sources of failure in Theorem~\ref{thm:informal-slr}. Namely, the model \eqref{eq:post_intro}, \eqref{eq:sas_intro} may fail to produce a sparse signal $\vths$ or bounded noise $\vxi$ (inhibiting tractability of the problem), and the sampling algorithm itself has an approximation error. Our formal result, Theorem~\ref{thm:sas_post_sampler}, is more general and can handle nonuniform inclusion weights in the prior (see Model~\ref{model:sas_basic}). Moreover, the runtime of Theorem~\ref{thm:informal-slr} is relatively practical, e.g., it scales linearly in $nd$.

Our $n \approx k^{1.5} \log^2 d$ requirement improves quadratically in its dependence on $k$ upon the previous state-of-the-art sampler by \cite{KumarSTZ25}, which uses $n = \Omega(k^3 \log^3 d)$ measurements (see also \cite{MontanariW26}, who gave a result in the regime $n = \Omega(d)$). Interestingly, the $k^{1.5}$ bottleneck appears inherent to our approach (discussed in the following Section~\ref{ssec:techniques}).
This motivates the tantalizing open question of whether spike-and-slab posterior sampling is tractable at $n = \Omega(k\log d)$ measurements, which would close the gap between frequentist and Bayesian SLR.

\subsection{Our techniques}\label{ssec:techniques}

In this section, we overview the main proof ideas behind our sparse
Dobrushin (Section~\ref{sec:dobrushin}) and trickle down frameworks (Section~\ref{sec:trickledown}), our annealing reduction from bounded-magnetization to
fixed-magnetization sampling (Section~\ref{sec:bm_sampling_anneal}), and our application to Bayesian SLR (Section~\ref{sec:slr}).  

\textbf{Sparse Dobrushin condition.}
In Section~\ref{sec:dobrushin}, we give a warm-up path coupling analysis of the down-up walk illustrating why restricting $k \ll d$ can make low-temperature sampling easier. For a measure supported on $k$-sized subsets, we compare the conditional laws of the up step from neighboring $(k-1)$-sized cores, after excluding the coordinates on which the two cores differ. If the resulting total variation discrepancy is $O(\frac 1 k)$,
the walk contracts in Hamming distance and mixes in
$
O\!\left(k \log \frac{k}{\varepsilon}\right)
$
steps (Lemma~\ref{lem:du_coupling}). This framework depends on a sparse variant of the classical \emph{Dobrushin influence matrix} \cite{Dobruschin68}, so we term it a sparse Dobrushin condition (Definition~\ref{def:sparse_dobrushin}).

For fixed-magnetization Ising models, this discrepancy is controlled by
$
\beta \max_{i \neq j} |\mj_{ij}|,
$
independently of the external field $\vh$. Sparsity sets the required discrepancy bound at $\approx \frac 1 k$, i.e., a relaxed bound at higher magnetizations. The maximum entry magnitude under the SK model scales as $\approx \sqrt{\log d/ d}$, so this warm-up result already shows a variant of Theorem~\ref{thm:informal-trickledown} at the higher magnetization level $k \le c_\beta \sqrt{d/\log d}$. The sparse Dobrushin condition is simple and broadly applicable, but it cannot exploit cancellation among signed interactions. This motivates our main technique for extending to the range $k = \Theta(d)$, based on spectral expansion and the trickle down theorem.

\textbf{Spectral mixing via trickle down.}
To obtain sharper parameter ranges in fixed-magnetization Ising models,
we leverage the trickle
down theorem \cite{oppenheim2018local,alev2020improved}, a foundational result in the study of high-dimensional expansion. For a measure
$\pi$ over
$\binom{\calU}{k}$, and a core
$R\in\binom{\calU}{k-2}$, define the link graph of $R$ on
$\calU\setminus R$ to have edge weights
$\mw_{ij}:=\pi(R\cup\{i,j\})$, and let $\mpp$ be the corresponding random walk matrix. The trickle down theorem (Lemma~\ref{lem:trickledown}) shows that if we can show $\lam_2(\mpp)= O(\frac 1 k)$ for all cores $R$, 
then the down-up walk satisfies a $O(\frac 1 k)$-Poincar\'e inequality.
Thus, the global mixing problem reduces to proving uniform spectral
expansion of all the induced $\mpp$.

For a fixed-magnetization Ising model, each link has a particularly useful
form. Fix a core $R$, and let $\br\in\{\pm1\}^d$ be the spin vector whose
positive coordinates are $R$. For some interaction matrix $\bK$, Lemma~\ref{lem:ising_trickle} shows that the link
weights satisfy, for an appropriate vector $\va \in \R^{\calU \setminus R}$,
\begin{equation*}
    \mw_{ij}
    \propto
    \va_i \va_j\exp(\bK_{ij}),
    \quad
    \mk_{ij}:=4\beta\mj_{ij}
    \quad (i\neq j).
\end{equation*}
We view $\mw$ as a bounded perturbation (parameterized by $\mk$) of the rank-one weights $\va\va^\top$: when $\mk$ is the all-zeroes matrix, $\lam_2(\mw) \le 0$ follows simply by a rank argument.\footnote{Formally, $\mw = \va\va^\top - \diag{\va}^2$ after removing self-loops, but removing $\diag{\va}^2$ can only improve $\lam_2(\mpp)$.} Our main trickle down framework gives a tighter characterization of $\lam_2(\mpp)$ as a function of $\mk$. Specifically, Lemma~\ref{lem:mixed_quadratic} shows using a second-order Taylor expansion of the exponential that if
\begin{equation}\label{eq:mixed_norm_intro}\Abs{\vu^\top \mk \vu} \le \tau\norm{\vu}_1 \norm{\vu}_2 + \tau^2\norm{\vu}_1^2,\end{equation}
then $\lam_2(\mpp) = O(\tau^2)$. 
This sets the required bound on $\tau$ at $\approx k^{-1/2}$ for the trickle down argument. As a point of comparison, the $\tau^2 \approx \frac 1 k$ term in this argument, along with $\vu^\top \mk \vu \le (\max |\mk_{ij}|)\norm{\vu}_1^2$, already qualitatively recovers the sparse Dobrushin condition (Lemma~\ref{lem:max_entry_recover}).

We next show that the mixed norm condition \eqref{eq:mixed_norm_intro} allows more fine-grained control of $\mk$, in terms of its \emph{sparse operator norms} $\normkop{\mk}{s}$, i.e., the largest $\normsop{\mk_{S \times S}}$ among any $s$-sized sets $S$. This argument proceeds by applying a \emph{shelling decomposition}, a classic technique from the sparse recovery literature \cite{CandesRT06} that places the ``effective sparsity'' of a vector $\vu$ on the scale of $\frac{\|\vu\|_1^2}{\|\vu\|_2^2}$. 

A basic application of this strategy (Lemma~\ref{lem:comparison_mixednorm}) shows that if $\normkop{\mk}{s} \lesssim \tau\sqrt{s} + \tau^2 s$ at all scales $s \in [d]$, then  \eqref{eq:mixed_norm_intro} holds. Plugging in standard bounds on the sparse operator norms of various matrix ensembles now already gives Theorem~\ref{thm:sk_trickle} up to a logarithmic loss in the tolerated $k$ range (Corollary~\ref{cor:near_proportional_sk}), as well as our strongest conclusion for \emph{Gaussian Hopfield} Ising models (Corollary~\ref{cor:trickle_ghop}). Our final application to SK models in the linear magnetization regime $k = \Theta(d)$ (Theorem~\ref{thm:sk_trickle}) uses more fine-grained estimates of sparse operator norms for GOE matrices.

\textbf{Annealing.} We take a brief detour to discuss a complementary part of our framework: a technique for lifting fixed-magnetization samplers to the bounded-magnetization setting (measures supported on $S \subseteq \calU : |S| \le k$). 
Our approach is based on the fact that \emph{bounded-magnetization measures are mixtures of fixed-magnetization measures}, with weights proportional to \emph{normalizing constants}. That is, to approximate a global density $\pi$ on $\calU_{\le k}$ to total variation $\delta$, it is enough to estimate
\[Z_i \defeq \sum_{\omega \in \calU_i} \pi(\omega) \text{ for all } 0 \le i \le k\]
to multiplicative error $O(\delta)$.
We formalize this argument in Lemma~\ref{lem:tv_bound_mixtures}.

Conveniently, there is a rich literature in theoretical computer science reducing between the problems of \emph{counting} (e.g., normalization constant estimation) and \emph{sampling}. An existing result, Theorem 6 in \cite{Kolmogorov18}, is essentially black-box applicable to our setting, giving $\delta$-multiplicative estimates to any $Z_i$ by using oracle calls to \emph{fixed-magnetization} samplers. By building upon the estimator of \cite{Kolmogorov18}, we state our generic reduction from bounded-magnetization sampling to fixed-magnetization sampling in Lemma~\ref{lem:annealing_sampler_gen_ada}, and give an example of its use for the SK model in Corollary~\ref{cor:fast_mixing_sparse_set_dobrushin_gen}.

\textbf{Spike-and-slab posterior sampling.} We finally turn to our second main application: Bayesian SLR with a spike-and-slab prior, as studied by \cite{KumarSTZ25, MontanariW26}. The primary challenge is to correctly sample the support $S \defeq \supp(\vths) \subseteq [d]$ from the posterior density 
\[\pisupp(S \mid \mx, \vy) \defeq \pi(\{\vth: \supp(\vth) = S\} \mid \mx, \vy).\]
As derived in prior work (cf.\ Fact~\ref{fact:pisupp}), this density is proportional to a closed-form expression:
\begin{equation}\label{eq:pisupp_intro}\pisupp(S) \propto \Par{\frac k {d - k}}^{|S|}\exp\Par{\half\norm{\vb_S}^2_{\ma_S^{-1}}}\frac 1 {\sqrt {\det \ma_S}},\end{equation}
where $\ma, \vb$ are induced by the measurements $(\mx, \vy)$ and defined in \eqref{eq:pisupp_def}. Note that the first term in the above expression can be absorbed into an external field.

Our posterior sampler has three components. The first follows
\cite{KumarSTZ25} and uses a sparse recovery preprocessing step (Proposition~\ref{prop:approx_post_i}) that
identifies coordinates whose inclusion is nearly deterministic from $(\mx, \vy)$, providing regularity to the residual posterior density. The second applies our annealing procedure from Section~\ref{sec:bm_sampling_anneal} to reduce the problem to a fixed-magnetization variant of \eqref{eq:pisupp_intro}.

The remaining step is to use our trickle down framework to demonstrate mixing of the down-up walk on fixed-magnetization slices of the residual posterior. The induced weight matrices $\mw$ after pinning a core are more complex than in the SK setting, because of the inverse and determinantal terms in \eqref{eq:pisupp_intro}. In particular, the resulting perturbation $\mk$ depends on the pinned core through a Schur complement correction, which yields various dependencies. Due to a union bound over all possible cores, our uniform estimate of $\tau$ in \eqref{eq:mixed_norm_intro} scales as $n^{-1/2} + \frac k n$ (e.g., see \eqref{eq:entrywise} in Lemma~\ref{lem:combine_params}), and setting this to $k^{-1/2}$ as required by the trickle down framework gives a $n \gtrsim k^{1.5}$ bottleneck.
Removing this $\sqrt{k}$ factor beyond the measurement complexity of (frequentist) sparse recovery is an exciting problem, that we leave open as a testbed of ``average-case'' trickle down theorems avoiding the union bound over worst-case dependencies suffered by our approach.

\subsection{Related work}\label{ssec:related}

\textbf{Fixed- and bounded-magnetization sampling.} The most conceptually relevant prior algorithmic works are by \cite{Carlson2022, KuchukovaPappikPerkinsYap2025}, both of which study fixed-magnetization problems that exhibit improved phase transitions or critical $\beta$ as the sparsity (positive spin count) $k$ becomes small. For example, Theorem 2 in \cite{Carlson2022} shows that for  ferromagnetic Ising models with bounded degree $\Delta$, there is a critical ``tree threshold'' $\beta_c(\Delta) = O(\frac 1 \Delta)$ such that for $\beta > \beta_c(\Delta)$, the model undergoes a computational phase transition at a certain magnetization level. Notably, the sparsity tradeoff required by their paper beyond $\beta > \beta_c(\Delta)$ decays \emph{exponentially} in $\beta$: we give a more formal derivation in Appendix~\ref{app:critical}, but e.g., for constant $\Delta$, they require
\[k = d \exp\Par{-\Theta_\Delta(\beta)},\]
which only implies a nontrivial setting ($k \neq 0$) up to $\beta \approx \log d$. Similarly, Theorem 1 of \cite{KuchukovaPappikPerkinsYap2025} shows that the Kawasaki dynamics \cite{Kawasaki66} (a standard magnetization-conservative dynamics) exhibits a phase transition at $\beta > \beta_c(\Delta)$, but their magnetization thresholds inherit the same exponential decay in $\beta$, and hence also apply only up to $\beta \approx \log d$. On the other hand, our results are derived through more general sparsity-aware analysis frameworks (Definition~\ref{def:sparse_dobrushin}) that continue to improve as $k \to 1$, tolerating $\beta$ up to $\poly(d)$. This improvement is closer in spirit to existing results in the sparse recovery literature, where general structural conditions (e.g., RIP) yield sample complexity requirements that improve monotonically with the sparsity level $k$.

There have been other works that study sampling in high-magnetization regimes, that are motivated by giving analysis frameworks for natural sampling dynamics at fixed magnetization (e.g., the down-up walk or Kawasaki dynamics), rather than obtaining improved thresholds as $k \to 1$. For example, \cite{BauerschmidtBodineauDagallier2024} study the local Kawasaki dynamics on random $\Delta$-regular graphs, and prove rapid mixing for $\beta = O(\Delta^{-1/2})$; notably, their thresholds do not improve as $k \to 1$. Similarly, a line of work has obtained estimates for the mixing time of the down-up walk that improve as $k$ decreases \cite{AnariLGV19, CryanGM19, AnariKV24}, e.g., that it mixes in $O(k\log k)$ steps for \emph{strongly log-concave} measures, but their focus was not the relationship between $k$ and the allowable inverse temperature $\beta$.

\textbf{High-field strength spin glasses.} The Almeida-Thouless line \cite{deAlmeidaThouless78} is the canonical benchmark for the replica symmetry phase transition in the SK model with a homogeneous external field $\theta\1_d$. From a geometric perspective, a recent work of \cite{Lopatto26} established replica symmetry throughout the AT region, and \cite{kusuoka2026quantitative} obtained quantitative overlap concentration in the strict AT region. Stronger forms of quantitative overlap concentration in a more restrictive region (cf.\ Definition~\ref{def:wat_cond}) have also appeared in the literature \cite{Talagrand11, JagannathTobasco17}; in particular, we use a bound from \cite{rajaraman2026markov} in our reduction from high-field strength sampling to high-magnetization sampling (Appendix~\ref{sec:sk_at_line}). On the algorithmic side, \cite{BandeiraElAlaouiRodder26} prove polynomial-time mixing of the Glauber dynamics in the SK model at sufficiently high field strength, by exploiting replica overlap and sparse operator norm bounds. We quantitatively improve upon the field strength tolerance of \cite{BandeiraElAlaouiRodder26} to within constant factors of the AT line, or within $1 + o(1)$ if granted the signs of the mean vector.

\textbf{Other applications of sparsity in sampling.} Recent works by \cite{BandeiraElAlaouiRodder26, DaviesLSS26} both use sparse operator norm bounds to derive mixing times for sampling from Ising models, related to the analytical core of our work (e.g., Lemma~\ref{lem:comparison_mixednorm}). Theorem 1 of \cite{BandeiraElAlaouiRodder26} shows that if the sparse operator norm is bounded at a sparsity level $k = \Theta(d)$, then the Glauber dynamics mixes in polynomial time under a sufficiently large external field. In a similar spirit, Section 7 of \cite{DaviesLSS26} uses sparse operator norm bounds to prove rapid mixing of the SK model on small Hamming balls, but only tolerates $\beta < \half$, as opposed to the arbitrary $\beta > 0$ handled by our Theorem~\ref{thm:informal-trickledown}.

More generally, a recurring theme in high-dimensional sampling is that sparsity and fixed-cardinality
constraints can influence algorithmic landscapes. 
Recent work on discrete sampling has obtained sharp thresholds and improved mixing guarantees for fixed-size
structures such as independent sets and matchings. In particular, \cite{DP21} identified the
computational threshold for approximately counting and sampling independent sets of a given
size in bounded-degree graphs, \cite{JMPV23} proved optimal \(O(k\log n)\) mixing of the
down-up walk for independent sets of size \(k\), and
\cite{JM24} established polynomial-time mixing of the down-up walk for matchings of size $k$. Recent progress on Ising models has also extended beyond dense mean-field settings: for example, \cite{LiuMRW24} proves near-linear time mixing of the Glauber dynamics for sparse random Ising models, including the Viana--Bray spin glass, and also treats certain interaction matrices arising from stochastic block models. 

\textbf{Bayesian sparse linear regression.} Spike-and-slab priors and their induced posteriors are classical tools for Bayesian variable selection \cite{MitchellB88,GeorgeM93}. A
large statistical literature studies posterior contraction and uncertainty quantification for sparse
priors, including spike-and-slab formulations and closely related continuous relaxations
\cite{CSHVdV15,Roc18}. 
Notably, in the moderate signal-to-noise ratio (SNR) regime, the support posterior measure is known to be highly multimodal, which poses an algorithmic challenge \cite{CSHVdV15}. In comparison, the modes collapse at a high SNR, and the posterior nearly reduces to the prior under a low SNR.

On the computational side, many works target posterior modes or
tractable approximations, including expectation maximization, Lasso-based procedures, and variational Bayes
approximations \cite{RG14,RG18,RS22,MS22}. While these works demonstrate the empirical performance of their algorithms, they do
not yield end-to-end algorithmic guarantees for the sampling problem we consider. Among provable posterior
samplers, \cite{YangWJ16} analyzes a Metropolis-Hastings chain under a truncated sparsity prior, which applies only under high or low signal-to-noise ratio regimes (see discussion after their Eq.\ (10)), while \cite{MontanariW26} gives a
polynomial-time sampler via measure decomposition when the number of measurements $n$ grows at
least linearly with the ambient dimension $d$. A complementary diffusion-based approach for linear inverse
problems was recently developed in \cite{BrunaH24} (and analyzed in continuous time), though their framework does not seem to directly apply to spike-and-slab posterior sampling at $n \ll d$.
 The closest prior work is \cite{KumarSTZ25}, which
gave the first provable spike-and-slab posterior samplers that apply at arbitrary SNRs while
allowing $n$ to remain sublinear in \(d\). Our result quadratically improves upon the $k$ dependence of \cite{KumarSTZ25}.

\section{Preliminaries}
\label{sec:prelim}

In this section we develop preliminaries for the rest of the paper. Section~\ref{ssec:notation} provides notation used throughout. Section~\ref{ssec:mcmc} gives basic notation and facts about Markov chains used in our analysis. Section~\ref{ssec:models} introduces the main statistical models we consider for our applications.

\subsection{Notation}\label{ssec:notation}

\textbf{General notation.} We use $\lesssim$, $\gtrsim$, and $\approx$ in informal exposition only to suppress polylogarithmic factors in problem parameters; formal statements have all dependences explicitly stated.

For $n \in \N$ we let $[n] \defeq \{i \in \N: i \le n\}$.
We reserve uppercase and lowercase boldface for matrices and vectors respectively. We use $\1_d$ and $\0_d$ to denote the all-ones and all-zeroes vectors in $\R^d$, $\id_d$ to denote the identity in $\R^d$, and $\0_{m \times n}$ is the $m \times n$ all-zeroes matrix. We denote the entrywise (Hadamard) product of equal-length vectors $\va$, $\vb$ by $\va \circ \vb$. $\Sym^{d \times d}$ and $\PSD^{d \times d}$ respectively denote the symmetric and positive semidefinite $d \times d$ matrices, where $\preceq$ is the Loewner partial order. We use $\nnz$ to denote the number of nonzero entries of a vector or matrix, and $\supp$ denotes the corresponding index set. $\ve_i$ denotes the $i^{\text{th}}$ standard basis vector, and we define the row and column selectors $\mm_{:i} \defeq \mm \ve_i$, $\mm_{i:} \defeq \mm^\top \ve_i$. For $\mm \in \R^{m \times n}$ and $(S, T) \subseteq [m] \times [n]$, $\mm_{S \times T}$ means the appropriate submatrix; our convention is to transpose before indexing, so $\mm^\top_{T \times S} = (\mm_{S \times T})^\top$.

For $1 \le p \le \infty$, $\norm{\cdot}_p$ denotes the vector $\ell_p$ norm, and for $1 \le p, q \le \infty$ and a matrix $\mm$, the associated operator norm is $\norm{\mm}_{p \to q} \defeq \max_{\vv \in \R^d: \norm{\vv}_p \le 1} \norm{\mm\vv}_{q}$. For any $\mm \in \PSD^{d \times d}$ we define $\norm{\vv}_{\mm}^2 \defeq \vv^\top \mm \vv$. We use $\normf{\cdot}$ and $\normop{\cdot}$ to denote the Frobenius and $(2 \to 2)$ operator norms. For $\mm \in \Sym^{d \times d}$ we let $\vlam(\mm)$ be its eigenvalues sorted so $\vlam_1(\mm) \ge \ldots \ge \vlam_d(\mm)$; we define $\vsig$ to similarly return the sorted singular values of its input. We define $\omega < 2.373$ \cite{AlmanDWXXZ25} so that multiplying, inverting, and eigendecomposing $d \times d$ matrices takes $O(d^\omega)$ time \cite{Strassen69, PanC99}.

We frequently use the following ``restricted'' or ``sparse'' quantities to parameterize our results: for $k \in [d]$, the top-$k$ norm of a vector $\vv \in \R^d$ and matrix $\mm \in \Sym^{d \times d}$ are
\begin{equation}\label{eq:sparse_norm_def}\begin{gathered}\norm{\vv}_{k, p} \defeq \sup_{\substack{S \subseteq [d]: |S| = k}} \norm{\vv_S}_p\text{ for all } p \ge 1,\\ \normkop{\mm}{k} \defeq \sup_{\substack{\vv \in \R^d \\ \norm{\vv}_2 \le 1,\nnz(\vv) \le k}} \Abs{\vv^\top \mm \vv} = \sup_{S \subseteq [d]: |S| = k} \normop{\mm_{S \times S}}. \end{gathered}\end{equation}
For two subsets $S, T$ with the same size of the same universe $\calU$, we use $\dham(S, T) \in \N \cup \{0 \}$ to mean the Hamming distance between $S$, $T$, which we define as half the number of differing elements.

\textbf{Probability.} For a state space $\Omega$, we let $\calP(\Omega)$ denote all probability measures over $\Omega$. For an event $\calE \subseteq \Omega$, we let $\ind(\calE)$ denote the corresponding $0$-$1$ indicator random variable, and $\mu(\calE)$ denote the probability of the event.
 We let $\E[\cdot]$ and $\Var[\cdot]$ denote the expectation and variance. For jointly distributed scalar random variables $X, Y$, $\sCov(X, Y) \defeq \E[XY] - \E[X]\E[Y]$ denotes their covariance; if the random variables are instead vector-valued, $\Cov(\vX, \vY)$ is a matrix of appropriate dimension. We abbreviate $\Cov(\vX) \defeq \Cov(\vX, \vX)$. We frequently use the following distances between $\mu, \nu \in \calP(\Omega)$, where $\dd \omega$ denotes a counting measure if $\Omega$ is discrete:
\begin{align*}
\TV{\mu, \nu} &\defeq \half \int_\Omega \Abs{\mu(\omega) - \nu(\omega)} \dd \omega = \sup_{\calE \subseteq \Omega} \mu(\calE) - \nu(\calE),\\
\CS{\mu}{\nu} &\defeq \int_\Omega \Par{\frac{\mu(\omega)}{\nu(\omega)} - 1}^2\nu(\omega) \dd \omega.
\end{align*}

We denote the set of \emph{couplings} of $\mu \in \calP(\Omega)$, $\nu \in \calP(\Omega')$ (joint measures on $\Omega \times \Omega'$ whose marginals agree with $\mu, \nu$), by $\Gamma(\mu, \nu)$. It is standard that when $\Omega = \Omega'$, an alternative definition of the total variation distance is $\TV{\mu, \nu} = \inf_{\gamma \in \Gamma(\mu, \nu)}\Pr_{(\omega, \omega') \sim \gamma}[\omega \neq \omega']$ (Proposition 4.7, \cite{LevinPW09}).

We denote the multivariate normal distribution with specified mean and covariance by $\Nor(\vmu, \msig)$. We let $\Bern(p)$ be the distribution on $\{0, 1\}$ with $\E_{X \sim \Bern(p)}[X] = p$. We use $\bigotimes_{i \in [n]} \pi_i$ to denote a product measure with specified marginals, and we use $\delta_\omega$ to mean a Dirac measure at $\omega$. When $Z \in \R_{\ge 0}^{\calI}$ are indexed by a set $\calI$, we let $\MN(Z)$ denote a draw $i \in \calI$ where $\Law(i) \propto Z$.

\subsection{Markov chains}\label{ssec:mcmc}

Let $\calT = \{\calT_{\omega}\}_{\omega \in \Omega}$ be a set of transition distributions for a Markov chain on a state space $\Omega$. For an arbitrary measure $\mu \in \calP(\Omega)$ we let $\calT\mu$ denote the marginal law of $\omega'$ where $\omega \sim \mu$ and $\omega' \sim \calT_\omega$.
We say that $\pi$ is a stationary measure for $\calT$ if $\calT \pi = \pi$, i.e.,
\[\int \calT_{\omega'}(\omega) \pi(\omega') \dd \omega' = \pi(\omega), \text{ for all } \omega \in \Omega.\]
We call a Markov chain $\calT$ reversible if it has stationary measure $\pi$, and
\[\pi(\omega) \calT_\omega(\omega') = \pi(\omega')\calT_{\omega'}(\omega), \text{ for all } (\omega, \omega') \in \Omega \times \Omega.\]

We often consider sampling from discrete measures over the hypercube with \emph{fixed-magnetization} or \emph{bounded-magnetization}. For shorthand we always let $\calX \defeq \{\pm 1\}$, and $\calX^d_k \defeq \{\vx \in \calX^d: \sum_{i \in [d]} \vx_i = -d + 2k\}$, where the quantity $\sum_{i \in [d]} \vx_i = \1_d^\top \vx$ is the magnetization. In other words, $\calX^d_k$ is the slice of the hypercube $\calX^d$ corresponding to elements with exactly $k$ copies of $1$ and $d- k$ copies of $-1$. For the analogous \emph{bounded-magnetization} problem, we similarly define $\calX^d_{\le k} \defeq \bigcup_{j = 0}^k \calX^d_j$.

It is often helpful to associate elements of $\calX^d_{ k}$ with subsets of $[d]$ with size $k$ (the locations of the $1$s). We define $\set(\vx) \subseteq [d]$ for $\vx \in \calX^d$ and $\vvec(S) \in \calX^d$ for $S \subseteq [d]$ in the natural way.

We frequently consider the ($k$-)\emph{down-up walk} Markov chain. This Markov chain can be defined for any measure $\pi$ supported on $\calX^d_k$ where $k \in [d]$. We denote its transitions by $\TDU{\pi}$ and $k$ is inferred from the definition of $\pi$. The transition $\TDU{\pi}_{\vx}$ is defined as follows for $\vx \in \calX^d_k$, where $S \defeq \set(\vx)$.
\begin{enumerate}
    \item (Down step.) A uniformly random $T \subseteq S$ with $|T| = k - 1$ is chosen.
    \item (Up step.) A set $S' \supseteq T$ with $|S'| = k$ is sampled $\propto \pi(S')$, and we step to $\vx' = \vvec(S')$.
\end{enumerate}

It is standard that $\TDU{\pi}$ is reversible with stationary measure $\pi$ (Definitions 6 and 7, and Corollary 11, \cite{KaufmanO20}). More formal pseudocode is provided in Algorithm~\ref{alg:DU_walk}.

\textbf{Spectral theory.} Let $\calT$ be a reversible Markov chain, and have stationary measure $\pi \in \calP(\Omega)$. We define the associated \emph{Dirichlet form} by its action on two functions $f, g: \Omega \to \R$:\footnote{For brevity we will omit discussion of integrability issues throughout, but always restrict the functions under consideration to a class where the integration makes sense.}
\[\calE_{\calT}(f, g) \defeq \int f(\omega) g(\omega)\pi(\omega) \dd \omega - \iint f(\omega) g(\omega') \pi(\omega) \calT_\omega(\omega') \dd\omega \dd \omega', \]
and note that the following identity holds:
\begin{equation}\label{eq:variance_form_dir}\calE_{\calT}(f, f) = \half \iint (f(\omega) - f(\omega'))^2 \pi(\omega)\calT_{\omega}(\omega') \dd \omega \dd \omega'.\end{equation}
For $\lam \in (0, 2]$, we say that $\calT$ satisfies a $\lam$-\emph{Poincar\'e inequality} if, for all $f: \Omega \to \R$,
\[\Var_\pi[f] \le \frac 1 \lam \calE_{\calT}(f, f).\]
Bounding the Poincar\'e constant $\lam$ results in an estimate of the mixing time of the Markov chain through the comparison inequality $\CS{\mu}{\pi} \ge 4\TV{\mu, \pi}^2$, and the following fact.
\begin{lemma}[Chapters 12 and 13, \cite{LevinPW09}]\label{lem:variance_decay}
Let $\calT$ be a reversible Markov chain with stationary measure $\pi$, and let $\lazy(\calT)$ denote the Markov chain that in each step transitions according to $\calT$ with probability $\half$, and otherwise does not transition. Then if $\calT$ satisfies a $\lam$-Poincar\'e inequality, if we let $\pi_t$ be the law of an iterate taking $t \in \N$ steps of $\lazy(\calT)$ starting from $\pi_0$, we have
\[\chi^2(\pi_t \| \pi) \le \Par{1 - \frac \lam 2}^t \chi^2(\pi_0 \| \pi).\]
\end{lemma}

\textbf{Path coupling.} In Section~\ref{sec:dobrushin}, we develop a new tool for bounding mixing times on fixed-magnetization or bounded-magnetization measures. This tool is based on \emph{path coupling}, which we introduce here.

\begin{lemma}[Path coupling]
    \label{lem:mixing_time_path coupling}
Let $\calT$ be a Markov chain with stationary measure $\pi \in \calP(\Omega)$.
Let $\met:\Omega\times \Omega\to \mathbb{N} \cup \{0\}$ be an integer-valued metric.
Assume that for any $\omega ,\omega'\in \Omega$ with $\met(\omega,\omega')=k$, there exists a path
$\omega=\omega_0,\omega_1,\ldots,\omega_k=\omega'$ such that for every $i\in[k]$,
\begin{equation}\label{eq:path_conds}
\met(\omega_{i-1},\omega_i)=1
\quad\text{and}\quad
\calT_{\omega_{i-1}}(\omega_i)>0.
\end{equation}
Assume furthermore that there exists $\alpha \in (0, 1)$ such that for any $(\omega,\omega')\in \Omega \times \Omega$ with $\met(\omega,\omega')=1$, there exists $\gamma \in \Gamma(\calT_\omega,\calT_{\omega'})$ with
$
\mathbb{E}_{(\psi, \psi') \sim \gamma}[\met(\psi, \psi')]\le 1-\alpha$.
Then for any $\pi_0 \in \calP(\Omega)$ and $\eps \in (0, \half)$,
\[\TV{\calT^T \pi_0, \pi} \le \eps, \text{ for } T \ge \frac 1 \alpha\Par{\log\frac{\mathrm{diam}(\Omega)}{\eps}},\text{ where } \mathrm{diam}(\Omega):=\sup_{\omega,\omega'\in\Omega \times \Omega} \met(\omega,\omega').\]
\end{lemma}
\begin{proof}
Fix any pair of initial states $\omega,\omega'\in\Omega$ and let $k:=\met(\omega,\omega')$.
Choose a path $\omega=\omega_0,\omega_1,\ldots,\omega_k=\omega'$ meeting the conditions \eqref{eq:path_conds}. For all $i \in [k]$ let $\gamma_i \in \Gamma(\calT_{\omega_{i - 1}}, \calT_{\omega_i})$ satisfy
\[\E_{(\psi, \psi') \sim \gamma_i}\Brack{\met(\psi, \psi')} \le 1 - \alpha.\]
From these couplings we can define a joint measure $\mu$ over the product space of $\psi_i \sim \calT_{\omega_i}$ for all $0 \le i \le k$, such that each marginal on adjacent pairs $(\omega_{i - 1}, \omega_i)$ agrees with $\gamma_i$. This construction is standard and follows from the ``gluing lemma'' on couplings: see e.g., Lemma 14.3, \cite{LevinPW09}.

Now draw $(\psi_0, \ldots, \psi_k) \sim \mu$. We have
\begin{align*}
\E_\mu\Brack{\met(\psi_0, \psi_k)} \le \E_\mu\Brack{\sum_{i \in [k]} \met(\psi_{i - 1}, \psi_i)} = \sum_{i \in [k]} \E_{\gamma_i}\Brack{\met(\psi_{i - 1}, \psi_i)} \le (1 - \alpha)\met(\omega, \omega').
\end{align*}
This gives a one-step coupling showing a contraction in $\met$. Thus, after $T \ge \frac 1 \alpha \log \frac{\textrm{diam}(\Omega)}{\eps}$ steps, if we independently draw $(\omega_0, \omega'_0) \sim \pi_0 \times \pi$, and then iterate the above construction to produce coupled iterates $(\omega_t, \omega'_t)$ for all $t \in [T]$ such that $\omega_t \sim \calT_{\omega_{t - 1}}$ and $\omega'_t \sim \calT_{\omega'_{t - 1}}$, we have
\[\E\Brack{\met(\omega_T, \omega'_T)} \le (1 - \alpha)^T \E\Brack{\met(\omega_0, \omega'_0)} \le (1 - \alpha)^T \textrm{diam}(\Omega) \le \eps.\]

Because $\met$ takes values in $\mathbb{N} \cup \{0\}$, and we have $\met(\omega,\omega')\ge 1$ whenever $\omega\neq \omega'$,
\[
\ind(\omega_T\neq \omega_T') \le\ \met(\omega_T,\omega_T').
\]
The conclusion follows from the coupling definition of total variation, as $\omega_T \sim \calT^T \pi_0$, $\omega'_T \sim \calT^T \pi$.
\end{proof}

\subsection{Statistical models}\label{ssec:models}

We describe the statistical models that induce the main structured distributions we consider.

\textbf{Ising model.} An \emph{Ising model} is specified by an \emph{interaction matrix} $\mj \in \Sym^{d \times d}$ and optionally, an \emph{external field} $\vh \in \R^d$. It induces a \emph{Gibbs measure} $\pi$ over $\calX^d$, with an additional parameter $\beta > 0$ governing the \emph{inverse temperature}:
\begin{equation}\label{eq:ising_def}\pi(\vx) \propto \exp\Par{\beta\Par{\half \vx^\top \mj \vx + \vh^\top \vx}} \cdot \ind(\vx \in \calX^d).\end{equation}
We often refer to \emph{fixed-magnetization} or \emph{bounded-magnetization} Ising models, where the magnetization level $k \in [d]$ is clear from context. Under these models, our goal is to sample from the Gibbs measure $\pi$ in \eqref{eq:ising_def}, further conditioned on $\vx \in \calX^d_k$ or $\vx \in \calX^d_{\le k}$ respectively.

The literature on Ising models considers a variety of statistical models for the interaction matrix. We summarize a few standard parameterizations here.

\begin{model}[Sherrington-Kirkpatrick model, \cite{SherringtonK75}]\label{model:sk}
In the \emph{Sherrington-Kirkpatrick (SK) model}, $\mj$ is drawn from the \emph{Gaussian orthogonal ensemble} $\textup{GOE}(d)$: $\mj \in \Sym^{d \times d}$ has $\mj_{ij} \simiid \Nor(0, \frac 1 d)$ for all $(i, j) \in [d] \times [d]$ with $i < j$, and $\mj_{ii} \simiid \Nor(0, \frac 2 d)$ for all $i \in [d]$.\footnote{Different sources parameterize the diagonal of $\mj$ differently in the SK model, but measures on $\calX^d$ are invariant to the diagonal of $\mj$, as $\vx_i^2 = 1$ for all $i \in [d]$, $\vx_i \in \calX$, so $\sum_{i \in [d]} \mj_{ii}\vx_i^2$ is constant.}
\end{model}

\begin{model}[Gaussian Hopfield model, \cite{BovierEN99, BarraG08, HamzeRPBK20}]\label{model:hopfield}
In the \emph{Gaussian Hopfield model},\footnote{Hopfield models \cite{Hopfield82} describe Ising models where $\mj$ is a Gram matrix, and have been enormously influential in machine learning \cite{AmitGS85, AmitGS87}. We follow this naming convention for the case when $\mj$ is Wishart.} $\mj$ is drawn from the \emph{Wishart ensemble} with $n$ degrees of freedom: $\mj \in \Sym^{d \times d}$ has $\mj = \frac 1 n \mg^\top \mg$, where $\mg \in \R^{n \times d}$ has $\mg_{ij} \simiid \Nor(0, 1)$ for all $(i, j) \in [n] \times [d]$.
\end{model}

Our algorithms' guarantees will depend on an appropriate norm of $\mj$ (and sometimes, $\vh$). Here we present some standard estimates for the random matrix ensembles in Models~\ref{model:sk},~\ref{model:hopfield}.

\begin{fact}[Section 2.5, \cite{Vershynin18} and Theorem 2.3.5, \cite{AndersonGZ10}]\label{fact:sk_facts}
For any $\delta \in (0, \half)$, the following hold under Model~\ref{model:sk} for an appropriate constant $C > 0$, simultaneously with probability $\ge 1 - \delta$.
\begin{enumerate}
    \item $\max_{(i,j) \in [d] \times [d]} |\mj_{ij}| \le C \sqrt{\log(d/\delta)/d}$. \label{item:sk_entrywise}
    \item $\normsop{\mj} \le 2 + C \sqrt{\log(1/\delta)/d}$. \label{item:sk_opnorm}
\end{enumerate}
\end{fact}

\begin{fact}[Lemma 6.26, \cite{Wainwright19} and Exercise 4.7.3, \cite{Vershynin18}]\label{fact:ghop_facts}
For any $\delta, \eps \in (0, \half)$, the following hold under Model~\ref{model:hopfield} for an appropriate constant $C > 0$, simultaneously with probability $\ge 1 - \delta$.
\begin{enumerate}
    \item If $n \ge C\log(\frac d \delta) \cdot \frac 1 {\eps^2}$, $\max_{(i, j) \in [d] \times [d]} |\mj_{ij} - \ind(i = j)| \le \eps$.
    \item If $n \ge C (k\log(\frac {ed} k) + \log (\frac 1 \delta)) \cdot \frac 1 {\eps^2}$ for $k \in [d]$, $\normkop{\mj}{k} \le 1 + \eps$.
\end{enumerate}
\end{fact}

\begin{remark}\label{rem:subg}
We focus on Gaussian $\mg$ in Model~\ref{model:hopfield} for simplicity, although Fact~\ref{fact:ghop_facts} holds for any $\mj = \frac 1 n \mg^\top \mg$ where the entries of $\mg$ are drawn i.i.d.\ from a $1$-sub-Gaussian distribution (see Section 2.5, \cite{Vershynin18}). Our analyses only rely on the properties of Model~\ref{model:hopfield} in Fact~\ref{fact:ghop_facts}, so they apply to any sub-Gaussian ensemble as well. This captures other common Hopfield model instances, e.g., Rademacher $\mg$ as often considered in the associative memory literature \cite{Little74, PasturF77, Hopfield82}.
\end{remark}

\textbf{Bayesian sparse linear regression.} Our second main application considers Bayesian sparse linear regression, i.e., sampling from the posterior distribution of a sparse linear model. We focus on a canonical parameterization of the problem, induced by the \emph{spike-and-slab prior} \cite{MitchellB88, GeorgeM93} (see also \cite{Chipman96, Geweke96}) and Gaussian measurement noise, a standard formulation recently studied by the sampling algorithms community \cite{KumarSTZ25, MontanariW26}.

\begin{model}[Spike-and-slab posterior sampling]\label{model:sas_basic}
Let $\vq \in (0, 1)^d$ and $\sigma > 0$ be known, and let $\bark \defeq \onenorm{\vq}$. Let $\mx \in \R^{n \times d}$ have entries $\simiid \calN(0, \frac 1 n)$, and suppose that we observe $(\mx, \vy)$ where 
\begin{equation}\label{eq:slr}\vths \sim \pi \defeq \bigotimes_{i \in [d]} \Par{(1 - \vq_i) \delta_0 + \vq_i \Nor(0, 1)},\quad \vxi \sim \Nor(\vzero_n, \sigma^2 \id_n),\quad\vy = \mx \vths + \vxi,\end{equation}
and $\vth^\star$ and $\vxi$ are independent. Our goal is to sample from the posterior $\pi(\cdot \mid \mx, \vy)$.
\end{model}

When designing algorithms for Model~\ref{model:sas_basic}, there is a reparameterization in terms of $\pisupp$ the distribution of $\supp(\vths) \mid \mx, \vy$. Indeed, the main algorithmic challenge is sampling from $\pisupp$. 

\begin{fact}[Lemma 8, \cite{KumarSTZ25}]\label{fact:pisupp}
$\vths \sim \pi(\cdot \mid \mx, \vy)$ in Model~\ref{model:sas_basic} can be equivalently generated as follows.
\begin{enumerate}
    \item First, $S \subseteq [d]$ is sampled from $\pisupp$ where     \begin{equation}\label{eq:pisupp_def}\begin{gathered}\pisupp(S) \propto \Par{\prod_{i \in S} \frac{\vq_i}{1 - \vq_i}}\exp\Par{\half\norm{\vb_S}_{\ma_S^{-1}}^2} \frac 1 {\sqrt{\det \ma_S}}, \\
\ma_S \in \R^{S \times S} \defeq \frac 1 {\sig^2}\Brack{\mxtx}_{S \times S} + \id_S,\quad  \vb_S \in \R^S \defeq \frac 1 {\sig^2}\mx_{S:}^\top \vy.\end{gathered}\end{equation}
    \item Second, $\vths \mid S, \mx, \vy$ is sampled from $\Nor(\ma_S^{-1} \vb_S, \ma_S^{-1})$.
\end{enumerate}
\end{fact}

\section{Sparse Dobrushin Condition}\label{sec:dobrushin}

In this section, we give our first technical result: a simple sufficient condition for rapid mixing of the down-up walk, patterned off of the classical \emph{Dobrushin uniqueness condition} \cite{Dobruschin68, Wu06}.

The applications of our sparse Dobrushin framework to Ising models (Theorem~\ref{thm:fast_mixing_slice_general} and Corollary~\ref{cor:ising_fm_gen}) in this section generally obtain weaker parameter tradeoffs than our framework in Section~\ref{sec:trickledown} does. Nonetheless, this section serves as a useful proof-of-concept of the temperature improvements achievable in fixed-magnetization settings. We include Theorem~\ref{thm:fast_mixing_slice_general} both due to its ease of applicability, and because the samplers resulting from it are formally incomparable to those in Section~\ref{sec:trickledown}, as it trades off a faster mixing time for a stricter requirement on relevant parameters.

To ease notation, this section works in a more abstract formulation than in Section~\ref{ssec:mcmc}, where we use the down-up walk to sample subsets $S \subseteq \calU$, from a distribution $\pi$ supported on
\begin{equation}\label{eq:Uk_def}\calU_k \defeq \Brace{S \subseteq \calU: |S| = k}.\end{equation}
Here, $\calU$ is a discrete universe of candidate elements. This straightforwardly captures the setting of the down-up walk in Section~\ref{ssec:mcmc} by equating $\calU \equiv [d]$ and $S \equiv \vx \defeq \vvec(S)$. We provide pseudocode implementing a one-step transition of the down-up walk in Algorithm~\ref{alg:DU_walk}. For brevity, we define
\begin{gather*}
\down(S) \defeq \Brace{T \in \calU_{k - 1}: T \subset S} \text{ for all } S \in \calU_k, \\ 
\up(T) \defeq \Brace{S \in \calU_k: S \supset T} \text{ for all } T \in \calU_{k - 1}.
\end{gather*}

\begin{algorithm}[ht]
\DontPrintSemicolon
    \caption{$\DU(S, \calU, \pi)$}
    \label{alg:DU_walk}
    \textbf{Input:} $S \in \calU_k$, discrete universe $\calU$, $\pi \in \calP(\calU_k)$ \;
    \textbf{Output:} Sample $S' \in \calU_k$ from $\TDU{\pi}_S$\;
    $T \simu \down(S)$ \;\label{line:downstep}
    $S' \sim \pi(\cdot \mid \cdot \in \up(T))$\;\label{line:upstep}
    \Return $S'$
\end{algorithm}

We state our general framework in Section~\ref{ssec:basic_sparse} and apply it to Ising models in Section~\ref{ssec:gen_ising}.

\subsection{Basic analysis}\label{ssec:basic_sparse}

We begin by stating our new sparse Dobrushin condition.

\begin{definition}[Sparse Dobrushin condition]
    \label{def:sparse_dobrushin}
Let $\pi \in \calP(\calU_k)$ with full support where $\calU$ is a discrete universe with $|\calU| \ge k$, and let $\alpha \in (0, 1)$. We say $\pi$ satisfies an \emph{$\alpha$-sparse Dobrushin condition} if
\[\TV{\pi_{U \| V}, \pi_{V \| U}} \le \alpha,\text{ for all } (U, V) \in \calU_{k - 1} \times \calU_{k - 1} \text{ with } \dham(U, V) = 1, \]
where for all $(U, V) \in \calU_{k - 1} \times \calU_{k - 1}$, we define $\pi_{U \| V}(i) \in \calP(\calU \setminus (U \cup V))$ by
\begin{equation}\label{eq:piuv_def}
\pi_{U \| V}(i) \defeq \frac{\pi(U \cup \{i\})}{\sum_{j \in \calU \setminus (U \cup V)} \pi(U \cup \{j\})}
\end{equation}
\end{definition}

The utility of Definition~\ref{def:sparse_dobrushin} reveals itself through the following bound.

\begin{lemma}
\label{lem:du_coupling}
Assume that $\pi \in \calP(\calU_k)$ satisfies an $\alpha$-sparse Dobrushin condition, and let $(S, T) \in \calU_{k} \times \calU_{k}$ satisfy $\dham(S, T) = 1$. Then there exists a coupling $\gamma \in \Gamma(\TDU{\pi}_S, \TDU{\pi}_T)$ satisfying
\[\E_{(S', T') \sim \gamma}\Brack{\dham(S', T')} \le 1 - \frac 1 k + \alpha.\]
\end{lemma}

\begin{proof}
Let $W \defeq S \cap T$ and assume $S = W \cup \{s\}$ and $T = W \cup \{t\}$. We construct the coupling explicitly. First, for the down step (Line~\ref{line:downstep}) we couple the transitions as follows.
\begin{enumerate}
    \item Draw $u \simu W$ and $r \simu [0, 1]$ independently.
    \item If $r < \frac{1}{k}$, we set $S^\downarrow = T^\downarrow = W$. If $r \ge \frac{1}{k}$, we set $S^\downarrow = S \setminus \{u\}$ and $T^\downarrow = T \setminus \{u\}$.
\end{enumerate}
If $S^\downarrow = T^\downarrow$, we can perfectly couple the subsequent up step (Line~\ref{line:upstep}), yielding $\dham(S', T') = 0$.

In the other case, $S^\downarrow \neq T^\downarrow$. Let $\rho$ denote the optimal coupling of $(\pi_{S^\downarrow \| T^\downarrow}, \pi_{T^\downarrow \| S^\downarrow})$ (see Definition~\ref{def:sparse_dobrushin}) inducing their TV distance. Define the marginal probabilities of selecting the disjoint elements as 
\[q_{S^\downarrow} = \frac{\pi(S^\downarrow \cup \{t\})}{\sum_{j \notin S^\downarrow} \pi(S^\downarrow \cup \{j\})},\quad q_{T^\downarrow} = \frac{\pi(T^\downarrow \cup \{s\})}{\sum_{j \notin T^\downarrow} \pi(T^\downarrow \cup \{j\})}.\]
Without loss of generality (by symmetry of the statement), assume $q_{S^\downarrow} \ge q_{T^\downarrow}$. 

For the up step (Line~\ref{line:upstep}) in the case $S^\downarrow \neq T^\downarrow$, we couple the transitions as follows.
\begin{enumerate}
    \item With probability $q_{T^\downarrow}$, set $S' = T' = S^\downarrow \cup \{t\} = T^\downarrow \cup \cbra{s}$.
    \item With probability $q_{S^\downarrow} - q_{T^\downarrow}$, set $S' = S^\downarrow \cup \{t\}$, draw $j \sim \pi_{T^\downarrow \| S^\downarrow}$, and set $T' = T^\downarrow \cup \{j\}$.
    \item With probability $1 - q_{S^\downarrow}$, draw $(j, j') \sim \rho$, and set $S' = S^\downarrow \cup \{j\}$ and $T' = T^\downarrow \cup \{j'\}$.
\end{enumerate}
We verify that this is a coupling. The first marginal sets $S' = S^\downarrow \cup \{t\}$ with probability $q_{S^\downarrow}$, and otherwise samples from the correct conditional distribution over $\calU \setminus T \cup \{u\}$. Similarly, the second marginal sets $T' = T^\downarrow \cup \{s\}$ correctly, and otherwise samples from the correct conditional distribution over $\calU \setminus S \cup \{t\}$. Also, in the third case, the probability $\dham(S', T') \neq 1$ is
\[\Pr\Brack{j \neq j'} \le \alpha.\]
Finally, we can compute the total expected Hamming distance:
\begin{align}\E_{(S', T') \sim \gamma}\Brack{\dham(S', T')} &\le \Par{1 - \frac 1 k} \Par{q_{T^\downarrow} \cdot 0 + (q_{S^\downarrow} - q_{T^\downarrow}) \cdot 1 + (1 - q_{S^\downarrow}) \cdot ((1 - \alpha) + 2\alpha)} \label{eq:expect_ham_diff}\\
&\le \Par{1 - \frac 1 k}\Par{1 + \alpha} \le 1 - \frac 1 k + \alpha.\nonumber
\end{align}
\end{proof}

It is clear that the down-up walk over $\calU_k$ satisfies the condition in \eqref{eq:path_conds} with $\met = \dham$. Thus, applying Lemma~\ref{lem:du_coupling} within the framework of Lemma~\ref{lem:mixing_time_path coupling} gives the following result.

\begin{proposition}
    \label{prop:fast_mixing_slice_general}
Let $\pi \in \calP(\calU_k)$ satisfy a $\alpha \le \tfrac{1}{2k}$-sparse Dobrushin condition, and let $\eps \in (0, \half)$. Then if $T = \Omega(k\log\tfrac{k}{\eps})$ for a sufficiently large constant, we have for any $\pi_0 \in \calP(\calU_k)$,
\[\TV{(\TDU{\pi})^T \pi_0, \pi} \le \eps.\]
\end{proposition}
Of course, the parameter $\alpha$ in Proposition~\ref{prop:fast_mixing_slice_general} can be taken to be any constant factor smaller than $\frac 1 k$.
Proposition~\ref{prop:fast_mixing_slice_general} also explains our naming choice, as its conclusion becomes stronger (as a function of the sparse Dobrushin condition parameter) as $k$ gets smaller, i.e., $\pi$ is supported on sparser sets.

\subsection{Fixed-magnetization Ising models}\label{ssec:gen_ising}

In this section, we demonstrate how to apply Proposition~\ref{prop:fast_mixing_slice_general} to fixed-magnetization Ising models:
\begin{equation}\label{eq:ising_slice}\pi(\vx) \propto \exp\Par{\beta\Par{\half \vx^\top \mj \vx + \vh^\top \vx}} \cdot \ind_{\vx \in \calX^d_k}.\end{equation}
We first require a helper tool to control the sparse Dobrushin condition parameter.

\begin{lemma}\label{lem:tv-tanh}
Let $\pi \in \calP(\Omega)$, $\mu \in \calP(\Omega)$ have $\pi \propto P$ and $\mu \propto Q$ for unnormalized densities $P, Q$. Then
\[\sup_{\omega \in \Omega} \Abs{\log \frac{P(\omega)}{Q(\omega)}} \le \Delta \implies \TV{\pi, \mu} \le \Delta.\]
\end{lemma}

\begin{proof}
First observe that
\begin{equation*}
    \sup_{\omega \in \Omega} \Abs{\log \frac{\pi(\omega)}{\mu(\omega)}} \le \sup_{\omega \in \Omega}\Abs{\log \frac{P(\omega)}{Q(\omega)}} + \Abs{\log\frac{\int_{\Omega} Q(\omega') \dd \omega'}{\int_{\Omega} P(\omega')\dd\omega'}} \le 2\sup_{\omega\in\Omega} \Abs{\log \frac{P(\omega)}{Q(\omega)}} \le 2\Delta.
\end{equation*}
Let $L(\omega):=\frac{\pi(\omega)}{\mu(\omega)}$ so $\exp(-2\Delta)\le L(\omega)\le \exp(2\Delta)$ for all $\omega \in \Omega$.
Using $\pi= L\,\mu$,
\[
\TV{\pi, \mu}
=\frac12\int_\Omega |\pi(\omega)-\mu(\omega)|\dd\omega
=\frac12 \int_\Omega \mu(\omega)\Abs{L(\omega) - 1}\dd\omega
=\frac12\,\E_\mu\bigl[|L-1|\bigr].
\]
Next we use the following convexity bound: if $\phi$ is convex on $[a,b]$ and random variable $X\in[a,b]$, then writing
$X=t a+(1-t)b$ with $t=\frac{b-X}{b-a}\in[0,1]$ gives
\[
\phi(X)\le t\phi(a)+(1-t)\phi(b).
\]
Taking expectations yields
\[
\E[\phi(X)]\le \frac{b-\E[X]}{b-a}\,\phi(a)+\frac{\E[X]-a}{b-a}\,\phi(b).
\]
We apply this bound with $X=L$, $\E_\mu[L]=1$, $\phi(u)=|u-1|$, and $[a, b] = [\exp(-2\Delta), \exp(2\Delta)]$. Since $a< 1 < b$,
$\phi(a)=1-a$ and $\phi(b)=b-1$, so
\[
\TV{\pi, \mu} = \frac{1}{2}\E_\mu[\phi(L)]\le \frac{(b-1)(1-a)}{b-a}.
\]
Now set $r:=\frac b a>1$ and write $b=ar$. Since $a\le 1\le ar$, we have $a\in[\frac 1 r,1]$. Define
\[
f_r(a):=\frac{(ar-1)(1-a)}{ar-a}=\frac{(ra-1)(1-a)}{a(r-1)}.
\]
A direct derivative computation shows that $f_r$ is maximized at $a=r^{-1/2}$, and hence
\[
\TV{\pi, \mu}\le f_r(r^{-1/2})
=\frac{\sqrt r-1}{\sqrt r+1} =\tanh\Par{\frac14\log r} \leq \Delta.
\]
\end{proof}

Lemma~\ref{lem:tv-tanh} lets us conclude that for fixed-magnetization Ising models (i.e., \eqref{eq:ising_def} restricted to $\calX^d_k$), the sparse Dobrushin condition parameter can be controlled by the largest off-diagonal entry of $\mj$.

\begin{lemma}
    \label{lem:sparse_dobrushin_ising_slice}
Let $\pi$ be a fixed-magnetization Ising model \eqref{eq:ising_slice} with 
\[\beta \max_{\substack{(i, j) \in [d] \times [d] \\ i \neq j}} \Abs{\mj_{ij}} \le \alpha.\]
Then $\pi$ satisfies a $8\alpha$-sparse Dobrushin condition.
\end{lemma}
\begin{proof}
Write $\mj \gets \beta\mj$ for simplicity in this proof, so we will prove the result when $\beta = 1$ without loss of generality.
Let $(\vvec(U), \vvec(V)) \in \calX^d_{k - 1} \times \calX^d_{k - 1}$ have $\dham(U, V) = 1$, and denote $W \defeq U \cap V$, $U = W \cup \{u\}$, and $V = W \cup \{v\}$. Note that the distributions $\pi_{U \| V}$ and $\pi_{V \| U}$ defined in \eqref{eq:piuv_def} are supported on the same set $\Omega \defeq [d] \setminus (U \cup V)$, so we are in the setting of Lemma~\ref{lem:tv-tanh}.

Let $i \in [d] \setminus (U \cup V)$, $X \defeq U \cup \{i\}$, $Y \defeq V \cup \{i\}$, and $\vx \defeq \vvec(X) = 2\vone_X - \vone_d$, $\vy \defeq \vvec(Y) = 2\vone_Y - \vone_d$. Observe that by expanding definitions and cancelling similar terms,
\begin{align*}
\log\Par{\frac{\exp(\half \vx^\top \mj \vx + \vh^\top \vx)}{\exp(\half \vy^\top \mj \vy + \vh^\top \vy)}} &= \Par{\half \vx^\top \mj \vx + \vh^\top \vx} - \Par{\half \vy^\top \mj \vy + \vh^\top \vy} \\
&= \Par{2\1_X^\top \mj \1_X + 2\Par{\vh - \mj \1_d}^\top \1_X + \half \1_d^\top \mj \1_d - \vh^\top \1_d} \\
&- \Par{2\1_Y^\top \mj \1_Y + 2\Par{\vh - \mj \1_d}^\top \1_Y + \half \1_d^\top \mj \1_d - \vh^\top \1_d} \\
&= 4\Par{\ve_i + \ve_W}^\top \mj(\ve_u - \ve_v) + 2\Par{\mj_{uu} - \mj_{vv}} + 2(\vh - \mj \1_d)^\top(\ve_u - \ve_v) \\
&= 4\ve_i^\top \mj(\ve_u - \ve_v) + 2(\mj_{uu} - \mj_{vv}) + 2\Par{\vh + \mj(2\1_W - \1_d)}^\top(\ve_u - \ve_v).
\end{align*}
Further, we have
\begin{align*}
\pi_{U \| V}(X) \propto \exp\Par{\half \vx^\top \mj \vx + \vh^\top \vx} \propto \exp\Par{\half \vx^\top \mj \vx + \vh^\top \vx - 2\mj_{uu} - 2(\vh + \mj(2\1_W - \1_d))^\top \ve_u}, \\
\pi_{V \| U}(Y) \propto \exp\Par{\half \vy^\top \mj \vy + \vh^\top \vy} \propto \exp\Par{\half \vy^\top \mj \vy + \vh^\top \vy - 2\mj_{vv} - 2(\vh + \mj(2\1_W - \1_d))^\top \ve_v}.
\end{align*}
Above we used that $W$ and $u$ are fixed in the definition of $\pi_{U \| V}(\cdot)$, so we can bring terms involving only these indices into the unnormalized density; a similar argument holds for $\pi_{V \| U}(\cdot)$. Thus, Lemma~\ref{lem:tv-tanh} applies with $(P, Q)$ set to the right-hand sides above, so its conclusion holds with
\[\Delta = \max_{\substack{(u, v, i) \in [d] \times [d] \times [d] \\ (u, v,i) \ \text{distinct}}}\Abs{4\ve_i^\top \mj(\ve_u - \ve_v)} \le 8\alpha. \]
\end{proof}

Importantly, the bound in Lemma~\ref{lem:sparse_dobrushin_ising_slice} is independent of $\vh$. Intuitively, this follows because the conditional distributions in \eqref{eq:piuv_def} exclude all elements in $U \cup V$, including the non-shared elements, which induce the only difference in the external field. By combining Lemma~\ref{lem:sparse_dobrushin_ising_slice} and Proposition~\ref{prop:fast_mixing_slice_general}, we thus obtain a mixing time bound for fixed-magnetization Ising models.

\begin{theorem}\label{thm:fast_mixing_slice_general}
Let $\delta \in (0, \half)$ and let $\pi$ be induced by a fixed-magnetization Ising model \eqref{eq:ising_slice} satisfying
\begin{equation}\label{eq:offdiag_bound}\beta \max_{\substack{(i, j) \in [d] \times [d] \\ i \neq j}} \Abs{\mj_{ij}} \le \frac 1 {16k}.\end{equation}
Then if $T = \Omega(k \log \frac k \delta)$ for a sufficiently large constant, we have for any $\pi_0 \in \calP(\calX^d_k)$,
\[\TV{(\TDU{\pi})^T \pi_0, \pi} \le \delta.\]
\end{theorem}

We conclude by briefly stating example implications of Theorem~\ref{thm:fast_mixing_slice_general} and Lemma~\ref{lem:sparse_dobrushin_ising_slice} for sampling from the Gibbs distributions induced by Models~\ref{model:sk} and~\ref{model:hopfield}.

\begin{corollary}\label{cor:ising_fm_gen}
Let $\delta \in (0, \half)$ and let $\pi$ be defined as in \eqref{eq:ising_slice}. If $\pi_0 \in \calP(\calX^d_k)$ and $T = \Omega(k \log \frac k \delta)$ for an appropriate constant, $\TVs{(\TDU{\pi})^T\pi_0, \pi} \le \delta$, under any of the following conditions.
\begin{enumerate}
\item Under Model~\ref{model:sk} with probability $\ge 1 - \delta$, if $\beta = O(\frac 1 k\sqrt{d/\log(d/\delta)})$ for an appropriate constant.\label{item:gen_ising_2}
\item Under Model~\ref{model:hopfield} with probability $\ge 1 - \delta$, if $n = \Omega(\max(1, (\beta k)^2) \log \frac d \delta)$ for an appropriate constant.\label{item:gen_ising_3}
\end{enumerate}

\end{corollary}
\begin{proof}
It suffices to use high-probability bounds on the maximum off-diagonal entry of $\mj$ (via Facts~\ref{fact:sk_facts} and~\ref{fact:ghop_facts}), combined with Theorem~\ref{thm:fast_mixing_slice_general} and Lemma~\ref{lem:sparse_dobrushin_ising_slice}.
\end{proof}

We give a brief discussion of the runtime of Algorithm~\ref{alg:DU_walk} in the Ising model setting.

\begin{remark}[Runtime of down-up walk for Ising model]\label{rem:runtime} For the Ising model,
Algorithm~\ref{alg:DU_walk} can be implemented in time $O(d)$ per iteration, after $O(d^2)$ time preprocessing. We believe this is folklore, but sketch a proof here. Our implementation maintains a set $S \subseteq [d]$ and applies Algorithm~\ref{alg:DU_walk} to it. Clearly, Line~\ref{line:downstep} is implementable in $O(k)$ time. Next,  the law of $\{i\} = S' \setminus T$ in Line~\ref{line:upstep} is
\begin{align*}&\propto \exp\Par{\frac \beta 2 \Par{2\1_{T \cup \{i\}} - \1_d}^\top \mj \Par{2\1_{T \cup \{i\}} - \1_d} + \beta \inprod{\vh}{2\1_{T \cup \{i\}} - \1_d}}\\
&\propto \exp\Par{\beta\Par{2\mj_{ii} + 4 \inprod{\mj \ve_i}{\1_T}- 2\inprod{\mj\ve_i}{\1_d} + 2\vh_i}}.
\end{align*}
Therefore, it is enough to maintain the quantities, for all $i \in [d] \setminus T$,
\[\inprod{\mj \ve_i}{\1_T} = \sum_{j \in T} \mj_{ij},\quad\mj_{ii},\quad \inprod{\mj\1_d}{\ve_i}, \quad \vh_i.\]
After $O(d^2)$ time preprocessing, we store the latter three (constant) quantities, and we can update the first in $O(d)$ time per call to Algorithm~\ref{alg:DU_walk}, since at most $2$ coordinates change in $T$. Finally, given these values the sampling on Line~\ref{line:upstep} can be performed in time $O(d)$.
\end{remark}
\section{Spectral Mixing via Trickle Down}\label{sec:trickledown}

In this section, we develop a second approach to prove mixing bounds on fixed-magnetization measures, based on the \emph{trickle down theorem} \cite{oppenheim2018local, alev2020improved} from the literature on high-dimensional expanders (we recommend Section 5 of the excellent survey \cite{gotlib2023nowhere} as an introduction to this topic). This framework, summarized abstractly in Section~\ref{ssec:prelims_trickle}, achieves tighter parameter tradeoffs in our applications to Ising models than its sparse Dobrushin counterpart in Section~\ref{sec:dobrushin}. 

In Section~\ref{ssec:log_sk}, we begin by giving a more interpretable sufficient condition for our framework, and two applications, as warmups. Specifically, we show that our trickle down framework qualitatively subsumes the sparse Dobrushin condition, and implies fast mixing for the SK model (Model~\ref{model:sk}) in the near-proportional magnetization regime, $k = O_\beta( \frac d {\log d})$. We also derive an analogous result for the Gaussian Hopfield model (Model~\ref{model:hopfield}). Finally, we conclude with our strongest result on the SK model, handling the proportional regime $k = O_\beta(d)$, in Section~\ref{ssec:d_sk}.

\subsection{Trickle down framework}\label{ssec:prelims_trickle}

We develop our framework for analyzing fast mixing on fixed-magnetization Ising models in three parts. We begin by recalling preliminaries on spectral graph theory, and a statement of the trickle down theorem of \cite{oppenheim2018local, alev2020improved}. We then state a sufficient condition \eqref{eq:rho_def} for applying the trickle down theorem, when the weights of the distribution in question are governed by a small perturbation of a product graph. Finally, we specialize this framework to Ising models.

\textbf{Trickle down.} Let $[d]$ index a finite vertex set. For an edge weight matrix $\mw \in \R^{d \times d}_{\ge 0}$, we define its associated random walk matrix $\mpp$ to be the degree-normalized $\mw$, i.e.,
\begin{equation}\label{eq:random_walk_matrix}\mpp \defeq \md^{-1} \mw,\text{ where } \md \defeq \diag{\mw\1_d}.\end{equation}
This section only considers \emph{reversible} $\mpp$, associated with \emph{symmetric} edge weight matrices $\mw$. 

We next state the \emph{trickle down theorem} \cite{oppenheim2018local, alev2020improved}, which bounds the Poincar\'e constant of the down-up walk induced by a measure in $\calP(\calX_k^d)$, in terms of the worst spectral gap among certain restrictions of the measure. It is proven by using the law of total variance, which gives recursive relationships among the spectral gaps of various restricted down-up walks. We defer additional background to \cite{oppenheim2018local, alev2020improved}, and simply state a sufficient form for our purposes here.

For a measure $\pi \in \calP(\calX^d_k)$ and a set $R \subseteq [d]$ with $|R| = k - 2$, define the \emph{link graph} of $R$ to be the weighted graph on vertices $[d] \setminus R$ with edge weight matrix $\mw$ given by
\[\mw_{ij} = \pi(R \cup \{i, j\})\text{ for all } (i, j) \in ([d] \setminus R) \times ([d] \setminus R),\; i \neq j.\]
Here we associate the set $R \cup \{i, j\}$ with an element of $\calX^d_k$ with those positive coordinates, per our convention.
We denote the associated random walk matrix by $\mpp_R$, following \eqref{eq:random_walk_matrix}.

\begin{lemma}[Theorem 2.5, \cite{oppenheim2018local} and Theorem 3.1, \cite{alev2020improved}]\label{lem:trickledown}
For $\pi \in \calP(\calX^d_k)$ with full support, if 
\[\lam_2\Par{\mpp_R} \le \frac \alpha {k - 1} \text{ for all } R \in \binom{[d]}{k - 2},\]
for some $\alpha < 1$, then $\TDU{\pi}$ satisfies a $\lam$-Poincar\'e inequality, for $\lam = \frac{1 - \alpha}k$.
\end{lemma}

\textbf{Spectral gap for rank-one perturbations.} To apply Lemma~\ref{lem:trickledown}, we require tools for bounding the spectral gaps of the link graphs induced by each $R \in \binom{[d]}{k-2}$. The next piece of our framework, Lemma~\ref{lem:main_sgt}, shows such a spectral gap for graphs where the edge weight matrix $\mw$ is induced by an appropriately-bounded perturbation $\mk$ of a rank-one matrix $\va\va^\top$.

\begin{remark}\label{rem:no_k}
An instructive warmup is when $\mk$ is the all-zeroes matrix in Lemma~\ref{lem:main_sgt}, in which case
\begin{equation}\label{eq:rank_one_warmup}\mw = \va\va^\top - \diag{\va}^2\end{equation}
is a rank-one matrix with its diagonal removed. Because $\lam_2(\va\va^\top) = 0$ and $\diag{\va}^2 \in \PSD^{d \times d}$, the min-max characterization of eigenvalues gives $\lam_2(\mw) \le 0$. The same strategy, applied to the similar matrix $\md^{-1/2} \mw \md^{-1/2}$, implies $\lam_2(\mpp) \le 0$, following the notation \eqref{eq:random_walk_matrix}. Lemma~\ref{lem:main_sgt} robustly extends this bound to the case where $\mw$ in \eqref{eq:rank_one_warmup} is perturbed by a bounded matrix $\mk$.
\end{remark}

\begin{lemma}\label{lem:main_sgt}
Let $\va \in \R^d_{>0}$, and let $\mk \in \Sym^{d \times d}$ satisfy $\mk_{ii} = 0$ for all $i \in [d]$. Define
\begin{equation}\label{eq:induced_w}\mw_{ij} = \begin{cases} \va_i\va_j\exp\Par{\mk_{ij}} & i \neq j \\ 
0 & i = j\end{cases}, \text{ for all } (i, j) \in [d] \times [d].\end{equation}
Also, let $m(\mk) \defeq \max_{(i,j) \in[d]\times[d]}|\mk_{ij}|$, let $\mx \in \Sym^{d \times d}$ have $\mx_{ij} = \exp(\mk_{ij}) - 1$ entrywise, and let
\begin{equation}\label{eq:rho_def}\rho(\mk) \defeq \max\Par{0, \sup_{\vu \in \R^d: \nnz(\vu) > 1} \frac{\vu^\top\Par{\mx - \id_d}\vu}{\norm{\vu}_1^2 - \norm{\vu}_2^2}}.\end{equation}
Then, following the notation \eqref{eq:random_walk_matrix}, $\lam_2(\mpp) \le \exp(m(\mk))\rho(\mk)$.
\end{lemma}
\begin{proof}
Let $\mn \defeq \md^{-1/2}\mw\md^{-1/2}$, so that $\lam_2(\mpp) = \lam_2(\mn)$.
We first write $\mn$ as the sum of a rank-one matrix and a correction. For $\ma \defeq \diag{\va}$,
\[\mn = \md^{-\half}\va\va^\top \md^{-\half} + \md^{-\half}\ma \Par{\mx - \id_d}\ma\md^{-\half}.\]
Because the min-max characterization of eigenvalues gives
\[\lam_2\Par{\mn} \le \lam_2\Par{\md^{-\half}\va\va^\top\md^{-\half}} + \lam_1\Par{ \md^{-\half}\ma \Par{\mx - \id_d}\ma\md^{-\half}},\]
it is enough to bound the second term above.
We next have
\begin{equation}\label{eq:num_denom}\begin{aligned}\lam_1\Par{ \md^{-\half}\ma \Par{\mx - \id_d}\ma\md^{-\half}} &= \sup_{\vy \in \R^d: \vy \neq \0_d} \frac{\vy^\top\md^{-\half}\ma(\mx - \id_d)\ma\md^{-\half}\vy}{\norm{\vy}_2^2} \\ 
&= \sup_{\vu \in \R^d: \vu \neq \0_d} \frac{\vu^\top (\mx - \id_d) \vu}{\vu^\top \ma^{-2}\md\vu}.
\end{aligned}\end{equation}
If the supremum is achieved by a $1$-sparse $\vu$, then the lemma statement holds, because $\rho(\mk) \ge 0$, and any $1$-sparse $\vu$ has a negative numerator in \eqref{eq:num_denom}, because $\mx$ has an all-zeroes diagonal.

We now bound the denominator of \eqref{eq:num_denom} when $\nnz(\vu) > 1$:
\begin{align*}
\vu^\top \ma^{-2} \md \vu &= \sum_{i \in [d]} \frac{\vu_i^2 \md_{ii}}{\va_i^2} = \sum_{i \in [d]} \Par{\frac{\vu_i^2}{\va_i} \sum_{j \in [d]: j \neq i} \va_j \exp(\mk_{ij})} \\
&\ge \exp\Par{-m(\mk)} \sum_{i \in [d]} \frac{\vu_i^2 (\norm{\va}_1 - \va_i)}{\va_i} \\
&= \exp\Par{-m(\mk)} \Par{\sum_{i \in [d]} \frac{\vu_i^2\norm{\va}_1}{\va_i} - \norm{\vu}_2^2 } \ge \exp\Par{-m(\mk)} \Par{\norm{\vu}_1^2 - \norm{\vu}_2^2}.
\end{align*}
The last line applied the Cauchy-Schwarz inequality. Plugging this into \eqref{eq:num_denom} gives the result.
\end{proof}

As a simple application of Lemma~\ref{lem:main_sgt}, we rederive a standard mixing result on the \emph{Curie-Weiss model}.

\begin{model}[Curie-Weiss model, \cite{Weiss07, Ellis12}]\label{model:cw}
In the \emph{Curie-Weiss model}, $\mj = \frac 1 d \1_d\1_d^\top$.
\end{model}


\begin{corollary}\label{cor:cw_trickle}
Let $\pi \in \calP(\calX^d_k)$ be induced by a fixed-magnetization Ising model \eqref{eq:ising_slice}, under the Curie-Weiss model (Model~\ref{model:cw}). Then for any $\beta \in \R$, $\TDU{\pi}$ satisfies a $\frac 1 k$-Poincar\'e inequality.
\end{corollary}
\begin{proof}
Recall that the Curie-Weiss interaction matrix is $\mj = \frac 1 d \1_d\1_d^\top$. Because $\1_d^\top \vx$ is constant over $\vx \in \calX^d_k$, identifying each $\vx$ with its set $S$ of positive coordinates, we have
\begin{equation}\label{eq:cw_is_product}\pi(S) \propto \exp\Par{2\beta \vh^\top \1_S}.\end{equation}
Now for the random walk matrix $\mpp_R$ associated with the link graph of $R \in \binom{[d]}{k - 2}$, we have that the associated edge weights $\mw$ follow \eqref{eq:induced_w} with $\mk$ set to the all-zeroes matrix, and
\[\va_i = \exp\Par{2\beta \vh_i} \text{ for all } i \in [d] \setminus R.\]
Therefore Lemma~\ref{lem:main_sgt} applies with $\mk = \vzero$ and gives a bound of $\alpha = 0$ for use with Lemma~\ref{lem:trickledown}.
\end{proof}

Corollary~\ref{cor:cw_trickle} rephrases the following proof: the fixed-magnetization Curie-Weiss model is a product distribution restricted to a Hamming slice \eqref{eq:cw_is_product}, which Theorem 1.1, \cite{AnariLGV19} proves a Poincar\'e inequality for. We include this example to illustrate a trivial case of our framework.

\textbf{Specialization to Ising model.} We next derive a generic application of the framework given by Lemmas~\ref{lem:trickledown} and~\ref{lem:main_sgt} to Ising models. Consider a fixed-magnetization Ising model \eqref{eq:ising_slice}, where
\[\pi(\vx) \propto \exp\Par{\beta\Par{\half \vx^\top \mj \vx + \vh^\top \vx}} \cdot \ind_{\vx \in \calX^d_k}.\]
Ising models are invariant to changes in the diagonal of $\mj$, so without loss of generality, we explicitly assume in this section that $\mj$ has zero diagonal, i.e., $\mj_{ii} = 0$ for all $i \in [d]$.

\begin{lemma}\label{lem:ising_trickle}
For a fixed-magnetization Ising model $\pi$ \eqref{eq:ising_slice}, and following the notation \eqref{eq:rho_def}, if 
\[\rho(4\beta\mj)\exp(m(4\beta\mj)) \le \frac 1 {2(k - 1)},\]
then $\TDU{\pi}$ satisfies a $\frac 1 {2k}$-Poincar\'e inequality.
\end{lemma}
\begin{proof}
Following the notation of Lemma~\ref{lem:trickledown}, it is enough to show that for all $R \in \binom{[d]}{k - 2}$,
\[\lam_2\Par{\mpp_R} \le \frac 1 {2(k - 1)}.\]
We prove this using Lemma~\ref{lem:main_sgt}. Fix some $R$ for the remainder of this proof, and let $S \defeq [d] \setminus R$.
Let $\vr \in \calX^d_{k - 2}$ have positive coordinates with indices $R$, and consider some $\vx = \vr + 2(\ve_i + \ve_j) \in \calX^d_k$, i.e., corresponding to the set $R \cup \{i, j\}$. We have
\[\beta\Par{\half \vx^\top \mj \vx + \vh^\top \vx} = \beta\Par{\half \vr^\top \mj \vr + \vh^\top \vr} + 2\beta\Par{\vh+\mj\vr}^\top(\ve_i + \ve_j) + 4\beta\mj_{ij}. \]
The first term is a constant for all pairs $(i, j) \in S \times S$.
Therefore, up to a proportionality constant, the edge weights are given by \eqref{eq:induced_w}, where for all $(i, j) \in S \times S$,
\[\va_i = \exp\Par{2\beta\Par{\vh + \mj \vr}^\top \ve_i},\quad \mk_{ij} = 4\beta \mj_{ij}.\]
The conclusion now follows from Lemma~\ref{lem:main_sgt} and the assumption, because excluding $R$ from the coordinates can only decrease both $m(4\beta\mj)$ and $\rho(4\beta\mj)$.
\end{proof}

\subsection{Simple sufficient conditions for fast mixing}\label{ssec:log_sk}

Section~\ref{ssec:prelims_trickle} gives a generic strategy for sampling in fixed-magnetization Ising models. By combining Lemmas~\ref{lem:trickledown},~\ref{lem:main_sgt}, and~\ref{lem:ising_trickle}, our task reduces to bounding $\rho(\mk)$ and $m(\mk)$ for $\mk \gets 4\beta\mj$. The bottleneck is typically to control $\rho(\mk)$, whose definition \eqref{eq:rho_def} is somewhat opaque. 

In Lemma~\ref{lem:mixed_quadratic}, we give a more interpretable sufficient condition for applying this framework. We show that one specialization of this condition qualitatively recovers the sparse Dobrushin condition, and that another implies improvements in the same applications as considered in Corollary~\ref{cor:ising_fm_gen}.

\textbf{Mixed-norm quadratic form bound.} Our first strategy for controlling $\rho(\mk)$ decomposes $\mx_{ij} = \exp(\mk_{ij}) - 1$ into a linear term in $\mk_{ij}$, and a high-order term. The high-order contribution to the quadratic form in $\mx$ is folded into our assumption \eqref{eq:mixed_norm_bound}, which implies a bound on $\rho(\mk)$.

\begin{lemma}\label{lem:mixed_quadratic}
In the setting of Lemma~\ref{lem:ising_trickle}, define $\mk \defeq 4\beta\mj$. Then, if 
\begin{equation}\label{eq:mixed_norm_bound}\Abs{\vu^\top \mk \vu} \le \tau\norm{\vu}_1\norm{\vu}_2 + \tau^2 \norm{\vu}_1^2\text{ for all } \vu \in \R^d,\end{equation}
for some $\tau \in [0, \frac 1 4]$, we have $\rho(\mk)\exp(m(\mk)) \le 9\tau^2$.
\end{lemma}
\begin{proof}
First, note that \eqref{eq:mixed_norm_bound} implies a bound on $m(\mk)$: taking $\vu \gets \ve_i + \ve_j$ gives
\[\Abs{\mk_{ij}} \le \sqrt 2 \tau + 2\tau^2 \le 2\tau.\]
Therefore, $m(\mk) \le 2\tau$. Further, if we decompose $\mx = \mk + \mr$, then
\[0 \le \mr_{ij} \le \sup_{|x| \le 2\tau} \exp(x) - 1 - x \le 2\sqrt{e}\tau^2, \text{ for all } (i, j) \in [d] \times [d].\]
Now by the triangle inequality, and Young's inequality applied to \eqref{eq:mixed_norm_bound},
\begin{align*}\Abs{\vu^\top \mx \vu} &\le \Abs{\vu^\top \mk \vu} + \Abs{\vu^\top \mr \vu} \\
&\le \half \norm{\vu}_2^2 + \Par{1 + \half + 2\sqrt e}\tau^2\norm{\vu}_1^2 \le \half\norm{\vu}_2^2 + 5\tau^2\norm{\vu}_1^2.
\end{align*}
We thus have
\[\vu^\top \mx \vu - \norm{\vu}_2^2 \le -\half\norm{\vu}_2^2 + 5\tau^2\norm{\vu}_1^2 \le 5\tau^2\Par{\norm{\vu}_1^2 - \norm{\vu}_2^2}.\]
Therefore, $\rho(\mk) \le 5\tau^2$, and the conclusion follows from
$\sup_{\tau \in [0, \frac 1 4]} \exp\Par{2\tau} \le \frac 9 5$.
\end{proof}

\textbf{Recovering a sparse Dobrushin condition.} We next observe that the $\tau^2 \norm{\vu}_1^2$ term alone in \eqref{eq:mixed_norm_bound} already qualitatively recovers the sparse Dobrushin condition of Theorem~\ref{thm:fast_mixing_slice_general}.

\begin{lemma}\label{lem:max_entry_recover}
In the setting of Lemma~\ref{lem:ising_trickle}, if $m(\beta \mj) \le \frac 1 {72k}$, $\TDU{\pi}$ satisfies a $\frac 1 {2k}$-Poincar\'e inequality.
\end{lemma}
\begin{proof}
By combining Lemmas~\ref{lem:ising_trickle} and~\ref{lem:mixed_quadratic}, it suffices to show that \eqref{eq:mixed_norm_bound} holds with $\tau^2 = \frac 1 {18k}$. Under the assumption on $m(\beta\mj)$, we have the desired
\[\Abs{\vu^\top \mk \vu} \le m(\mk)\norm{\vu}_1^2= m(4\beta\mj)\norm{\vu}_1^2 \le \frac 1 {18k}\norm{\vu}_1^2.\]
\end{proof}

We remark that compared to Theorem~\ref{thm:fast_mixing_slice_general}, Lemma~\ref{lem:max_entry_recover} loses a constant factor in the allowable temperature range, and only implies mixing in $\chi^2$ divergence, which typically loses a $\approx k$ factor compared to analogous mixing time bounds in Hamming distance.

\textbf{Applications to Models~\ref{model:sk} and~\ref{model:hopfield}.} Our second application of Lemma~\ref{lem:mixed_quadratic} controls the $\tau$ required in \eqref{eq:mixed_norm_bound} via the largest $\normsop{\mk_{S \times S}}$, appropriately normalized by $|S|$, over all $S \subseteq [d]$. As intuition for why, if $\vu$ is the $0$-$1$ indicator vector for some $S$,
\begin{align*}\frac{|\vu^\top \mk \vu|}{\norm{\vu}_1\norm{\vu}_2} \le |S|\normop{\mk_{S \times S}} \cdot \frac 1 {|S|^{1.5}} = \frac{\normop{\mk_{S \times S}}}{\sqrt{|S|}}, \\
\frac{|\vu^\top \mk \vu|}{\norm{\vu}_1^2} \le |S|\normop{\mk_{S \times S}} \cdot \frac 1 {|S|^{2}} = \frac{\normop{\mk_{S \times S}}}{|S|}.
\end{align*}
Lemma~\ref{lem:comparison_mixednorm} uses a \emph{shelling decomposition} to make this intuition rigorous. Interestingly, the shelling decomposition is a standard strategy for passing to continuous notions of sparsity \cite{CandesRT06}, further strengthening connections between our paper's toolkit and the broader literature.

\begin{lemma}\label{lem:comparison_mixednorm}
Let $\mk \in \Sym^{d \times d}$, and suppose that
\begin{equation}\label{eq:size_reg_bound}\normkop{\mk}{s} \le \alpha \sqrt{s} + \alpha^2 s \text{ for all } s \in [d].\end{equation}
Then \eqref{eq:mixed_norm_bound} holds with $\tau \defeq 32\alpha$.
\end{lemma}
\begin{proof}
Fix a vector $\vu \neq \0_d$, and let
\[\sigma \defeq \left\lceil\frac{\norm{\vu}_1^2}{\norm{\vu}_2^2}\right\rceil\]
which can intuitively be thought of as a numerical analog of the sparsity of $\vu$. 

Sort the coordinates of $\vu$ (relabel $[d]$ by a permutation) so that $|\vu_1| \ge |\vu_2| \ge \ldots \ge |\vu_d|$. Now partition $[d]$ into consecutive blocks $B_1 = [\sigma]$, $B_2 = [2\sigma] \setminus B_1$, $\ldots$ of size at most $\sigma$ each. For each $\ell \ge 2$, monotonicity of the coordinates gives
\[\norm{\vu_{B_\ell}}_2 \le \sqrt{\sigma} \norm{\vu_{B_\ell}}_\infty \le \frac 1 {\sqrt \sigma} \norm{\vu_{B_{\ell - 1}}}_1, \]
so summing, we have
\[\sum_\ell \norm{\vu_{B_\ell}}_2 \le \norm{\vu}_2 + \frac 1 {\sqrt \sigma}\norm{\vu}_1 \le 2\norm{\vu}_2.\]
Now we decompose $\vu^\top \mk \vu$ blockwise, and control each block using $\alpha$: because each union $B_\ell \cup B_{\ell'}$ is $2\sigma$-sparse, and each blockwise contribution is only supported on this set,
\begin{equation}\label{eq:shell_quadform}\begin{aligned}\Abs{\vu^\top \mk \vu} &\le 2(\alpha \sqrt \sigma + \alpha^2 \sigma)\sum_{\ell, \ell'} \norm{\vu_{B_\ell}}_2\norm{\vu_{B_{\ell'}}}_2 \\
&\le 2(\alpha \sqrt \sigma + \alpha^2 \sigma) \Par{\sum_\ell \norm{\vu_{B_\ell}}_2}^2 \le 8(\alpha \sqrt \sigma + \alpha^2 \sigma) \norm{\vu}_2^2.
\end{aligned}\end{equation}
Finally, using $\norm{\vu}_2 \le 2\sigma^{-1/2} \norm{\vu}_1$ gives the claim.
\end{proof}

We now derive an application to the SK model (Model~\ref{model:sk}). This application only uses the $\alpha \sqrt{s}$ term in \eqref{eq:size_reg_bound}, due to the operator norm behavior of entrywise sub-Gaussian matrices.
Per our convention in this section, we use $\mj_{ii} = 0$ for all $i \in [d]$, which does not affect the Gibbs measure \eqref{eq:ising_slice}.

\begin{corollary}\label{cor:near_proportional_sk}
Let $\delta \in (0, \half)$ and let $\pi$ be induced by a fixed-magnetization Ising model \eqref{eq:ising_slice} under the SK model (Model~\ref{model:sk}). Further, assume that for an appropriate constant,
\[\beta = O\Par{\sqrt{\frac{d}{k\log \frac d \delta}}}.\]
Then with probability $\ge 1 - \delta$, $\TDU{\pi}$ satisfies a $\frac 1 {2k}$-Poincar\'e inequality.
\end{corollary}
\begin{proof}
We first claim that for Model~\ref{model:sk} with zero diagonal,
\begin{equation}\label{eq:alphabound}\max_{s \in [d]} \frac{\normkop{\mj}{s}}{\sqrt{s}} \le C\sqrt{\frac{\log \frac d \delta}{d}},\end{equation}
with probability $\ge 1 - \delta$, for a universal constant $C$.
To see this, fix $s \in [d]$. With probability $\ge 1 - \frac \delta {d^{s + 1}}$,  Corollary 3.9, \cite{BandeiraVH16} shows that for a fixed $S \in \binom{[d]}{s}$,
\begin{equation}\label{eq:set_size_opnorm}\normop{\mj_{S \times S}} \le C\sqrt{\frac{s\log \frac d \delta}{d}}.\end{equation}
Now a union bound over the $\le d^{s}$ possible $S$, and the $d$ possible $s \in [d]$, shows \eqref{eq:alphabound}. 

Finally, for $\mk = 4\beta\mj$, the assumed range on $\beta$ implies that, following the notation \eqref{eq:size_reg_bound}, $32\alpha \le (18k)^{-1/2}$. Also, Lemma~\ref{lem:comparison_mixednorm} implies that \eqref{eq:mixed_norm_bound} holds with $\tau = 32\alpha$. Combining gives $9\tau^2 \le \frac 1 {2k}$ in Lemma~\ref{lem:mixed_quadratic}, and then the claim follows from Lemma~\ref{lem:ising_trickle}.
\end{proof}

Corollary~\ref{cor:near_proportional_sk} directly improves the allowable temperature range in Corollary~\ref{cor:ising_fm_gen}'s SK model specialization by a $\sqrt{k}$ factor. Unfortunately, it does not permit taking arbitrary $\beta = O(1)$ unless $k$ is sufficiently sublinear in $d$, i.e., smaller than $\frac d {\log d}$. This is an inherent artifact of using the bound \eqref{eq:alphabound}, because the maximum of $\Omega(d)$ Gaussians grows with $\sqrt{\log d}$, causing an obstruction at $|S| = 2$. However, it is not inherent to the SK model, and in Section~\ref{ssec:d_sk}, we show how to further shave this extraneous logarithmic factor, by more directly controlling $\rho(4\beta\mj)$ in Lemma~\ref{lem:ising_trickle}.

We conclude the section with a similar improvement upon Corollary~\ref{cor:ising_fm_gen}'s Gaussian Hopfield model specialization. Crucially, to be compatible with the two-regime sub-exponential concentration of the operator norms of Wishart matrices, our proof uses both of the terms in \eqref{eq:size_reg_bound}.

\begin{corollary}\label{cor:trickle_ghop}
Let $\delta \in (0, \half)$ and let $\pi$ be induced by a fixed-magnetization Ising model \eqref{eq:ising_slice} under the Gaussian Hopfield model (Model~\ref{model:hopfield}). Further, assume that for an appropriate constant,
\[n = \Omega\Par{\max(\beta, \beta^2)k\log \frac d \delta}.\]
Then with probability $\ge 1 - \delta$, $\TDU{\pi}$ satisfies a $\frac 1 {2k}$-Poincar\'e inequality.
\end{corollary}
\begin{proof}
We claim that for Model~\ref{model:hopfield} with zero diagonal, for each fixed $S \subseteq [d]$ with $|S| = s$,
\begin{equation}\label{eq:sbound} \normop{\mj_{S \times S}} \le C\Par{\sqrt{\frac{s \log \frac d \delta}{n}} + \frac{s \log \frac d \delta}{n}},\end{equation}
with probability $\ge 1 - \frac \delta {d^{s + 1}}$, for a universal constant $C$. At this point, the proof follows identically to that of Corollary~\ref{cor:near_proportional_sk}, using the tighter bound in \eqref{eq:size_reg_bound}.
To see that \eqref{eq:sbound} holds, let $\mj = \frac 1 n \mg^\top \mg - \md$ for $\mg \in \R^{n \times d}$ where the $\mg_{ij}$ are i.i.d.\ Gaussian, and $\md$ is a diagonal matrix agreeing with the diagonal of $\frac 1 n \mg^\top \mg$. Then \eqref{eq:sbound} follows because with probability $\ge 1 - \delta$,
\begin{equation}\label{eq:two_term_conc}
\begin{aligned}\normop{\frac 1 n [\mg^\top \mg]_{S \times S} - \id_{S}} = O\Par{\sqrt{\frac{s + \log \frac 1 \delta}{n}} + \frac{s + \log \frac 1 \delta}{n}},\\
\normop{\md_{S \times S} - \id_{S}} = O\Par{\sqrt{\frac{\log \frac s \delta}{n}} + \frac{\log \frac s \delta}{n}},\end{aligned}
\end{equation}
where the first bound above uses Exercise 4.7.3 of \cite{Vershynin18}, and the second uses Theorem 3.1.1 of \cite{Vershynin18} with a union bound over all $s$ of the diagonal coordinates. Now \eqref{eq:sbound} follows by substituting $\delta \gets \frac \delta {d^{s + 1}}$ above and applying the triangle inequality.
\end{proof}
\subsection{Linear magnetization in low-temperature SK models}\label{ssec:d_sk}

We conclude with a tighter analysis of the quantities required by Lemma~\ref{lem:main_sgt} for the SK model, that removes the extraneous logarithmic factor from Corollary~\ref{cor:near_proportional_sk}. The results in this section hold assuming a high-probability event under Model~\ref{model:sk}, captured in the following lemma.

\begin{lemma}\label{lem:conditions_sk}
Let $\delta \in (0, \half)$ and $C > 0$ be a sufficiently large universal constant. Then the following events simultaneously hold with probability $\ge 1 - \delta$ over Model~\ref{model:sk}.
\begin{enumerate}
    \item For all $S \subseteq [d]$ with $|S| = s \in [d]$,
    \[\normop{\mj_{S \times S}} \le C\Par{\sqrt{\frac{s\log \frac {ed} s + \log \frac d \delta}{d}}}.\]\label{item:subset_opnorm}
    \item We have
    \[\max_{(i, j) \in [d] \times [d]} |\mj_{ij}| \le C\sqrt{\frac{\log \frac d \delta} d},\quad \norm{\mj \1_d}_\infty \le C\sqrt{\log \frac d \delta}.\]\label{item:entry_column_bound}
    \item For all nonzero $\vu \in \R^d$ with $q \defeq \frac{\norm{\vu}_2^2}{\norm{\vu}_1^2}$,
    \[\Abs{\vu^\top \mj \vu} \le C\norm{\vu}_1^2\Par{\sqrt{\frac{q\log(edq)}{d}} + q\sqrt{\frac{\log \frac d\delta}{d}}}.\]\label{item:l2l1_bound}
    \item Let $\mh \defeq \mj \circ \mj - \frac 1 d (\1_d\1_d^\top - \id_d)$. Then,
    \[\normop{\mh} \le C\Par{\sqrt{\frac{\log \frac d \delta}{d}} + \frac{\log \frac d \delta} d}\]\label{item:hadamard_bound}
\end{enumerate}
\end{lemma}
\begin{proof}
We allot a $\frac \delta 3$ failure probability for Items~\ref{item:subset_opnorm},~\ref{item:entry_column_bound}, and~\ref{item:hadamard_bound}, and Item~\ref{item:l2l1_bound} will follow from Item~\ref{item:subset_opnorm}.

Item~\ref{item:subset_opnorm} follows from the  calculation in \eqref{eq:alphabound}, using the tighter estimate $\binom{d}{s} = \exp(O(s \log \frac {ed} s))$.

Both parts of Item~\ref{item:entry_column_bound} follow from standard bounds on the maximum of $\poly(d)$ i.i.d.\ Gaussians, where the variance of each entry of $\mj\1_d$ is at most $1$.

Item~\ref{item:l2l1_bound} follows from the same shelling decomposition argument as in Lemma~\ref{lem:comparison_mixednorm}. Concretely, perform the same decomposition into blocks $B_1, B_2, \ldots$ of size at most $\sigma = \lceil\frac 1 q\rceil$. Then the same argument as in \eqref{eq:shell_quadform}, combined with the estimate on $\normkop{\mj}{2\sigma}$ already derived in Item~\ref{item:subset_opnorm}, yields
\[\Abs{\vu^\top \mj \vu} \le 8C \Par{\sqrt{\frac{\log (edq)}{qd}} + \sqrt{\frac{\log \frac d \delta}{d}}}\norm{\vu}_2^2.\]
The claim then follows by adjusting the constant $C$, and substituting the definition of $q$. There is an edge case when $2\sigma \ge d$, but in this case $\normop{\mj}$ satisfies the required bound (see Fact~\ref{fact:sk_facts}).

Finally, Item~\ref{item:hadamard_bound} asks to bound the operator norm of a matrix with i.i.d.\ sub-exponential entries. Concretely, let $\me_{ij} \defeq \ve_i\ve_j^\top + \ve_j\ve_i^\top$ for each $1 \le i < j \le d$. Then $\mh = \sum_{1 \le i < j \le d} (\mj_{ij}^2 - \frac 1 d)\me_{ij}$. A straightforward calculation shows that the sub-exponential matrix Bernstein inequality (e.g., Theorem 6.2 in \cite{Tropp12} with $\sigma^2 = R = O(\frac 1 d)$) now applies, which concludes the proof.
\end{proof}

We now use Items~\ref{item:entry_column_bound},~\ref{item:l2l1_bound}, and~\ref{item:hadamard_bound} of Lemma~\ref{lem:conditions_sk} to derive estimates on the parameters $\rho(\mk)$ and $m(\mk)$ defined in Lemma~\ref{lem:main_sgt}, when $\mk = 4\beta\mj$ as derived in Lemma~\ref{lem:ising_trickle}.

\begin{lemma}\label{lem:improved_rho}
Assume the success of the events in Lemma~\ref{lem:conditions_sk}, and let $\beta \in \R^+$, $\bar{\beta} \defeq \max(\beta, 1)$, and $\mk \defeq 4\beta\mj$. Then for a sufficiently large universal constant $C' > 0$, assuming
\begin{equation}\label{eq:delta_condition}\Delta \defeq \sqrt{\frac{\log \frac d \delta}{d}} + \frac{\log \frac d \delta} d \le \frac 1 {C' \bar{\beta}^2\log(e\bar{\beta})},\end{equation}
we have
\[\rho(\mk)\le \frac{C'\bar{\beta}^2\log(e\bar{\beta})}{d},\quad m(\mk) \le C'\bar{\beta}^2 \Delta. \]
\end{lemma}
\begin{proof}
The bound on $m(\mk)$ follows from the first condition in Item~\ref{item:entry_column_bound} of Lemma~\ref{lem:conditions_sk}, and any $C' \ge 4C$. 

To bound $\rho(\mk)$, we follow the notation of Lemma~\ref{lem:main_sgt}, so $\mx_{ij} = \exp(\mk_{ij}) - 1$ entrywise. We also decompose $\mx = \mk + \mr$ as in Lemma~\ref{lem:mixed_quadratic}. For the linear term,
letting $\vu \in \R^d$ have $q \defeq \frac{\norm{\vu}_2^2}{\norm{\vu}_1^2}$,
\begin{equation}\label{eq:linear_sk_bound}
\begin{aligned}\Abs{\vu^\top \mk \vu} = 4\beta\Abs{\vu^\top \mj \vu} &\le 4\beta C\norm{\vu}_1^2\Par{\sqrt{\frac{q\log(edq)}{d}} + q\sqrt{\frac{\log \frac d\delta}{d}}} \\
&\le \norm{\vu}_1^2\Par{\frac q 4 + 4\beta Cq\Delta +  \frac{C''\bar{\beta}^2 \log(e\bar{\beta})}{d}},
\end{aligned}
\end{equation}
where the second line used the scalar inequality, for an appropriate $C''$ depending on $C$,
\[4\beta C \sqrt{s\log(es)} \le 4\bar{\beta}C\sqrt{s \log (es)} \le \frac s 4 + C''\bar{\beta}^2\log(e\bar{\beta}),\]
valid for any $s, \bar{\beta} \ge 1$.
For the residual term, we first establish the entrywise bound
\[0 \le \exp(\mk_{ij}) - \mk_{ij} - 1 = \mr_{ij} \le \mk_{ij}^2 \le 16\beta^2 \mj_{ij}^2,\]
as we have shown $m(\mk) \le C'\bar{\beta}^2 \Delta \le 1$ already. Thus, letting $\vw_i \defeq |\vu_i|$ for all $i \in [d]$,
\begin{equation}\label{eq:residual_sk_bound}
\begin{aligned}\Abs{\vu^\top \mr \vu} &\le 16\beta^2 \vw^\top \Par{\mj \circ \mj} \vw \\
&= 16\beta^2 \vw^\top \mh \vw + \frac {16\beta^2} d \vw^\top\Par{\1_d \1_d^\top - \id_d}\vw \\
&\le 16\beta^2C\Delta\norm{\vu}_2^2+ \frac{16\beta^2\norm{\vu}_1^2}{d} \end{aligned}
\end{equation}
where the second line used the definition of $\mh$ from Item~\ref{item:hadamard_bound}, and the last line applied Item~\ref{item:hadamard_bound} and our definition of $\Delta$. Finally, by combining \eqref{eq:linear_sk_bound} and \eqref{eq:residual_sk_bound},
\begin{align*}\vu^\top \mx \vu - \norm{\vu}_2^2 &\le \Par{-\frac 3 4 + 4\beta C \Delta+ 16\beta^2 C \Delta }\norm{\vu}_2^2+ \Par{\frac{C''\bar{\beta}^2 \log(e\bar{\beta}) + 16\beta^2}{d}}\norm{\vu}_1^2 \\
&\le -\half \norm{\vu}_2^2 + \frac{C' \bar{\beta}^2\log(e\bar{\beta})}{d}\norm{\vu}_1^2,
\end{align*}
for an appropriate $C'$. We now obtain the desired bound on $\rho(\mk)$ by noting $\frac{C' \bar{\beta}^2\log(e\bar{\beta})}{d} \le \half$.
\end{proof}

\begin{theorem}\label{thm:sk_trickle}
Let $\delta \in (0, \half)$, and let $\pi$ be induced by a fixed-magnetization Ising model \eqref{eq:ising_slice}, where
\[k = O\Par{\frac{d}{\bar{\beta}^2 \log(e\bar{\beta})}}\]
for a sufficiently small constant, defining $\bar{\beta} \defeq \max(1, \beta)$. Also, assume that the condition \eqref{eq:delta_condition} holds. 
With probability $\ge 1 - \delta$ over the SK model (Model~\ref{model:sk}), if
\[T = \Omega\Par{k\log \frac 1 \delta + k^2\Par{\log(d) + \beta\Par{\norm{\vh}_\infty + \sqrt{\log \frac d \delta}}}}\]
for a sufficiently large constant, we have for any $\pi_0 \in \calP(\calX^d_k)$,
\[\TV{\Par{\lazy(\TDU{\pi})}^T \pi_0, \pi} \le \delta.\]
\end{theorem}
\begin{proof}
The failure probability comes from Lemma~\ref{lem:conditions_sk}, so henceforth condition on its success. Combining Lemmas~\ref{lem:main_sgt},~\ref{lem:ising_trickle}, and~\ref{lem:improved_rho} implies that $\TDU{\pi}$ satisfies a $\frac 1 {2k}$-PI, for the assumed range on $k$. 
Lemma~\ref{lem:variance_decay} now shows the $\chi^2$ divergence of the lazy down-up walk contracts by a factor of $\Omega(\frac 1 k)$ in each iteration. 
The conclusion follows if
we can bound the initial $\chi^2$ divergence:
\[\chi^2\Par{\pi_0 \| \pi} = \E_{\pi}\Brack{\Par{\frac{\pi_0}{\pi}}^2} - 1 \le \frac{1}{\pi_{\min}^2},\text{ where } \pi_{\min} \defeq \min_{\vx \in \calX^d_k} \pi(\vx).\]
Thus it remains to control $\pi_{\min}$. Identify $\vx \in \calX^d_k$ with $S \in \binom{[d]}{k}$, so $\vx = 2\1_S - \1_d$. Then for
\[F(\vx) = 2\1_S^\top \mj \1_S- 2\1_S^\top \mj \1_d + 2\vh^\top \1_S,\]
the Ising measure is $\propto \exp(\beta F)$. On the events in Items~\ref{item:subset_opnorm} and~\ref{item:entry_column_bound} of Lemma~\ref{lem:conditions_sk},
\[\Abs{2\1_S^\top \mj \1_S } = O\Par{k(1 + \Delta)} = O(k),\quad \Abs{2\1_S^\top \mj\1_d} =O\Par{k\sqrt{\log \frac d \delta}},\]
where we used that \eqref{eq:delta_condition} gives $\Delta = O(1)$.
Hence, by bounding the range of $\beta F$ over $\calX^d_k$,
\begin{align*}\pi_{\min} &\ge \frac 1 {\binom{[d]}{k}}\exp\Par{-O\Par{k + \beta k \Par{\norm{\vh}_\infty +\sqrt{\log \frac d \delta}}}} \\
&= \exp\Par{-O\Par{k \log(d) + \beta k \Par{\norm{\vh}_\infty +\sqrt{\log \frac d \delta}}}}.\end{align*}
The conclusion follows by taking $T = \Omega(k\log \frac 1 {\pi_{\min}})$ and simplifying.
\end{proof}

We make two brief remarks on the statement of Theorem~\ref{thm:sk_trickle}. First, the condition \eqref{eq:delta_condition} is a lower bound on the failure probability $\delta$ that must apply for Theorem~\ref{thm:sk_trickle} to hold. This condition is relatively mild, and even permits taking $\delta = \exp(-\Omega(d))$ for appropriate constants. Second, the use of $\bar{\beta} = \max(\beta, 1)$ in the statement makes the $k$ upper bound uniform over all $\beta \ge 0$. In particular, if the assumptions hold at a target $\beta^\star \ge 0$, they hold simultaneously for every intermediate $\beta \in [0, \beta^\star]$, which becomes relevant in our applications of annealing in Section~\ref{sec:bm_sampling_anneal}.
\section{Annealing}\label{sec:bm_sampling_anneal}
Let $f: \calX^d \to \R$ and $\beta \ge 0$ parameterize a Gibbs measure $\pi_\beta \propto \exp(\beta f)$. In this section, we develop a general framework for sampling from bounded-magnetization restrictions over $\calX^d$, 
\begin{equation}\label{eq:gibbs_sparse}
\pi_{\le k, \beta}(\vx) \propto \exp\Par{\beta f(\vx)} \cdot \ind_{\vx \in \calX^d_{\le k}},
\end{equation}
leveraging samplers for fixed-magnetization measures for $0 \le s \le k$,
\begin{equation}\label{eq:gibbs_fixed}\pi_{s, \beta}(\vx) \propto \exp\Par{\beta f(\vx)} \cdot \ind_{\vx \in \calX^d_{s}},\end{equation}
e.g., those constructed in Sections \ref{sec:dobrushin} and \ref{sec:trickledown}, as black boxes. Our approach decomposes the task into two stages: we first construct an approximate sampler for the cardinality $s$, and then, conditioned on this size, invoke the corresponding fixed-size sampler for $\pi_{s, \beta}$ to obtain a sample $\sim \pi_{\le k, \beta}$.

In Section~\ref{ssec:est_normalizing_const}, we employ an annealing-based technique (adapted from \cite{Kolmogorov18}) to estimate the normalizing constants of the fixed-magnetization measures $\pi_{s, \beta}$. In Section~\ref{ssec:sparse_subsets_annealing}, we integrate these estimates into a unified framework for approximately sampling from bounded-magnetization measures $\pi_{\le k, \beta}$. Finally, in Section~\ref{ssec:bm_ising}, we instantiate our framework for bounded-magnetization variants of the Ising models in Sections~\ref{sec:dobrushin} and~\ref{sec:trickledown}, and derive the resulting mixing time bounds.

\subsection{Estimating normalizing constants}
\label{ssec:est_normalizing_const}

We first recall an estimation procedure for normalization constants of a discretely-supported measure $\pi \in \calP(\Omega)$, based directly on \cite{Kolmogorov18}. For some fixed $f: \Omega \to \R$, where $|\Omega| < \infty$, let
\begin{equation}\label{eq:Zdef}
Z(\beta) \defeq \sum_{\omega \in \Omega} \exp(\beta f(\omega))
\end{equation}
to be the normalizing constant 
of the tempered Gibbs distribution $\pi_{\beta} \propto \exp(\beta f)$ at inverse temperature $\beta \ge 0$. Observe that $Z(0) = |\Omega|$ is known exactly. For a target inverse temperature $\beta^\star$ and error tolerance $\eps \in (0, 1)$, our goal is to produce an estimate $\hZ$ such that
\begin{equation}\label{eq:Zgoal}(1 - \eps)Z(\beta^\star)\le \hZ \le (1 + \eps)Z(\beta^\star).\end{equation}

We now state a consequence of the estimation procedure of \cite{Kolmogorov18}.

\begin{proposition}\label{prop:adaptive_sche_range_f}
Let $f: \Omega \to [-R, R]$ for $R \ge 0$ and let $\beta^\star \ge 0$, $(\delta, \eps) \in (0, \half)^2$. There is an algorithm $\EstZ(f, \beta^\star, \delta, \eps, \calA)$, where $\calA$ is an algorithm that samples from
\[\pi_\beta \in \calP(\Omega) \propto \exp\Par{\beta f}\]
for any $\beta \in [0, \beta^\star]$. The output of $\EstZ$ satisfies \eqref{eq:Zgoal} with probability $\ge 1 - \delta$. Further, it uses 
\begin{equation}
    \label{eq:sample_complexity_est_const}
    N = O\Par{\frac{1 + \beta^\star R}{\eps^2}\log\Par{\frac 1 \delta}}
\end{equation}

calls to $\alg$.
\end{proposition}
\begin{proof}
We first obtain \eqref{eq:Zgoal} with probability $> \half$ using $O(\frac{1 + R\beta^\star}{\eps^2})$ calls to $\calA$, at which point the result follows by taking the median of $O(\log(\frac 1 \delta))$ copies via a standard Chernoff bound argument.

If $\beta^\star = 0$ or $R = 0$, the claim is immediate by outputting $|\Omega|$. Otherwise, define
\[H(\omega) \defeq \frac{2R-f(\omega)}{R},\quad t^\star \defeq R\beta^\star,\]
so that if we define $Z_H(t) \defeq \sum_{\omega \in \Omega} \exp(-tH(\omega))$ to be the corresponding normalizing constant for $-H$ at inverse temperature $t$, then
\[Z_H(t) = \exp\Par{-2t}Z\Par{\frac t R},\text{ for all } t \in [0, t^\star].\]
Thus to produce the required estimate \eqref{eq:Zgoal}, it is enough to estimate
\[Q \defeq \frac{Z_H(0)}{Z_H(t^\star)} = \exp\Par{2R\beta^\star} \frac{Z(0)}{Z(\beta^\star)}\]
to multiplicative error $1 \pm \eps$. Also, observe that $q \defeq \log Q$ satisfies $q \in [t^\star, 3t^\star]$, because the range of $H(\omega)$ is $[1, 3]$. Theorem 6 in \cite{Kolmogorov18} with $n = 3$ and $q = O(1 + R\beta^\star)$, and its suggested parameters $d = 64$, $m = O(1)$, and $r = O(\eps^{-2})$, now gives the claim. We note that Theorem 6 in \cite{Kolmogorov18} is stated with an expected query complexity, but Markov's inequality converts this into a deterministic runtime with a constant failure probability that can be folded into the estimator's failure.
\end{proof}

\subsection{Approximate bounded-magnetization sampling}
\label{ssec:sparse_subsets_annealing}

In this section, we develop a framework for approximate sampling from $\pi_{\le k, \beta^\star}$, defined in~\eqref{eq:gibbs_sparse}, assuming access to samplers for the densities $\pi_{s,\beta}$ \eqref{eq:gibbs_fixed} for all $0 \le s \le k$ and $0 \le \beta \le \beta^\star$. 

The key observation is that $\pi_{\le k, \beta^\star}$ admits the decomposition
\begin{equation}\label{eq:mixture}
\pi_{\le k, \beta^\star}
= \sum_{i=0}^k \alpha_i \, \pi_{i, \beta^\star},
\quad
\alpha_i = \frac{Z_i(\beta^\star)}{\sum_{j=0}^k Z_j(\beta^\star)},\quad Z_i(\beta) \defeq \sum_{\vx \in \calX^d_i}\exp\Par{\beta f(\vx)}.
\end{equation}
We first establish Lemma~\ref{lem:tv_bound_mixtures}, which shows that accurate estimates of the mixture weights $\alpha$ and component distributions $\pi_{i,\beta^\star}$ suffice to guarantee an accurate approximation of $\pi_{\le k, \beta^\star}$.

\begin{lemma}
\label{lem:tv_bound_mixtures}
Let $\calI$ be an index set and assume that $\pi, \pi' \in \calP(\Omega)$ admit the following decompositions: 
\begin{equation*}
    \pi = \bbE_{i \sim \rho}[\mu_i], \quad \pi' = \bbE_{j\sim \rho'}[\mu'_j], \quad \rho, \rho' \in \calP(\calI),\quad \mu_i, \mu'_i \in \calP(\Omega)\text{ for all } i \in \calI.
\end{equation*}
Then, 
\begin{equation*}
    \TV{\pi, \pi'} \le \TV{\rho, \rho'} + \sup_{i \in \calI}\TV{\mu_i, \mu'_i}
\end{equation*}
\end{lemma}

\begin{proof}
Let $\pi_m \defeq \bbE_{i\sim \rho}[\mu'_i]$.
By the triangle inequality of TV distance, 
\begin{align*}
\TV{\pi, \pi'} \le \TV{\pi, \pi_m} + \TV{\pi_m, \pi'}.
\end{align*}
For the first term, convexity of $|\cdot|$ yields
\[
\TV{\pi, \pi_m} = \frac{1}{2}\sum_{\omega \in \Omega}\Abs{\bbE_{i \sim \rho}[\mu_i(\omega) - \mu'_i(\omega)]} \le \bbE_{i\sim\rho}\Brack{\frac{1}{2}\sum_{\omega \in \Omega}\Abs{\mu_i(\omega) - \mu'_i(\omega)}} \le \sup_{i\in\calI}\TV{\mu_i, \mu'_i}.
\]
For the second term, couple $i \sim \rho$, $i' \sim \rho'$ to minimize $\Pr[i \neq i']$, and then couple the draws from $\mu'_i = \mu'_{i'}$ whenever $i = i'$. This produces different samples with probability $\le \TV{\rho, \rho'}$.
\end{proof}

We now present our sampler for bounded-magnetization measures \eqref{eq:gibbs_sparse} in Algorithm~\ref{alg:adaptive_bounded_size_sampler}, and establish correctness in Lemma~\ref{lem:annealing_sampler_gen_ada}. An obstacle to directly applying Proposition~\ref{prop:adaptive_sche_range_f} is that only approximate samplers $\hat{\pi}_{i,\beta}$ are available in place of $\pi_{i,\beta}$; this is addressed via a coupling argument.

\begin{lemma}
\label{lem:annealing_sampler_gen_ada}
Let $f: \calX^d_{\le k} \to [-R, R]$ for $R \ge 0$ and let $\beta^\star \ge 0$, $\delta\in (0, \half)$. For all $0 \le s \le k$, $0 \le \beta \le \beta^\star$, assume that algorithm $\calA_s(\beta, \delta')$ returns a sample within $\delta'$ total variation distance of $\pi_{s, \beta}$ defined in \eqref{eq:gibbs_fixed}. 
 Algorithm~\ref{alg:adaptive_bounded_size_sampler} uses $N$ calls, each to some $\calA_s(\beta, \frac \delta {6N})$ where 
 \begin{equation}\label{eq:overall_N}N = O\Par{\frac{k(1 + \beta^\star R  )\log(\frac{k}{\delta})}{\delta^2}}.\end{equation}
 Further, its output $\vx$ satisfies
\[
\TV{\Law(\vx), \pi_{\le k, \beta^\star}} \le \delta.
\]
\end{lemma}
\begin{proof}
For simplicity, in this proof we denote $\hat{\pi}_{s,\beta} \defeq \Law(\alg_s(\beta, \frac{\delta}{6N}))$. 
Fix an optimal coupling between $\hat{\pi}_{s,\beta}$ and $\pi_{s,\beta}$ for each oracle call to $\alg_s$. By assumption, each call produces an exact sample from $\pi_{s,\beta}$ with probability at least $1 - \frac{\delta}{6N}$. A union bound over the $N$ calls implies that, with probability at least $1 - \frac{\delta}{6}$, all oracle samples are exact.

Next, note that our choice of $N$ satisfies \eqref{eq:sample_complexity_est_const} with $\eps \gets \frac \delta 6$ and $\delta \gets \frac{\delta}{12k}$. Thus, all of the $k + 1 \le 2k$ estimates $\{\hZ_s\}_{s = 0}^k$ computed on Line~\ref{line:estZ} are correct with probability at least $1 - \frac \delta 6$. Altogether, by Proposition~\ref{prop:adaptive_sche_range_f}, we have that with probability at least $1 - \frac \delta 3$ that each $\hZ_s$ satisfies
\[
\frac{\hat{Z}_s}{Z_s} \in \left[1 - \frac{\delta}{6}, 1 + \frac{\delta}{6}\right].
\]
Applying Lemma~\ref{lem:tv-tanh} and the estimate $x \in [\pm \frac \delta 6] \implies \log(1 + x) \in [\pm \frac \delta 3]$ yields
\[
\TV{\Law(s), \MN(Z)} \le \frac{\delta}{3},
\]
where $s$ is the sampled index on Line~\ref{line:multinomial}, and $Z_i \defeq Z_i(\beta^\star)$ for all $0 \le i \le k$ as defined in \eqref{eq:mixture}.
Combining this with Lemma~\ref{lem:tv_bound_mixtures} and the failure probability of $\frac{\delta}{3}$ on the final sample, we obtain
\[
\TV{\Law(\vx), \pi_{\le k, \beta^\star}} \le \frac{2\delta}{3}
\]
on the above high-probability event.
Finally, the total failure probability from earlier was $\frac \delta 3$, and contributes additively in total variation. Therefore,
\[
\TV{\Law(\vx), \pi_{\le k, \beta^\star}} \le \delta.
\]
\end{proof}
\begin{algorithm}[ht]
\DontPrintSemicolon
    \caption{$\BMGS(k, \beta^\star, f, \cbra{\calA_s}_{s=0}^k, \delta)$}
    \label{alg:adaptive_bounded_size_sampler}
    \textbf{Input:} Magnetization bound $k \in [d]$, target inverse temperature $\beta^\star \ge 0$, $f: \calX^d_{\le k} \to [-R, R]$, approximate samplers $\{\calA_s\}_{s=0}^k$ such that for all $0 \le \beta \le \beta^\star$, $\delta' \in (0, \half)$, $\calA_s(\beta, \delta')$ returns a sample $\vx$ with $\TV{\Law(\vx), \pi_{s,\beta}} \le \delta'$, failure probability $\delta \in (0, \half)$\;
    \textbf{Output:} Sample $\vx$ such that $\TV{\Law(\vx), \pi_{\le k, \beta^\star}} \le \delta$ where $\pi_{\le k, \beta^\star} \propto \exp(\beta^\star f) \cdot \ind_{\cdot \in \calX^d_{\le k}}$\;
    \For{$i = 0, 1, \ldots, k$}{
        $\hZ_i \gets \EstZ(f_i, \beta^\star, \frac \delta {12k}, \frac \delta 6, \alg_i(\cdot, \frac \delta {6N}))$ where $f_i = f$ with domain ${\calX^d_i}$, and $N$ is as defined in \eqref{eq:overall_N}\label{line:estZ}
    }
    $s \sim \MN(\hZ)$\;\label{line:multinomial}
    \Return {$\vx \sim \calA_s(\beta^\star, \frac \delta 3)$}
\end{algorithm}

\subsection{Bounded-magnetization Ising models}
\label{ssec:bm_ising}
In this section, we show how to apply Lemma~\ref{lem:annealing_sampler_gen_ada} to bounded-magnetization Ising models, where
\[
f(\vx) = \frac{1}{2}\vx^\top \mj \vx + \vh^\top \vx
\]
as in~\eqref{eq:gibbs_sparse}. To apply our framework, we require an upper bound on $|f|$ over the domain $\calX^d_{\le k}$. We derive such a bound, which is slightly tightened by the observation that shifting the potential by a constant does not affect the Gibbs measure. The same strategy can be applied in any setting where 
$f$ has large magnitude but small variation over its domain. 

\begin{lemma}\label{lem:range_ising}
    We have
    \begin{equation*}
        \Abs{f(\vx) - f(-\vone_d)}
        \le
        (2k + 4d)\norm{\mj}_{2k,\textup{op}}
        +
        2\norm{\vh}_{k,1},
        \qquad
        \text{for all } \vx \in \calX^d_{\le k}.
    \end{equation*}
\end{lemma}

\begin{proof}
    Let $S \defeq \set(\vx)$ and $\vv \defeq \vone_S$. Because \( \vx - (-\vone_d) = 2\vv, \vx + (-\vone_d) = 2\vv - 2\vone_d,\) we obtain
    \begin{equation*}
        \Abs{f(\vx) - f(-\vone_d)}
        \le
        2\Abs{\vv^\top \mj \vv}
        +
        2\Abs{\Par{\vh-\mj\vone_d}^\top\vv}
        \le
        2k\norm{\mj}_{2k,\textup{op}}
        +
        2\norm{\vh}_{k,1}
        +
        2\Abs{\vone_d^\top\mj\vv},
    \end{equation*}
    where the last step exploits the fact that $\vv$ is $k$-sparse and
    $\twonorm{\vv}\le\sqrt{k}$. Finally, the conclusion follows from the bound
    \[
        \Abs{\vone_d^\top\mj\vv}
        \le
        \sum_{i\in[m]}
        \Abs{\vone_{S_i}^\top\mj\vv}
        \le
        2d\norm{\mj}_{2k,\textup{op}},
    \]
    for an arbitrary partition of $[d]$ into $k$-sparse sets
    $S_1\cup S_2\cup\cdots\cup S_m$, where \( m\le \frac{d}{k}+1\le\frac{2d}{k}\).
\end{proof}

By instantiating Lemma~\ref{lem:annealing_sampler_gen_ada} with the fixed-magnetization samplers from Sections~\ref{sec:dobrushin} and~\ref{sec:trickledown}, we derive samplers for the analogous bounded-magnetization Gibbs measures. For brevity, we only present the bounded-magnetization generalization of Theorem~\ref{thm:sk_trickle} here.

\begin{corollary}\label{cor:fast_mixing_sparse_set_dobrushin_gen}
In the setting of Theorem~\ref{thm:sk_trickle}, let $\pi$ be as in \eqref{eq:ising_slice} with restriction set $\calX^d_{\le k}$ instead of $\calX^d_k$. There is an algorithm that returns a sample $\vx$ with
$\TV{\Law(\vx), \pi} \le \delta$, with probability $\ge 1 - \delta$ over the SK model (Model~\ref{model:sk}),
in time
\[O\Par{\frac{dk\rho\log\Par{\frac k \delta}}{\delta^2} \cdot \Par{k \log \Par{\frac{k\rho}{\delta}} + \beta k^2\Par{\norm{\vh}_\infty + \log \frac d \delta}}}, \text{ where } \rho \defeq d + \beta\norm{\vh}_{k, 1}.\]
\end{corollary}
\begin{proof}
The algorithm is Algorithm~\ref{alg:adaptive_bounded_size_sampler} with $\beta^\star \gets \beta$, where we use
\[\Par{\lazy\Par{\TDU{\pi_{s, \beta}}}}^T \pi_0 \]
as $\calA_s(\beta, \frac \delta {6N})$ for all $0 \le s \le k$, for an arbitrary $\pi_0 \in \calP(\calX^d_s)$. Theorem~\ref{thm:sk_trickle} states that we need
\[T = O\Par{k\log \frac N \delta +  \beta k^2\Par{\norm{\vh}_\infty + \log \frac d \delta}},\]
merging terms for simplicity. Under the success of Item~\ref{item:subset_opnorm} in Lemma~\ref{lem:conditions_sk}, the assumed range on $k$ in Theorem~\ref{thm:sk_trickle}, and the condition \eqref{eq:delta_condition}, we have $\beta\normkop{\mj}{2k} = O(1)$. Thus,
Lemma~\ref{lem:range_ising} gives 
\[\beta^\star R = O(\rho) \implies N = O\Par{\frac{k\rho\log\Par{\frac{k}{\delta}}}{\delta^2}}\]
in our application of Lemma~\ref{lem:annealing_sampler_gen_ada}. The conclusion follows from the implementation in Remark~\ref{rem:runtime}.
\end{proof}
\section{Bayesian Sparse Linear Regression}\label{sec:slr}
In this section, we apply the local-to-global framework developed earlier to design an approximate sampler for the spike-and-slab posterior in Model~\ref{model:sas_basic}.  As in \cite{KumarSTZ25}, we first use a sparse recovery preprocessing step to remove coordinates whose inclusion is determined from the observations, up to negligible posterior mass.  We then sample directly from the preprocessed support posterior.

In Section~\ref{ssec:slr_setup}, we recall the preprocessing reduction of \cite{KumarSTZ25}.  In Section~\ref{ssec:sample_exact}, we prove rapid mixing of every fixed-size restriction of the exact support posterior and lift these samplers to the bounded-size posterior using tools from Section~\ref{sec:bm_sampling_anneal}.  Finally, Section~\ref{ssec:main_slr} combines the pieces.

\subsection{Setup}\label{ssec:slr_setup}

Throughout this section, we follow the shorthand
\[\me \defeq \mx^\top \mx - \id_d,\quad L \defeq \log \frac d \delta\]
to simplify the statement of bounds, where $\delta$ is a specified failure probability. We begin with a technical lemma used to decompose the output of the \cite{KumarSTZ25} reduction.

\begin{lemma}\label{lem:restricted_short}
Let $\mm \in \Sym^{d \times d}$ and let $\vv \in \R^d$ satisfy $|\supp(\vv)| \le r$. Then following the notation \eqref{eq:sparse_norm_def},
\[\norm{\mm \vv}_{s, 2} \le \normkop{\mm }{r + s} \norm{\vv}_2 \text{ for all } s \in [d - r].\]
\end{lemma}
\begin{proof}
Fix $T \subseteq [d]$ with $|T| \le s$, and let $S \defeq T \cup \supp(\vv)$ so $|S| \le r + s$. Then,
\[\norm{\Brack{\mm \vv}_T}_2 \le \normop{\mm_{S \times S}} \norm{\vv}_2 \le \normkop{\mm }{r + s} \norm{\vv}_2.\]
\end{proof}

We next state a variant of the preprocessing strategy used by \cite{KumarSTZ25}. To obtain our improved sample complexity, we require a somewhat more fine-grained guarantee on its output than the $\ell_\infty$ error bound used in prior works \cite{KumarSTZ25, ChenLTZ26}. We state our required property in the form of a decomposition \eqref{eq:short_flat}, which splits the preprocessing output vector $\vz$ into two terms: a \emph{short} vector (with bounded sparse $\ell_2$ norm), and a \emph{flat} vector (with bounded $\ell_\infty$ norm). Notably, this strategy is reminiscent of a similar decomposition used algorithmically by \cite{KelnerLLST23}.

\begin{proposition}[Section 3.1, Lemma 8, \cite{KumarSTZ25} and Theorem 3, \cite{ChenLTZ26}]
\label{prop:approx_post_i}
In the setting of Model~\ref{model:sas_basic}, let $\delta \in (0, 1)$. Then with probability $\ge 1 - \delta$ over the randomness in Model~\ref{model:sas_basic}, if $k \defeq 24(\bark + \log \frac 1 \delta)$, $\mx\simiid \calN(0, \frac 1 n)$ and $n = \Omega(k \log \frac{d}{\delta})$ for a large enough constant, there is a subset $\calU \subseteq [d]$ and a vector $\vz \in \R^d$ that can be computed in time $O(nd\log \frac {d}{\delta\min\{1, \sig\}})$, satisfying 
\[|\calU^c| = O(k),\quad \TV{\hpisupp, \pisupp} \le \delta,\]
where $\hpisupp$ is supported on $S \cup \calU^c$ where $S \subseteq \calU$, $\calU^c \defeq [d] \setminus \calU$, with
\begin{equation}\label{eq:nudef_general}
    \hpisupp(S \cup \calU^c)
    \propto
    \Par{\prod_{i \in S} \frac{\vq_i}{1 - \vq_i}}
    \exp\Par{\half\norm{\vz_{S \cup \calU^c}}_{\ma_{S \cup \calU^c}^{-1}}^2}
    \frac 1 {\sqrt{\det \ma_{S \cup \calU^c}}}
    \cdot \ind(|S| \le k).
\end{equation}
Moreover, $\vz$ admits a decomposition $\vz = \vz^{(0)} + \vz^{(1)}$, such that for any constant $a$, there exist constants $C_0, C_1 > 0$ (where $C_1$ depends only on $a$) with
\begin{equation}\label{eq:short_flat}\norm{\vz^{(0)}}_\infty \le C_0 \Par{\sig + \frac 1 \sig} \sqrt{L},\quad \norm{\vz^{(1)}}_{ak, 2} \le  \frac{C_1 kL}{\sig \sqrt{n}}.\end{equation}
\end{proposition}

\begin{proof}
We explain how to derive this result from \cite{KumarSTZ25, ChenLTZ26}, as it is not stated in this form. First, $\calU^c$ is set to $\supp(\vhth)$ where $\vhth$ is the estimator used by Lemma 6 of \cite{KumarSTZ25} satisfying
\begin{equation}\label{eq:inf_error}
    \norm{\vhth - \vths}_\infty = O\Par{\sig \sqrt{L}},
\end{equation}
with probability $\ge 1 - \frac \delta 8$. Also, $\norm{\vths}_\infty = O(\sqrt{L})$ under Model~\ref{model:sas_basic} with probability $\ge 1 - \frac \delta 8$, so
\begin{equation}\label{eq:hat_inf_bound}
\norm{\vhth}_\infty = O\Par{(1 + \sig)\sqrt L}.
\end{equation}
The existence of such an estimator $\vhth$ that takes inputs $(\mx, \vy)$, runs within the stated runtime, and satisfies \eqref{eq:inf_error} and $|\supp(\vhth)| = O(k)$ follows from Theorem 3, \cite{ChenLTZ26} for Gaussian ensembles. We note that the bound in \eqref{eq:inf_error} is obtained by using the tighter error bound for Gaussian observation matrices, discussed at the end of Page 15, \cite{ChenLTZ26}. 
The closeness of $\pisupp$ and $\hpisupp$ then follows from Lemma 6 in \cite{KumarSTZ25}, where the form of $\hpisupp$ comes from Fact~\ref{fact:pisupp}. 

Next, following Eq.\ (23) in \cite{KumarSTZ25}, we let
\begin{align*}
\vz &\defeq \frac 1 {\sig^2} \mx^\top\Par{\mx\vths + \vxi - \mx \vhth} - \vhth \\
&= \underbrace{\frac 1 {\sig^2}\Par{ \mx^\top \vxi + \vths - \vhth} - \vhth}_{\defeq \vz^{(0)}} + \underbrace{\frac 1 {\sig^2}\me \Par{\vths - \vhth}}_{\defeq \vz^{(1)}}.
\end{align*}
Clearly $\vz$ can be computed given knowledge of $\vhth$, because $\vy = \mx \vths + \vxi$ is given as an input. We now verify the conditions \eqref{eq:short_flat}. For $\vz^{(0)}$, the stated bound in \eqref{eq:short_flat} follows by combining Lemma 1 of \cite{KumarSTZ25}, which gives $\norms{\mx^\top \vxi}_\infty = O(\sig \sqrt{L})$ except with probability $\frac \delta 4$, with \eqref{eq:inf_error} and \eqref{eq:hat_inf_bound}. 

For $\vz^{(1)}$, we first condition on $\nnz(\vths) \le k$, which occurs with probability $\ge 1 - \frac \delta 4$ by Corollary 1, \cite{KumarSTZ25}. Thus, $\vths - \vhth$ is $bk$-sparse for a constant $b$. The bound in \eqref{eq:short_flat} then follows from Lemma~\ref{lem:restricted_short} with $\vv \gets \vths - \vhth$, and $s \gets ak$, where
\[\norm{\vths - \vhth}_2 = O(\sig\sqrt{kL}),\quad \normkop{\me}{(a + b)k} = O\Par{\sqrt{\frac{kL}{n}}}.\]
The first inequality above follows from \eqref{eq:inf_error} and our sparsity bound. For the second, following the proof of Corollary~\ref{cor:trickle_ghop} up until \eqref{eq:two_term_conc}, and adjusting the failure probability to union bound over subsets as in that proof, shows that for all $s \in [d]$, with probability $\ge 1 - \frac \delta 4$,
\begin{equation}\label{eq:kop_error}\normkop{\me}{s} = O\Par{\sqrt{\frac{sL}{n}} + \frac{sL}{n}}.\end{equation}
Finally, the claim follows from a union bound over all five random events in the proof.
\end{proof}

\begin{remark}\label{rem:slr_fixed_core}
For notational simplicity, in  Section~\ref{ssec:sample_exact} we present the argument assuming $\calC \defeq \calU^c = \emptyset$ in Proposition~\ref{prop:approx_post_i}. The general case of $\calC$ is identical up to a universal constant-factor enlargement of the sparsity $k$. Indeed, write the full-support potential in \eqref{eq:nudef_general} as
\[
    F(T)
    \defeq
    \sum_{i \in T}\log\frac{\vq_i}{1-\vq_i}
    +
    \half \vz_T^\top \ma_T^{-1}\vz_T
    -
    \half \log\det\ma_T \text{ for all } \calC \subseteq T \subseteq [d].
\]
Then the posterior on the collapsed support $\emptyset \subseteq S \subseteq \calU$ is $\propto \exp(F_\calC(S))\ind(|S| \le k)$, where $F_\calC(S) \defeq F(\calC \cup S)$. All conclusions in Section~\ref{ssec:sample_exact} go through unchanged after lifting by $\calC$ and adjusting constant factors to account for $|\calC| = O(k)$. For example, Lemma~\ref{lem:slr-schur-reduction} applies verbatim after replacing $R \gets \calC \cup R$, and since we only considered $|R| \le k$, the same sparsity bounds hold up to constant factors after this adjustment. Similarly, Lemma~\ref{lem:combine_params} considers differences between potentials $F(S)$, which goes through unchanged after replacing $F \gets F_{\calC}$ and $S \gets \calC \cup S$. 
\end{remark}

In the rest of the section, we fix the universe $\calU$ returned by Proposition~\ref{prop:approx_post_i}. We write \(\calU_{\le k}\defeq\bigcup_{s=0}^k\calU_s\), analogously to \eqref{eq:Uk_def}. Conditioned on the success of Proposition~\ref{prop:approx_post_i}, it is enough to sample from
the density $\hpisupp$ in \eqref{eq:nudef_general}. 
We write the (unnormalized) density as $\exp(F(S))$, where we define
\[F(S)
    =
    \sum_{i\in S}\log\frac{\vq_i}{1-\vq_i}
    +
    \frac12\vz_S^\top\ma_S^{-1}\vz_S
    -
    \frac12\log\det\ma_S.\]
    We now manipulate this expression to be on a more convenient scale. In particular, let
    \begin{equation}\label{eq:gamma-tau}
    \gamma\defeq\frac{\sig^2}{1+\sig^2},
    \quad
    \tau\defeq\frac{1}{1+\sig^2},
    \quad
    \vs\defeq\sqrt{\gamma\tau}\,\vz,
    \quad
    \ml_S\defeq\gamma\ma_S=\id_S+\tau\me_S\text{ for all } S \subseteq [d].
\end{equation}
Under this scaling, and defining $\vs^{(0)} \defeq \sqrt{\gamma\tau}\vz^{(0)}$ and $\vs^{(1)} \defeq \sqrt{\gamma\tau}\vz^{(1)}$, \eqref{eq:short_flat} implies
\begin{equation}\label{eq:short_flat_scaled}
\norm{\vs^{(0)}}_\infty \le C_0 \sqrt{L},\quad \norm{\vs^{(1)}}_{ak, 2} \le \frac{C_1 k L }{\sqrt n}.
\end{equation}
Moreover, again using the notation \eqref{eq:gamma-tau} and combining all linear terms, we obtain
\begin{equation}\label{eq:slr-F-normalized}
    F(S)
    =
    \vh^\top\vone_S
    +
    \frac{1}{2\tau}\vs_S^\top\ml_S^{-1}\vs_S
    -
    \frac12\log\det\ml_S,
    \text{ where }
    \vh_i \defeq \log\frac{\vq_i}{1-\vq_i}+\frac12\log\gamma.
\end{equation}
Finally, for $\beta\in[0,1]$ and $0\le s\le k$, we define
\begin{equation}\label{eq:slr-tempered}
    \nu_{\le s,\beta}(S)
    \propto
    \exp(\beta F(S))\ind(S\in\calU_{\le s}),
    \quad
    \nu_{s,\beta}(S)
    \propto
    \exp(\beta F(S))\ind(S\in\calU_s).
\end{equation}
With this notation, our target support posterior is $\nu=\nu_{\le k,1}$. 

\subsection{Trickle down for support posterior}
\label{ssec:sample_exact}
In this section, we apply the trickle down framework of Section~\ref{sec:trickledown} to sampling under Model~\ref{model:sas_basic}.
We begin by collecting several additional estimates we require of our draw from the model.

\begin{lemma}
\label{lem:slr-gaussian-profile}
Under Model~\ref{model:sas_basic}, assume $n = \Omega(kL)$ for an appropriate constant. Under the success of the event in Proposition~\ref{prop:approx_post_i}, we have the following additional guarantee, for a universal constant $C > 0$, and $\alpha \defeq \sqrt{L/n}$. Simultaneously for every
$R\in \calU_{\le k}$, $B \defeq \calU\setminus R$, and 
$\vu,\vv \in \R^d$,
\[
    \normkop{\me}{2k}\le C\alpha\sqrt{k},
    \quad
    \max_{(i,j)\in[d]\times[d]} \Abs{\me_{ij}} \le C\alpha,
    \quad
    \norm{\me_{R\times B}\vu_B}_2
    \le
    C\alpha\left(
        \sqrt{k}\twonorm{\vu}+\onenorm{\vu}
    \right),
\]
and
\begin{equation}\label{eq:mixnorm_mixvec}
    \Abs{\vu^\top\me\vv}
    \le
    C\Par{
        \alpha\left(
            \onenorm{\vu}\twonorm{\vv}
            +
            \twonorm{\vu}\onenorm{\vv}
        \right)
        +
        \alpha^2\onenorm{\vu}\onenorm{\vv}
    }.
\end{equation}
\end{lemma}
\begin{proof}
The only event we condition on in this proof is that \eqref{eq:kop_error} holds.
The first two inequalities then follow by plugging $s = 2k$ and $s = 2$ into \eqref{eq:kop_error}, and using our lower bound on $n$. 
For inequality \eqref{eq:mixnorm_mixvec}, assume $\vu, \vv \neq \0_d$, else the claim is immediate. Then let 
\[t^2 \defeq \frac{\norm{\vv}_1}{\norm{\vu}_1} \implies t\norm{\vu}_1 + \frac 1 t \norm{\vv}_1 = 2\sqrt{\norm{\vu}_1\norm{\vv}_1}.\]
Since
\[\vu^\top \me \vv = \frac 1 4\Par{t\vu + \frac 1 t \vv}^\top \me\Par{t\vu + \frac 1 t \vv} - \frac 1 4\Par{t\vu - \frac 1 t \vv}^\top \me \Par{t\vu - \frac 1 t \vv},\]
the claim follows from Lemma~\ref{lem:comparison_mixednorm} applied to the two terms, the triangle inequality, and
\begin{align*}
\Par{t\norm{\vu}_1 + \frac 1 t \norm{\vv}_1}\Par{t\norm{\vu}_2 + \frac 1 t \norm{\vv}_2} &= 2\Par{\norm{\vu}_1\norm{\vv}_2 + \norm{\vu}_2\norm{\vv}_1}, \\
\Par{t\norm{\vu}_1 + \frac 1 t \norm{\vv}_1}^2 &= 4\norm{\vu}_1\norm{\vv}_1.
\end{align*}
For the third inequality, applying \eqref{eq:mixnorm_mixvec} and the fact that $\norm{\me_{R \times B}\vu}_2 = \vv^\top \me \vu_B$ for some unit vector $\vv$ supported only on $R$, such that $\norm{\vv}_1 \le \sqrt{k}$, yields the claim upon simplifying with $\alpha\sqrt{k} = O(1)$.
\end{proof}
Conditioned on the success of the event in Lemma~\ref{lem:slr-gaussian-profile}, we next show how to apply our framework in Section~\ref{sec:trickledown} to sample from the posterior density $\nu_{\le k , 1}$ \eqref{eq:slr-tempered}. We begin with an exact characterization of the potential gain by including a set of coordinates.

\begin{lemma}
\label{lem:slr-schur-reduction}
For a fixed $R \subseteq \calU$, let $B \defeq \calU\setminus R$ and define
\begin{equation}\label{eq:rdef}
    \mr_R
    \defeq
    \me_{B\times B}
    -
    \tau\me_{B\times R}\ml_R^{-1}\me_{R\times B},
    \quad
    \vr_R
    \defeq
    \vs_B
    -
    \tau\me_{B\times R}\ml_R^{-1}\vs_R.
\end{equation}
Then, for every $T\subseteq B$, letting $\mr_T \defeq [\mr_R]_{T\times T}$ and $\vr_T \defeq [\vr_R]_T$,
\[
    F(R\cup T)-F(R)
    =
    \vh^\top\vone_T
    +
    \frac{1}{2\tau}
    \vr_T^\top
    \left(\id_T+\tau\mr_T\right)^{-1}
    \vr_T
    -
    \frac12
    \log\det\left(\id_T+\tau\mr_T\right).
\]
\end{lemma}

\begin{proof}
With the coordinates ordered as $R,T$,
\[
    \ml_{R\cup T}
    =
    \begin{pmatrix}
        \ml_R & \tau\me_{R\times T}\\
        \tau\me_{T\times R} & \ml_T
    \end{pmatrix}.
\]
Its Schur complement with respect to the $R$ block is
\[
\begin{aligned}
    \ml_T
    -
    \tau^2\me_{T\times R}\ml_R^{-1}\me_{R\times T}=
    \id_T+\tau\mr_T.
\end{aligned}
\]
Consequently,
$\det\ml_{R\cup T}
    =
    \det(\ml_R)
    \det\left(\id_T+\tau\mr_T\right)$.
The block inverse formula also gives
\[
\begin{aligned}
    \vs_{R\cup T}^\top\ml_{R\cup T}^{-1}\vs_{R\cup T}
    &=
    \vs_R^\top\ml_R^{-1}\vs_R+
    \left(
        \vs_T
        -
        \tau\me_{T\times R}\ml_R^{-1}\vs_R
    \right)^\top
    \left(\id_T+\tau\mr_T\right)^{-1}
    \left(
        \vs_T
        -
        \tau\me_{T\times R}\ml_R^{-1}\vs_R
    \right)\\
    &=
    \vs_R^\top\ml_R^{-1}\vs_R
    +
    \vr_T^\top
    \left(\id_T+\tau\mr_T\right)^{-1}
    \vr_T.
\end{aligned}
\]
Substituting both of the above displays into \eqref{eq:slr-F-normalized} now proves the claim.
\end{proof}

We now give the main technical result in this section, which provides various estimates required to apply the parameter bounds from Lemma~\ref{lem:mixed_quadratic} to an appropriate interaction matrix.

\begin{lemma}\label{lem:combine_params}
Assume the events of Proposition~\ref{prop:approx_post_i} and Lemma~\ref{lem:slr-gaussian-profile} hold. Let $R \in \calU_{\le k}$ and $B \defeq \calU \setminus R$, 
and following notation of Lemma~\ref{lem:slr-schur-reduction}, let $\mk \in \R^{B \times B}$ have zero diagonal and satisfy
\[\mk_{ij} \defeq F\Par{R \cup \{i, j\}} - F\Par{R \cup \{i\}} - F\Par{R \cup \{j\}} + F(R) \text{ for all } (i, j) \in B \times B,\; i\neq j.\]
Then the following bounds hold for a universal constant $C' > 0$, letting $\alpha \defeq \sqrt{L/n}$.
\begin{enumerate}
\item\label{item:quadform_bound} For all $\vu \in \R^B$, following the notation \eqref{eq:rdef},
\[\Abs{\sum_{\substack{(i, j) \in B \times B \\ i \neq j}} \mr_{ij}\vu_i\vu_j} \le C'\Par{\alpha^2 k \norm{\vu}_2^2 + \alpha\norm{\vu}_1\norm{\vu}_2 + \alpha^2\norm{\vu}_1^2 }.\]
\item For all $(i, j) \in B \times B$, $i \neq j$,
\[\mk_{ij} = -\mr_{ij}\vr_i\vr_j + \eps_{ij}, \text{ for } \Abs{\eps_{ij}} \le C'\Par{\alpha^2 + \alpha^4k^2}\Par{1 + \vr_i^2 + \vr_j^2}.\]\label{item:entry_perturb_bound}
\item For all $\vu \in \R^B$,
\[\norm{\vr \circ \vu}_1 \le C'\sqrt{L}\Par{\norm{\vu}_1 + \frac{kL}{\sqrt n}\norm{\vu}_2},\quad \norm{\vr \circ \vu}_2 \le C'\Par{\sqrt{L} + \frac{kL}{\sqrt n}}\norm{\vu}_2,\]
and
\[\sum_{i \in B} |\vu_i| \vr_i^2 \le C'\Par{L\norm{\vu}_1 + \frac{k^2L^2}{n}\norm{\vu}_2}.\]\label{item:reweight_u}
\end{enumerate}
\end{lemma}
\begin{proof}
We proceed with the three claims in order. First, observe that by Lemma~\ref{lem:slr-gaussian-profile} and $|R| \le k$,
\begin{equation}\label{eq:inv_op_bound}\normop{\ml_R^{-1}} \le \Par{1 - C\tau\alpha\sqrt k}^{-1} \le 2,\end{equation}
using $\tau \le 1$ and taking $n$ large enough.
Also, again by applying Lemma~\ref{lem:slr-gaussian-profile},
\begin{equation}\label{eq:entrywise}
\begin{aligned}
\norm{\mr}_{\max} \defeq \max_{(i, j) \in B \times B} \Abs{\mr_{ij}} &\le \max_{(i, j) \in B \times B}\Abs{\me_{ij}} +  \normop{\ml_R^{-1}} \norm{\me_{R \times \{i\}}}_2 \norm{\me_{R \times \{j\}}}_2 \\
&\le C\alpha + 2\Par{C\alpha (\sqrt k + 1)}^2 \le C\alpha + 8C^2\alpha^2 k \le \half.
\end{aligned}
\end{equation}
The last inequality again used our lower bound on $n$. Hence,
\begin{align*}
\Abs{\sum_{\substack{(i, j) \in B \times B \\ i \neq j}} \mr_{ij}\vu_i\vu_j} &\le \Abs{\vu^\top \mr \vu} + \norm{\mr}_{\max}\norm{\vu}_2^2 \\
&\le \Abs{\vu^\top \mr \vu} + 8C^2 \alpha^2 k \norm{\vu}_2^2 + C\alpha\norm{\vu}_1\norm{\vu}_2,
\end{align*}
and by again applying Lemma~\ref{lem:slr-gaussian-profile}, we have
\begin{align*}
\Abs{\vu^\top \mr \vu} &\le \Abs{\vu^\top \me \vu} + 2\norm{\me_{R \times B}\vu}_2^2 \\
&\le 2C\Par{\alpha\norm{\vu}_1\norm{\vu}_2 + \alpha^2\norm{\vu}_1^2} + 4C^2\Par{\alpha^2k \norm{\vu}_2^2 +\alpha^2\norm{\vu}_1^2}.
\end{align*}
Combining the above two displays proves Item~\ref{item:quadform_bound}. 

Next, for all $i \in B$, let $d_i \defeq 1 + \tau\mr_{ii}$, so that $d_i \in [1 \pm \tau\norm{\mr}_{\max}] \subseteq [\half, \frac 3 2]$. Also, let 
\[D_{ij} \defeq \det\Par{\id_{\{i, j\}} + \tau\mr_{\{i,j\}}} = d_id_j - \tau^2\mr^2_{ij}\]
so that
\[\Par{\id_{\{i, j\}} + \tau\mr_{\{i,j\}}}^{-1} = \begin{pmatrix} d_i & \tau\mr_{ij} \\ \tau\mr_{ij} & d_j \end{pmatrix}^{-1} = \frac{1}{D_{ij}}\begin{pmatrix} d_j & -\tau\mr_{ij} \\ -\tau\mr_{ij} & d_i \end{pmatrix}.\]
Then, applying Lemma~\ref{lem:slr-schur-reduction} with $T = \emptyset, \{i\}, \{j\}, \{i, j\}$, the linear terms cancel, so
\begin{equation}\label{eq:k_entry_derive}
\begin{aligned}\mk_{ij} &= \frac 1 {2\tau D_{ij}}\Par{d_j \vr_i^2 + d_i\vr_j^2 - 2\tau\mr_{ij}\vr_i\vr_j} - \frac 1 {2\tau}\Par{\frac{\vr_i^2}{d_i} + \frac{\vr_j^2}{d_j}} - \half\log\Par{D_{ij}} + \half\log\Par{d_id_j}
\\
&= \frac{\tau\mr_{ij}^2\Par{\frac{\vr_i^2}{d_i} + \frac{\vr_j^2}{d_j}} - 2\mr_{ij}\vr_i\vr_j}{2D_{ij}} - \half\log\Par{1 - \frac{\tau^2\mr_{ij}^2}{d_id_j}} \\
&= -\frac{\mr_{ij}\vr_i\vr_j}{D_{ij}} + \frac{\tau\mr_{ij}^2}{2D_{ij}}\Par
{\frac{\vr_i^2}{d_i} + \frac{\vr_j^2}{d_j}} - \half\log\Par{1 - \frac{\tau^2\mr_{ij}^2}{d_id_j}}.
\end{aligned}
\end{equation}
Using the estimates $|d_i - 1|, |d_j - 1|, |D_{ij} - 1| = O(\norm{\mr}_{\max}) = O(\alpha + \alpha^2 k)$ and Taylor expanding each error term (since \eqref{eq:entrywise} bounds $\norm{\mr}_{\max}$ by an arbitrary constant), gives for large enough $C'$,
\begin{align*}
\Abs{\mr_{ij}\vr_i\vr_j} \cdot \Abs{\frac 1 {D_{ij}} - 1} &\le \frac{C'}{6}\Par{\alpha + \alpha^2 k}^2\Par{\vr_i^2 + \vr_j^2},\\
\frac{\tau\mr_{ij}^2}{2D_{ij}}\Par
{\frac{\vr_i^2}{d_i}+ \frac{\vr_j^2}{d_j}} &\le \frac{C'}{6}\Par{\alpha + \alpha^2 k}^2\Par{\vr_i^2 + \vr_j^2},\\
\Abs{\log\Par{1 - \frac{\tau^2\mr_{ij}^2}{d_id_j}}} &\le \frac {C'} 6 \Par{\alpha + \alpha^2 k}^2.
\end{align*}
Combining these bounds within \eqref{eq:k_entry_derive}, and using $(a + b)^2 \le 2(a^2 + b^2)$, then yields Item~\ref{item:entry_perturb_bound}.

To conclude, let $\vc \defeq \me_{B \times R}\ml_R^{-1} \vs_R$ and $\vd \defeq \vs^{(1)}_B - \tau\vc$, so that following Lemma~\ref{lem:slr-schur-reduction} and $\vs = \vs^{(0)} + \vs^{(1)}$, we have $\vr_R = \vs^{(0)}_B + \vd$. Recall from \eqref{eq:short_flat_scaled} that $\vs^{(0)}_B$ has all entries bounded by $O(\sqrt{L})$, such that 
\[\norms{\vs^{(0)} \circ \vu}_1 = O(\sqrt{L})\norm{\vu}_1, \quad \norms{\vs^{(0)} \circ \vu}_2 = O(\sqrt{L})\norm{\vu}_2, \quad \sum_{i \in B} |\vu_i| (\vs_i^{(0)})^2 = O\Par{L}\norm{\vu}_1.\]
By choosing $C'$ large enough in Item~\ref{item:reweight_u}, all of the above contributions fit asymptotically within the claimed budgets. It thus suffices to prove Item~\ref{item:reweight_u} using $\vd$ to reweight $\vu$ rather than $\vr$.

Next, we bound $\norm{\vd}_{k, 2}$. Let $T \subseteq B$ index the largest $k$ coordinates of $\vd$ by magnitude.
By \eqref{eq:short_flat_scaled} and $n = \Omega(kL)$, for every $R \in \calU_{\le k}$, 
\[\norm{\vs_R}_2 \le \sqrt{k}\norm{\vs^{(0)}}_\infty + \norm{\vs^{(1)}}_{k, 2} = O(\sqrt{kL}).\]
Then, \eqref{eq:inv_op_bound} and the third bound in Lemma~\ref{lem:slr-gaussian-profile} imply, for a large enough constant $C''$,
\[\norm{\vc_T}_2 \le 2C\alpha\sqrt{k} \cdot O(\sqrt{kL}) \le kL\sqrt{\frac{C''}{n}}.\]Combining with \eqref{eq:short_flat_scaled} and $\tau \le 1$ gives $\norm{\vd_T}_2 \le kL\sqrt{\frac{C''}{n}}$ after adjusting $C''$. 

Thus, letting the $i^{\text{th}}$ largest magnitude amongst $\vd$'s coordinates be denoted $d_{(i)}$, we have shown
\begin{equation}\label{eq:d_bounds}d_{(i)}^2 \le \frac{C'' k^2 L^2}{n} \cdot \frac 1 {\min(i, k)},\text{ for all } i \in |B|.\end{equation}
Now, similarly denoting the $i^{\text{th}}$ largest magnitude amongst $\vu$'s coordinates by $u_{(i)}$,
\begin{equation}\label{eq:l1_weight_bound}
\begin{aligned}\norm{\vd \circ \vu}_1 &\le \sqrt{\frac{C'' k^2 L^2}{n}} \cdot \Par{\sum_{i \in [|B|]} \frac{u_{(i)}}{\sqrt{i}} + \frac{u_{(i)}}{\sqrt{k}}} \\
&\le C'\sqrt{L}\norm{\vu}_1 + \sqrt{\frac{C'' k^2 L^2}{n}} \cdot \sqrt{\Par{\sum_{i \in [|B|]}\frac 1 i}\norm{\vu}_2^2} \le C'\sqrt{L}\norm{\vu}_1 + \frac{C' kL^{1.5}}{\sqrt n}\norm{\vu}_2, \end{aligned}
\end{equation}
where the second inequality was by Cauchy-Schwarz. Similarly,
\begin{align*}
\norm{\vd \circ \vu}_2^2 \le \frac{C'' k^2 L^2}{n} \cdot \Par{\sum_{i \in [|B|]} \frac{u_{(i)}^2}{i} + \frac{u_{(i)}^2}{k}} \le \frac{C'' k^2 L^2}{n}\norm{\vu}_2^2.
\end{align*}
Finally, recalling $n = \Omega(kL)$, and applying the Cauchy-Schwarz inequality as in \eqref{eq:l1_weight_bound},
\begin{align*}
\sum_{i \in B} |\vu_i| \vd_i^2 \le \sum_{i \in [|B|]} u_{(i)}d_{(i)}^2 \le \frac{C'' k^2 L^2}{n} \cdot \Par{\sum_{i \in [|B|]} \frac{u_{(i)}}{i} + \frac{u_{(i)}}{k}} \le  \frac{C' k^2 L^2}{n}\norm{\vu}_2 + C' L \norm{\vu}_1 .
\end{align*}
\end{proof}

We are finally ready to give our application of the trickle down framework.

\begin{proposition}\label{prop:slr-exact-link-profile}
Assume the events of Proposition~\ref{prop:approx_post_i} and Lemma~\ref{lem:slr-gaussian-profile} hold, and that 
\[n = \Omega\Par{k^{1.5} L^2 + kL^3}\]
for a sufficiently large constant. Then for all $2 \le s \le k$ and $\beta \in [0, 1]$, following the notation \eqref{eq:slr-tempered}, $\TDU{\nu_{s,\beta}}$ satisfies a $\frac 1 {2k}$-Poincar\'e inequality.
\end{proposition}
\begin{proof}
Let $R \in \calU_{s - 2}$, and define $\mk \in \R^{B \times B}$ following the notation in Lemma~\ref{lem:combine_params}. We first claim 
\begin{equation}\label{eq:quadform_bound_slr}\vu^\top \mk \vu \le \frac 9 {16} \norm{\vu}_2^2 + \frac 1 {8k}\norm{\vu}_1^2, \text{ for all } \vu \in \R^B. \end{equation}
To see this, decompose $\mk$ as in Item~\ref{item:entry_perturb_bound} of Lemma~\ref{lem:combine_params}, and by homogeneity, assume that $\norm{\vu}_1 = 1$ and $\norm{\vu}_2 = q$. By combining the estimates in Items~\ref{item:quadform_bound} and~\ref{item:reweight_u}, and using that 
\begin{align*}
\Abs{\sum_{\substack{(i, j) \in B \times B \\ i \neq j}} \mr_{ij}\vr_i\vr_j \vu_i\vu_j } &= O\Par{\alpha^2 k\Par{L + \frac{k^2 L^2}{n}} q^2} \\
&+ O\Par{\alpha \sqrt{L}\Par{1 + \frac{kqL}{\sqrt n}}\Par{\sqrt L + \frac{kL}{\sqrt n}} q} \\
&+ O\Par{\alpha^2 L\Par{1 + \frac{k^2 L^2 q^2}{n}}} \\
&= O\Par{\Par{\frac{kL^2}{n} + \frac{k^3 L^3}{n^2} + \frac{k L^{2.5}}{n} + \frac{k^2 L^{3}}{n^{1.5}} + \frac{k^2 L^4}{n^2}} q^2} \\
&+ O\Par{\Par{\frac{L^{1.5}}{\sqrt n} + \frac{k L^{2}}{n}}q} + O\Par{\frac{L^2}{n}} \\
&\le \frac 1 {4} q^2 + \frac 1 {16\sqrt k} q + \frac 1 {32k} \le \frac 1 {2} q^2 + \frac 1 {16k},
\end{align*}
where the hidden constants above depend only on $C'$. The last line then follows by taking $n$ large enough as stated. Next, for the error term in Item~\ref{item:entry_perturb_bound}, recalling $\norm{\vu}_1 = 1$,
\begin{equation}\label{eq:error_quadform}
\begin{aligned}\Abs{\sum_{\substack{(i, j) \in B \times B \\ i \neq j}} \eps_{ij} \vu_i\vu_j} &= O\Par{\alpha^2 + \alpha^4k^2} \cdot \sum_{\substack{(i, j) \in B \times B \\ i \neq j}} |\vu_i\vu_j| \Par{1 + \vr_i^2 + \vr_j^2} \\
&= O\Par{\alpha^2 + \alpha^4 k^2} \cdot \Par{1 + 2\sum_{i \in B} |\vu_i|\vr_i^2} \\
&= O\Par{\Par{\alpha^2 + \alpha^4 k^2} \Par{L + \frac{k^2 L^2 q}{n}}} \\
&= O\Par{\Par{\frac{k^2 L^3}{n^2} + \frac{k^4 L^4 }{n^3}}q}+ O\Par{\frac {L^2} n + \frac{k^2 L^3}{n^2}} \\
&\le \frac 1 {16\sqrt k} q + \frac 1 {32k} \le \frac 1 {16} q^2 + \frac 1 {16k},
\end{aligned}
\end{equation}
where we used Item~\ref{item:reweight_u} in the third line, and again simplified by taking $n$ large enough. Combining the above two displays gives the desired \eqref{eq:quadform_bound_slr}. 

Next, we observe that the link graph induced by $R$, in the sense of Lemma~\ref{lem:trickledown}, has edge weight matrix $\mw$ exactly given in the form required by Lemma~\ref{lem:main_sgt}, with $\mk$ defined in Lemma~\ref{lem:combine_params} and
\[\va_i \defeq \exp\Par{\beta\Par{F\Par{R \cup \{i\}} - \half F(R)}} \text{ for all } i \in B.\]
Indeed, for $(i, j) \in B \times B$ with $i \neq j$,
\[\va_i\va_j\exp\Par{\beta \mk_{ij}} = \exp\Par{\beta F(R \cup \{i, j\})} \]
as is required by the link graph. We now provide bounds on $m(\beta\mk)$ and $\rho(\beta\mk)$ as used in Lemma~\ref{lem:main_sgt}.

To bound $m(\beta\mk)$, recall from the proof of Lemma~\ref{lem:combine_params} (i.e., $\vr_R = \vs_B^{(0)} + \vd$, \eqref{eq:short_flat_scaled}, and \eqref{eq:d_bounds}) that we showed $\norm{\vr_R}_\infty = O(\sqrt L + \frac{kL}{\sqrt n})$. Thus, combining with \eqref{eq:entrywise} and Item~\ref{item:entry_perturb_bound} shows that
\[m(\beta\mk) \le m(\mk) = O\Par{\Par{\alpha + \alpha^2 k}\Par{L + \frac{k^2L^2}{n}}} + O\Par{\Par{\alpha^2 + \alpha^4 k^2}\Par{L + \frac{k^2L^2}{n}}} \le \frac 1 {16}.\]
To bound $\rho(\beta\mk)$, it suffices to restrict to $\norm{\vu}_1 = 1$, $\norm{\vu}_2 = q$ by homogeneity. We follow the proof of Lemma~\ref{lem:mixed_quadratic}, and write $\mx = \beta\mk + \mq$ where $\mx_{ij} = \exp(\beta\mk_{ij}) - 1$ entrywise. By Taylor expansion, we have that $|\mq_{ij}| \le \beta^2\mk_{ij}^2$. Again, using Item~\ref{item:entry_perturb_bound}, denoting $A \defeq \alpha^2 + \alpha^4 k^2 = \frac L n + \frac{k^2 L^2}{n^2} \le 1$,
\[\mk_{ij}^2 = O\Par{A\vr_i^2\vr_j^2 + A^2\Par{1 + \vr_i^4 + \vr_j^4}}.\]
Thus, applying the bounds from Item~\ref{item:reweight_u}, and simplifying,
\begin{align*}\Abs{\vu^\top \mq \vu} &= O\Par{A\Par{L^2 + \frac{k^4 L^4q^2}{n^2}} + A^2\Par{1 + \norm{\vr_R}^2_\infty}\Par{1 + \sum_{i \in B} |\vu_i|\vr_i^2} } \\
&= O\Par{A\Par{L^2 + \frac{k^4 L^4q^2}{n^2}} + A^2\Par{L + \frac{k^2 L^2}{n}}\Par{L + \frac{k^2 L^2 q}{n}} } \\
&= O\Par{\frac{Ak^4L^4}{n^2} q^2} + O\Par{\Par{\frac{A^2 k^2 L^3}{n} + \frac{A^2 k^4 L^4}{n^2}}q} \\
&+ O\Par{AL^2 + A^2 L^2 + \frac{A^2 k^2 L^3}{n}} \\
&\le \frac{1}{32} q^2 + \frac 1 {16\sqrt k} q+ \frac 1 {32k} \le \frac 1 {16} q^2 + \frac 1 {16k}.
\end{align*}
Above, the first line used a similar simplification as in \eqref{eq:error_quadform}, as well as our bound on $\norm{\vr_R}_\infty$. 
The last inequality used our lower bound on $n$ to simplify the various terms.

In conclusion, combining with \eqref{eq:quadform_bound_slr}, we have shown that
\[\vu^\top \Par{\mx - \id_B} \vu \le \frac 1 {4k}\Par{\norm{\vu}_1^2 - \norm{\vu}_2^2} \implies \rho(\beta\mk) \le \frac 1 {4k} .\]
The rest of the proof follows analogously to Lemma~\ref{lem:ising_trickle}, using our bounds on $\rho(\beta\mk)$, $m(\beta\mk)$.
\end{proof}

We conclude the section by bounding the range of the potential, for use with Section~\ref{sec:bm_sampling_anneal}.

\begin{lemma}
\label{lem:slr-exact-range}
Assume the events of Proposition~\ref{prop:approx_post_i} and Lemma~\ref{lem:slr-gaussian-profile} hold, and that
\[n = \Omega\Par{k^{1.5} + kL^3}\]
for a sufficient constant. Then following the notation \eqref{eq:slr-F-normalized},
\[
    \max_{S\in\calU_{\le k}} |F(S)|
    = O\Par{k\Par{1 + \sig^2}L + k\norm{\vh}_\infty}.
\]
\end{lemma}

\begin{proof}
For every $S \in \calU_{\le k}$, Lemma~\ref{lem:slr-gaussian-profile} gives $\normsop{\me_{S \times S}} \le \half$, so the spectrum of $\ml_S^{-1}$ is in $[\frac 2 3, 2]$. Moreover, \eqref{eq:short_flat_scaled} and $n = \Omega(kL)$ give $\twonorm{\vs_S} = O(\sqrt{kL})$. Therefore,
\[
    \frac{1}{2\tau}
    \vs_S^\top\ml_S^{-1}\vs_S
    \le
    \frac{1}{\tau}\twonorm{\vs_S}^2
   = O\Par{k (1 + \sig^2) L}.
\]
Further, $|\vh^\top \1_S| \le k\norm{\vh}_\infty$. The claim follows, as $\Abs{\log\det\ml_S} = O(k)$ using our spectrum bound.
\end{proof}

For completeness, we record that under the event of Lemma~\ref{lem:slr-gaussian-profile}, combining the mixing guarantee in Lemma~\ref{lem:variance_decay}, the spectral gap in Proposition~\ref{prop:slr-exact-link-profile}, and the Lemma~\ref{lem:slr-exact-range}, shows that it suffices to take
\begin{equation}\label{eq:mixing_time_slr}
T = \Omega\Par{k^2\Par{\Par{1 + \sig^2}L + \norm{\vh}_\infty} \log\Par{\frac 1 {\delta'}}}
\end{equation}
steps of $\lazy(\TDU{\nu_{s,\beta}})$, 
for a sufficient constant, to sample from within $\delta'$ TV from $\nu_{s, \beta}$ \eqref{eq:slr-tempered}.
\subsection{Main result}
\label{ssec:main_slr}

In this section, we finally put together the pieces to give our main sampling result for spike-and-slab posterior densities under Model~\ref{model:sas_basic}. We use the notation $F_\calC(S) \defeq F(\calC \cup S)$ defined in Remark~\ref{rem:slr_fixed_core}, and for notational simplicity, we also let $\nu^{\calC}_{s,\beta}(S) \propto \exp(\beta F_{\calC}(S)) \ind(S \in \calU_{s})$.

\begin{algorithm}
\caption{$\SAS(\mx,\vy,\sig,\vq,\delta)$}
\label{alg:post_sample}
\DontPrintSemicolon
\textbf{Input:}
$\mx\in\R^{n\times d}$, $\vy\in\R^n$, $\sig>0$,
$\vq\in(0,1)^d$, and $\delta\in(0,\half)$\;
\textbf{Output:}
Sample $\vth$ such that $\TV{\Law(\vth), \pi(\cdot \mid \mx, \vy)} \le \delta$ with probability $\ge 1 - \delta$ under Model~\ref{model:sas_basic}\;
$(\vz,\calU)\gets$ output of
Proposition~\ref{prop:approx_post_i}, with error parameter $\frac \delta 3 $\;
$k\gets24(\bark+\log(\frac 3 \delta))$\;
$\calC \gets \calU^c$\;
$\alg_0 \gets $ algorithm that always outputs $\emptyset$\;
\For{$s \in [k]$}{
    $\alg_s(\beta,\delta')\gets (\lazy(\TDU{\nu^{\calC}_{s, \beta}}))^T$ applied to an arbitrary start, for $T$ as in \eqref{eq:mixing_time_slr}\;
}
$S\sim\BMGS(k,1,F_\calC,\{\alg_s\}_{s=0}^k,\frac \delta 3)$\;
$\widetilde S\gets S\cup\calU^c$\;
\Return
$\vth\sim\Nor(
    \ma_{\widetilde S}^{-1}\vb_{\widetilde S},
    \ma_{\widetilde S}^{-1}
)$ following Fact~\ref{fact:pisupp}\;
\end{algorithm}

\begin{theorem}
\label{thm:sas_post_sampler}
Let $\delta\in(0,\half)$, and following the notation of Model~\ref{model:sas_basic}, let $k=24(\bark+\log(\frac 3 \delta))$. Suppose 
\[n = \Omega\left( k^{1.5}\log^2\Par{\frac{d}{\delta}} + k\log^3\Par{\frac d \delta} \right)\]
for a sufficiently large constant, and assume that $\vq \in [\eta, 1- \eta]^d$ for $\eta > 0$.  Then, with probability $\ge 1-\delta$ over the randomness of Model~\ref{model:sas_basic}, the output of Algorithm~\ref{alg:post_sample} satisfies
\[
    \TV{
        \Law(\vth),
        \pi(\cdot\mid\mx,\vy)
    }
    \le
    \delta.
\]
Defining $Q \defeq k(1 + \sig^2)\log(\frac d \delta) + k\log (\frac 1 \eta + \frac 1 {\eta\sig})$, the algorithm runs in time
\[
    O\left(
        \frac{ndk^3Q^2}{\delta^2}
        \log\Par{\frac{k}{\delta}}\log\Par{\frac{kQ}{\delta}}
    \right).
\]
\end{theorem}

\begin{proof}
Throughout the proof, condition on the success of Proposition~\ref{prop:approx_post_i} and the event in Lemma~\ref{lem:slr-gaussian-profile}, which give the failure probability over Model~\ref{model:sas_basic}. Under these events, the extended draw $\widetilde{S}$ in Algorithm~\ref{alg:post_sample} is within total variation distance $\delta$ from the support posterior $\pi(\supp(\cdot)\mid \mx, \vy)$. This is because the density $\hpisupp$ in \eqref{eq:nudef_general} is within TV $\frac \delta 3$ of the support posterior by Proposition~\ref{prop:approx_post_i} and Remark~\ref{rem:slr_fixed_core}, and 
the guarantees of $\BMGS$ in Lemma~\ref{lem:annealing_sampler_gen_ada} imply $S$ is within TV $\frac \delta 3$ of an exact sample from $\hpisupp$. Finally, the sample $\vth \mid \widetilde{S}$ is exact (Fact~\ref{fact:pisupp}) and cannot increase TV.

It remains to bound the implementation cost. Observe that Lemma~\ref{lem:slr-exact-range} shows that the potential $F_{\calC}$ is bounded by $O(Q)$ under the assumption $\vq \in [\eta, 1-\eta]^d$, so Lemma~\ref{lem:annealing_sampler_gen_ada} uses
\[N = O\Par{\frac{kQ\log(\frac{k}{\delta})}{\delta^2}}\]
calls to algorithms $\calA_s(\beta, \frac \delta {6N})$, each using $T$ steps of the down-up walk where $T$ is defined in \eqref{eq:mixing_time_slr}. We claim that each step can be implemented in $O(ndk)$ time, which gives the runtime claim.

To prove the implementation cost of $\alg_s$ for $s \in [k]$, fix some $R \in \binom{\calU}{s-1}$ that is the result of dropping an element (i.e., after Line~\ref{line:downstep} of Algorithm~\ref{alg:DU_walk} has completed). For simplicity, we consider the case when $\calC = \emptyset$. In the general case every core $R$ appearing below is replaced by $\calC \cup R$, whose size remains $O(k)$. To implement Line~\ref{line:upstep}, Lemma~\ref{lem:slr-schur-reduction} gives
\[
    F_\calC(R\cup\{j\})-F_\calC(R)
    =
    \vh_j
    +
    \frac{1}{2\tau}
    \frac{\vr_j^2}{d_j}
    -
    \frac12\log d_j,
\]
where we followed the notation of \eqref{eq:k_entry_derive}, so
\[
    r_j
    =
    \vs_j
    -
    \tau
    \me_{\{j\}\times R}
    \ml_R^{-1}\vs_R,\quad  d_j \defeq 1 + \tau\mr_{jj} = 1 + \tau\me_{jj} - \tau^2 \me_{\{j\} \times R} \ml_R^{-1} \me_{R \times \{j\}}.
\]
We can compute $\me_{B \times R} = [\mx^\top \mx - \id_d]_{B \times R}$ in time $O(ndk)$, and similarly we can compute and invert $\ml_R = \id_R + \tau \me_R$ in this time. We can also compute all diagonal entries of $\me$ in time $O(nd)$. Thus, computing all of the $r_j$ takes time $O(ndk)$, and computing each $d_j$ takes time $O(k^2) = O(nk)$ to compute the relevant quadratic form, so computing all of them takes time $O(ndk)$ as well. 

\end{proof}

\section*{Acknowledgments}
We thank Thuy-Duong (June) Vuong for her participation at an earlier stage of this project, as well as Sidhanth Mohanty for several helpful conversations. We also thank an anonymous FOCS reviewer for making a suggestion that led to our strategy in Section~\ref{sec:trickledown}. SK gratefully acknowledges funding support from the Amazon AI PhD Fellowship. KT and YZ thank the NSF AI Institute for Foundations of Machine Learning (IFML) for supporting this project.

\section*{AI Disclosure}
\phantomsection
\addcontentsline{toc}{section}{AI Disclosure}

A preliminary version of this paper, consisting of Sections~\ref{sec:dobrushin},~\ref{sec:bm_sampling_anneal}, and~\ref{sec:slr} (at a measurement complexity $n \gtrsim k^2$), was previously submitted to FOCS, with all ideas contributed by the authors. Based on a reviewer's suggestion, the authors used GPT 5.6 Pro to explore applications of the trickle down theorem to sharpen our results. Specifically, the key perturbation strategy in Lemma~\ref{lem:main_sgt} was suggested by GPT, which led to a weaker variant of Theorem~\ref{thm:informal-trickledown} (Corollary~\ref{cor:near_proportional_sk}). The authors then built on this approach in our final applications  in Theorems~\ref{thm:informal-trickledown} and~\ref{thm:informal-slr}. We also acknowledge the use of LLMs in understanding the literature on the Almeida-Thouless line, as well as the prior works \cite{Carlson2022, KuchukovaPappikPerkinsYap2025}, which helped us prepare Appendices~\ref{sec:sk_at_line} and~\ref{app:critical}. The manuscript was written solely by the authors, who take full responsibility for the organization and presentation of all results.

\bibliographystyle{alpha}
\bibliography{refs}

@inproceedings{AlaouiMS22,
  author       = {Ahmed El Alaoui and
                  Andrea Montanari and
                  Mark Sellke},
  title        = {Sampling from the Sherrington-Kirkpatrick Gibbs measure via algorithmic
                  stochastic localization},
  booktitle    = {63rd {IEEE} Annual Symposium on Foundations of Computer Science, {FOCS}
                  2022},
  pages        = {323--334},
  publisher    = {{IEEE}},
  year         = {2022}
}

@inproceedings{alev2020improved,
  title={Improved analysis of higher order random walks and applications},
  author={Alev, Vedat Levi and Lau, Lap Chi},
  booktitle={Proceedings of the 52nd annual ACM SIGACT symposium on theory of computing},
  pages={1198--1211},
  year={2020}
}

@article{Aizenman_2014,
   title={Random Currents and Continuity of Ising Model’s Spontaneous Magnetization},
   volume={334},
   ISSN={1432-0916},
   url={http://dx.doi.org/10.1007/s00220-014-2093-y},
   DOI={10.1007/s00220-014-2093-y},
   number={2},
   journal={Communications in Mathematical Physics},
   publisher={Springer Science and Business Media LLC},
   author={Aizenman, Michael and Duminil-Copin, Hugo and Sidoravicius, Vladas},
   year={2014},
   month=jul, pages={719–742} }

@inproceedings{AlmanDWXXZ25,
  author       = {Josh Alman and
                  Ran Duan and
                  Virginia {Vassilevska Williams} and
                  Yinzhan Xu and
                  Zixuan Xu and
                  Renfei Zhou},
  title        = {More Asymmetry Yields Faster Matrix Multiplication},
  booktitle    = {Proceedings of the 2025 Annual {ACM-SIAM} Symposium on Discrete Algorithms,
                  {SODA} 2025},
  pages        = {2005--2039},
  publisher    = {{SIAM}},
  year         = {2025}
}

@article{AmitGS85,
  title={Storing infinite numbers of patterns in a spin-glass model of neural networks},
  author={Amit, Daniel J and Gutfreund, Hanoch and Sompolinsky, Haim},
  journal={Physical review letters},
  volume={55},
  number={14},
  pages={1530},
  year={1985},
  publisher={APS}
}

@article{AmitGS87,
  title={Statistical mechanics of neural networks near saturation},
  author={Amit, Daniel J and Gutfreund, Hanoch and Sompolinsky, Haim},
  journal={Annals of physics},
  volume={173},
  number={1},
  pages={30--67},
  year={1987},
  publisher={Elsevier}
}

@inproceedings{AnariLGV19,
  author       = {Nima Anari and
                  Kuikui Liu and
                  Shayan Oveis Gharan and
                  Cynthia Vinzant},
  editor       = {Moses Charikar and
                  Edith Cohen},
  title        = {Log-concave polynomials {II:} high-dimensional walks and an {FPRAS}
                  for counting bases of a matroid},
  booktitle    = {Proceedings of the 51st Annual {ACM} {SIGACT} Symposium on Theory
                  of Computing, {STOC} 2019},
  pages        = {1--12},
  publisher    = {{ACM}},
  year         = {2019}
}

@inproceedings{AnariJKPV22,
  author       = {Nima Anari and
                  Vishesh Jain and
                  Frederic Koehler and
                  Huy Tuan Pham and
                  Thuy{-}Duong Vuong},
  title        = {Entropic independence: optimal mixing of down-up random walks},
  booktitle    = {{STOC} '22: 54th Annual {ACM} {SIGACT} Symposium on Theory of Computing},
  pages        = {1418--1430},
  publisher    = {{ACM}},
  year         = {2022}
}

@inproceedings{AnariKV24,
  author       = {Nima Anari and
                  Frederic Koehler and
                  Thuy{-}Duong Vuong},
  title        = {Trickle-Down in Localization Schemes and Applications},
  booktitle    = {Proceedings of the 56th Annual {ACM} Symposium on Theory of Computing,
                  {STOC} 2024},
  pages        = {1094--1105},
  publisher    = {{ACM}},
  year         = {2024}
}

@book{AndersonGZ10,
  title={An introduction to random matrices},
  author={Anderson, Greg W and Guionnet, Alice and Zeitouni, Ofer},
  year={2010},
  publisher={Cambridge university press}
}

@misc{BandeiraElAlaouiRodder26,
      title={Mixing of Glauber Dynamics on High Overlap Gibbs Measures}, 
      author={Afonso S. Bandeira and Ahmed El Alaoui and Almut Rödder},
      year={2026},
      eprint={2607.06813},
      archivePrefix={arXiv},
      primaryClass={math.PR},
      url={https://arxiv.org/abs/2607.06813}, 
}

@article{BarraG08,
  title={About the ergodic regime in the analogical Hopfield neural networks: moments of the partition function},
  author={Barra, Adriano and Guerra, Francesco},
  journal={Journal of mathematical physics},
  volume={49},
  number={12},
  year={2008},
  publisher={AIP Publishing}
}

@article{Bonati2014,
   title={The Peierls argument for higher dimensional Ising models},
   volume={35},
   ISSN={1361-6404},
   url={http://dx.doi.org/10.1088/0143-0807/35/3/035002},
   DOI={10.1088/0143-0807/35/3/035002},
   number={3},
   journal={European Journal of Physics},
   publisher={IOP Publishing},
   author={Bonati, Claudio},
   year={2014},
   month=mar, pages={035002} }

@article{BaiRG21,
    title = {Spike-and-Slab Meets LASSO: A Review of the Spike-and-Slab LASSO},
    author = {Ray Bai and Veronika Rockova and Edward I.\ George},
    journal = {Handbook of Bayesian Variable Selection},
    pages = {81--108},
    year = {2021}
}

@article{BandeiraVH16,
  author  = {Afonso S. Bandeira and Ramon van Handel},
  title   = {Sharp Nonasymptotic Bounds on the Norm of Random Matrices with Independent Entries},
  journal = {The Annals of Probability},
  volume  = {44},
  number  = {4},
  pages   = {2479--2506},
  year    = {2016}
}

@article{BauerschmidtBodineauDagallier2024,
  author  = {Roland Bauerschmidt and Thierry Bodineau and Benoit Dagallier},
  title   = {Kawasaki Dynamics Beyond the Uniqueness Threshold},
  journal = {Probability Theory and Related Fields},
  volume  = {192},
  number  = {1--2},
  pages   = {267--302},
  year    = {2024},
  doi     = {10.1007/s00440-024-01326-9}
}

@article{BovierEN99,
  title={Stochastic symmetry-breaking in a Gaussian Hopfield model},
  author={Bovier, Anton and van Enter, Aernout CD and Niederhauser, Beat},
  journal={Journal of statistical physics},
  volume={95},
  number={1},
  pages={181--213},
  year={1999},
  publisher={Springer}
}

@article{brennecke2022replica,
  title={The replica symmetric formula for the SK model revisited},
  author={Brennecke, Christian and Yau, Horng-Tzer},
  journal={Journal of Mathematical Physics},
  volume={63},
  number={7},
  year={2022},
  publisher={AIP Publishing}
}

@inproceedings{BrunaH24,
  author       = {Joan Bruna and
                  Jiequn Han},
  title        = {Provable Posterior Sampling with Denoising Oracles via Tilted Transport},
  booktitle    = {Advances in Neural Information Processing Systems 38: Annual Conference
                  on Neural Information Processing Systems 2024, NeurIPS 2024},
  year         = {2024}
}

@article{CandesRT06,
  title={Robust uncertainty principles: Exact signal reconstruction from highly incomplete frequency information},
  author={Cand{\`e}s, Emmanuel J and Romberg, Justin and Tao, Terence},
  journal={IEEE Transactions on information theory},
  volume={52},
  number={2},
  pages={489--509},
  year={2006},
  publisher={IEEE}
}

@inproceedings{Carlson2022,
author = {Carlson, Charlie and Davies, Ewan and Kolla, Alexandra and Perkins, Will},
title = {Computational thresholds for the fixed-magnetization Ising model},
year = {2022},
isbn = {9781450392648},
publisher = {Association for Computing Machinery},
url = {https://doi.org/10.1145/3519935.3520003},
doi = {10.1145/3519935.3520003},
booktitle = {Proceedings of the 54th Annual ACM SIGACT Symposium on Theory of Computing},
pages = {1459–1472},
series = {STOC 2022}
}

@article{CandesT05,
  title={Decoding by linear programming},
  author={Candes, Emmanuel J and Tao, Terence},
  journal={IEEE transactions on information theory},
  volume={51},
  number={12},
  pages={4203--4215},
  year={2005},
  publisher={IEEE}
}

@article{CandesT06,
  title={Near-optimal signal recovery from random projections: Universal encoding strategies?},
  author={Candes, Emmanuel J and Tao, Terence},
  journal={IEEE transactions on information theory},
  volume={52},
  number={12},
  pages={5406--5425},
  year={2006},
  publisher={IEEE}
}

@inproceedings{CarvalhoPS09,
  author       = {Carlos M. Carvalho and
                  Nicholas G. Polson and
                  James G. Scott},
  title        = {Handling Sparsity via the Horseshoe},
  booktitle    = {Proceedings of the Twelfth International Conference on Artificial
                  Intelligence and Statistics, {AISTATS} 2009},
  series       = {{JMLR} Proceedings},
  volume       = {5},
  pages        = {73--80},
  publisher    = {JMLR.org},
  year         = {2009}
}

@article{castillo2012needles,
  author  = {Castillo, Isma{\"e}l and van der Vaart, Aad},
  title   = {Needles and Straw in a Haystack: Posterior Concentration for Possibly Sparse Sequences},
  journal = {The Annals of Statistics},
  volume  = {40},
  number  = {4},
  pages   = {2069--2101},
  year    = {2012},
  month   = aug,
  doi     = {10.1214/12-AOS1029}
}

@article{Chipman96,
    title = {Bayesian variable selection with related predictors},
    author = {H. Chipman},
    journal = {The Canadian Journal of Statistics},
    volume = {24},
    pages = {17--36},
    year = {1996}
}

@article{ChenLTZ26,
  title={Separating Oblivious and Adaptive Models of Variable Selection},
  author={Chen, Ziyun and Li, Jerry and Tian, Kevin and Zhu, Yusong},
  journal={arXiv preprint arXiv:2602.16568},
  year={2026}
}

@inproceedings{CryanGM19,
  author       = {Mary Cryan and
                  Heng Guo and
                  Giorgos Mousa},
  title        = {Modified log-Sobolev Inequalities for Strongly Log-Concave Distributions},
  booktitle    = {60th {IEEE} Annual Symposium on Foundations of Computer Science, {FOCS}
                  2019},
  pages        = {1358--1370},
  publisher    = {{IEEE} Computer Society},
  year         = {2019}
}

@article{CSHVdV15,
 ISSN = {00905364},
 URL = {http://www.jstor.org/stable/43818568},
 author = {Castillo, Ismael and Schmidt-Hieber, Johannes and Van der Vaart, Aad},
 journal = {The Annals of Statistics},
 number = {5},
 pages = {1986-2018},
 publisher = {Institute of Mathematical Statistics},
 title = {BAYESIAN LINEAR REGRESSION WITH SPARSE PRIORS},
 urldate = {2026-03-31},
 volume = {43},
 year = {2015}
}

@article{DaviesLSS26,
  author       = {Ewan Davies and
                  Holden Lee and
                  Juspreet Singh Sandhu and
                  Jonathan Shi},
  title        = {Potential Hessian Ascent {III:} Sampling the Sherrington-Kirkpatrick
                  Model at Beta {\textless} 1/2},
  journal      = {CoRR},
  volume       = {abs/2605.03718},
  year         = {2026}
}

@article{deAlmeidaThouless78,
  title={Stability of the Sherrington-Kirkpatrick solution of a spin glass model},
  author={de Almeida, Jairo RL and Thouless, David J},
  journal={Journal of Physics A: Mathematical and General},
  volume={11},
  number={5},
  pages={983--990},
  year={1978}
}

@article{Dobruschin68,
  title={The description of a random field by means of conditional probabilities and conditions of its regularity},
  author={Dobruschin, PL},
  journal={Theory of Probability \& Its Applications},
  volume={13},
  number={2},
  pages={197--224},
  year={1968},
  publisher={SIAM}
}

@article{Donoho06,
  title={Compressed sensing},
  author={Donoho, David L},
  journal={IEEE Transactions on information theory},
  volume={52},
  number={4},
  pages={1289--1306},
  year={2006},
  publisher={IEEE}
}

@article{DP21,
  title={Approximately counting independent sets of a given size in bounded-degree graphs},
  author={Davies, Ewan and Perkins, Will},
  journal={SIAM Journal on Computing},
  volume={52},
  number={2},
  pages={618--640},
  year={2023},
  publisher={SIAM}
}

@article{EldanKZ22,
  title={A spectral condition for spectral gap: fast mixing in high-temperature Ising models},
  author={Eldan, Ronen and Koehler, Frederic and Zeitouni, Ofer},
  journal={Probability theory and related fields},
  volume={182},
  number={3},
  pages={1035--1051},
  year={2022},
  publisher={Springer}
}

@book{Ellis12,
  title={Entropy, large deviations, and statistical mechanics},
  author={Ellis, Richard S},
  year={2012},
  publisher={Springer Science \& Business Media}
}

@inproceedings{GarnaevG84,
  title={The widths of a Euclidean ball},
  author={Garnaev, Andrei Yur'evich and Gluskin, Efim Davydovich},
  booktitle={Doklady Akademii Nauk},
  volume={277},
  pages={1048--1052},
  year={1984},
  organization={Russian Academy of Sciences}
}

@article{GeorgeM93,
  author  = {Edward I George and Robert E McCulloch},
  title   = {Variable selection via Gibbs sampling},
  journal = {Journal of the American Statistical Association},
  volume  = {88},
  number  = {423},
  pages   = {881--889},
  year    = {1993}
}

@article{Geweke96,
    title = {Variable selection and model comparison in regression},
    author = {J. Geweke},
    journal = {Bayesian Statistics},
    volume = {5},
    pages = {609--620},
    year = {1996}
}

@inproceedings{gotlib2023nowhere,
  title={Nowhere to go but high: a perspective on high-dimensional expanders},
  author={Gotlib, Roy and Kaufman, Tali},
  booktitle={International Congress of Mathematicians},
  pages={4842--4871},
  year={2023},
  organization={European Mathematical Society-EMS-Publishing House GmbH}
}

@article{HamzeRPBK20,
  title={Wishart planted ensemble: A tunably rugged pairwise Ising model with a first-order phase transition},
  author={Hamze, Firas and Raymond, Jack and Pattison, Christopher A and Biswas, Katja and Katzgraber, Helmut G},
  journal={Physical Review E},
  volume={101},
  number={5},
  pages={052102},
  year={2020},
  publisher={APS}
}

@article{Hopfield82,
  title={Neural networks and physical systems with emergent collective computational abilities},
  author={Hopfield, John J},
  journal={Proceedings of the national academy of sciences},
  volume={79},
  number={8},
  pages={2554--2558},
  year={1982}
}

@article{IshwaranS11,
    title = {Consistency of spike and slab regression},
    author = {Hemant Ishwaran and J. Sunil Rao},
    journal = {Statistics \& Probability Letters},
    volume = {81},
    number = {12},
    pages = {1920--1928},
    year = {2011}
}

@article{JagannathTobasco17,
  title={Some properties of the phase diagram for mixed p-spin glasses},
  author={Jagannath, Aukosh and Tobasco, Ian},
  journal={Probability Theory and Related Fields},
  volume={167},
  number={3},
  pages={615--672},
  year={2017},
  publisher={Springer}
}

@inproceedings{JMPV23,
  title={Optimal mixing of the down-up walk on independent sets of a given size},
  author={Jain, Vishesh and Michelen, Marcus and Pham, Huy Tuan and Vuong, Thuy-Duong},
  booktitle={2023 IEEE 64th Annual Symposium on Foundations of Computer Science (FOCS)},
  pages={1665--1681},
  year={2023},
  organization={IEEE}
}

@InProceedings{JM24,
  author =	{Jain, Vishesh and Mizgerd, Clayton},
  title =	{{Rapid Mixing of the Down-Up Walk on Matchings of a Fixed Size}},
  booktitle =	{Approximation, Randomization, and Combinatorial Optimization. Algorithms and Techniques (APPROX/RANDOM 2024)},
  pages =	{63:1--63:13},
  series =	{Leibniz International Proceedings in Informatics (LIPIcs)},
  ISBN =	{978-3-95977-348-5},
  ISSN =	{1868-8969},
  year =	{2024},
  volume =	{317}
}

@article{JohnstoneS04,
    title = {Needles and Straw in Haystacks: Empirical Bayes Estimates of Possibly Sparse Sequences},
    author = {Iain M. Johnstone and Bernard W. Silverman},
    journal = {Annals of Statistics},
    volume = {32},
    number = {4},
    pages = {1594--1649},
    year = {2004}
}

@article{Kashin77,
  title={Diameters of some finite-dimensional sets and classes of smooth functions},
  author={Kashin, Boris Sergeevich},
  journal={Izvestiya Rossiiskoi Akademii Nauk. Seriya Matematicheskaya},
  volume={41},
  number={2},
  pages={334--351},
  year={1977},
  publisher={Russian Academy of Sciences, Steklov Mathematical Institute of Russian~…}
}

@article{KaufmanO20,
  title={High order random walks: Beyond spectral gap},
  author={Kaufman, Tali and Oppenheim, Izhar},
  journal={Combinatorica},
  volume={40},
  number={2},
  pages={245--281},
  year={2020},
  publisher={Springer}
}

@article{Kawasaki66,
  title={Diffusion constants near the critical point for time-dependent Ising models. II},
  author={Kawasaki, Kyozi},
  journal={Physical Review},
  volume={148},
  number={1},
  pages={375},
  year={1966},
  publisher={APS}
}

@inproceedings{KelnerLLST23,
  author       = {Jonathan A. Kelner and
                  Jerry Li and
                  Allen Liu and
                  Aaron Sidford and
                  Kevin Tian},
  title        = {Semi-Random Sparse Recovery in Nearly-Linear Time},
  booktitle    = {The Thirty Sixth Annual Conference on Learning Theory, {COLT} 2023},
  series       = {Proceedings of Machine Learning Research},
  volume       = {195},
  pages        = {2352--2398},
  publisher    = {{PMLR}},
  year         = {2023}
}

@inproceedings{Kolmogorov18,
  author       = {Vladimir Kolmogorov},
  title        = {A Faster Approximation Algorithm for the Gibbs Partition Function},
  booktitle    = {Conference On Learning Theory, {COLT} 2018},
  series       = {Proceedings of Machine Learning Research},
  volume       = {75},
  pages        = {228--249},
  publisher    = {{PMLR}},
  year         = {2018}
}

@article{KuchukovaPappikPerkinsYap2025,
  author  = {Aiya Kuchukova and Marcus Pappik and Will Perkins and Corrine Yap},
  title   = {Fast and Slow Mixing of the Kawasaki Dynamics on Bounded-Degree Graphs},
  journal = {Random Structures \& Algorithms},
  volume  = {67},
  number  = {4},
  year    = {2025},
  doi     = {10.1002/rsa.70038}
}

@inproceedings{KumarSTZ25,
  author       = {Symantak Kumar and
                  Purnamrita Sarkar and
                  Kevin Tian and
                  Yusong Zhu},
  title        = {Spike-and-Slab Posterior Sampling in High Dimensions},
  booktitle    = {The Thirty Eighth Annual Conference on Learning Theory},
  series       = {Proceedings of Machine Learning Research},
  volume       = {291},
  pages        = {3407--3462},
  publisher    = {{PMLR}},
  year         = {2025}
}

@article{kusuoka2026quantitative,
  title={A quantitative replica-symmetric bound of Sherrington--Kirkpatrick model in the entire de Almeida--Thouless region},
  author={Kusuoka, Seiichiro and Nakajima, Shuta},
  journal={arXiv preprint arXiv:2608.23413},
  year={2026}
}

@book{LevinPW09,
	author = {David Asher Levin and Yuval Peres and Elizabeth Wilmer},
	title = {Markov Chains and Mixing Times},
	publisher = {American Mathematical Society},
	year = {2009}
}

@article{Little74,
  title={The existence of persistent states in the brain},
  author={Little, William A},
  journal={Mathematical biosciences},
  volume={19},
  number={1-2},
  pages={101--120},
  year={1974},
  publisher={Elsevier}
}

@inproceedings{LiuMRW24,
  author       = {Kuikui Liu and
                  Sidhanth Mohanty and
                  Amit Rajaraman and
                  David X. Wu},
  title        = {Fast Mixing in Sparse Random Ising Models},
  booktitle    = {65th {IEEE} Annual Symposium on Foundations of Computer Science, {FOCS}
                  2024},
  pages        = {120--128},
  publisher    = {{IEEE}},
  year         = {2024}
}

@article{Lopatto26,
  title={Replica symmetry up to the de Almeida-Thouless line in the Sherrington-Kirkpatrick model},
  author={Lopatto, Patrick},
  journal={arXiv preprint arXiv:2604.11921},
  year={2026}
}

@article{Lyons89,
  title={The Ising model and percolation on trees and tree-like graphs},
  author={Lyons, Russell},
  journal={Communications in Mathematical Physics},
  volume={125},
  number={2},
  pages={337--353},
  year={1989},
  publisher={Springer}
}

@article{MitchellB88,
  author  = {Toby J Mitchell and John J Beauchamp},
  title   = {Bayesian variable selection in linear regression},
  journal = {Journal of the American Statistical Association},
  volume  = {83},
  number  = {404},
  pages   = {1023--1032},
  year    = {1988}
}

@article{MontanariW26,
  title={Provably efficient posterior sampling for sparse linear regression via measure decomposition},
  author={Montanari, Andrea and Wu, Yuchen},
  journal={Journal of the American Statistical Association},
  pages={1--19},
  year={2026},
  publisher={Taylor \& Francis}
}

@misc{Mossel06,
  author       = {Elchanan Mossel},
  title        = {Ising Model on Trees},
  howpublished = {Lecture notes for STAT 206A, University of California, Berkeley},
  year         = {2006},
  note         = {Lecture 20},
  url          = {https://www.stat.berkeley.edu/~mossel/teach/206af06/scribes/oct31.pdf}
}

@article{MS22,
  author = {Sumit Mukherjee and Subhabrata Sen},
  title = {Variational inference in high-dimensional linear regression},
  journal = {Journal of Machine Learning Research},
  volume = {23},
  pages = {1--56},
  year = {2022}
}

@article{oppenheim2018local,
  title={Local spectral expansion approach to high dimensional expanders part I: Descent of spectral gaps},
  author={Oppenheim, Izhar},
  journal={Discrete \& Computational Geometry},
  volume={59},
  number={2},
  pages={293--330},
  year={2018},
  publisher={Springer}
}

@inproceedings{PanC99,
  author       = {Victor Y. Pan and
                  Zhao Q. Chen},
  title        = {The Complexity of the Matrix Eigenproblem},
  booktitle    = {Proceedings of the Thirty-First Annual {ACM} Symposium on Theory of
                  Computing},
  pages        = {507--516},
  publisher    = {{ACM}},
  year         = {1999}
}

@article{PasturF77,
  title={Exactly soluble model of a spin glass},
  author={Pastur, Leonid A and Figotin, Alexander L},
  journal={Soviet Journal of Low Temperature Physics},
  volume={3},
  number={6},
  pages={378--383},
  year={1977},
  publisher={American Institute of Physics}
}

@article{PolsonS19,
    title = {Bayesian $\ell_0$-regularized least squares},
    author= {Nicholas G. Polson and Lei Sun},
    journal = {Applied Stochastic Models in Business and Industry},
    volume = {35},
    number = {3},
    pages = {717--731},
    year = {2019}
}

@inproceedings{rajaraman2026markov,
  title={Markov Chains Approximate Message Passing},
  author={Rajaraman, Amit and Wu, David X},
  booktitle={Proceedings of the 58th Annual ACM Symposium on Theory of Computing},
  pages={1192--1199},
  year={2026}
}

@article{RG14,
  title={EMVS: The EM approach to Bayesian variable selection},
  author={Veronika Ročková and Edward I. George},
  journal={Journal of the American Statistical Association},
  volume={109},
  number={506},
  pages={828--846},
  year={2014},
  publisher={Taylor & Francis}
}

@article{RG18,
  author  = {Veronika Ročková and Edward I. George},
  title   = {The spike-and-slab LASSO},
  journal = {Journal of the American Statistical Association},
  volume  = {113},
  number  = {521},
  pages   = {431--444},
  year    = {2018}
}

@article{Roc18,
    author = {Veronika Rockova},
    title = {Bayesian estimation of sparse signals with a continuous spike-and-slab prior},
    journal = {Annals of Statistics},
    volume = {46},
    number = {1},
    pages = {401--437},
    year = {2018}
}

@article{RS22,
  author = {Kolyan Ray and Botond Szabó},
  title = {Variational Bayes for high-dimensional linear regression with sparse priors},
  journal = {Journal of the American Statistical Association},
  volume = {117},
  number = {539},
  pages = {1270--1281},
  year = {2022}
}

@article{SherringtonK75,
  title={Solvable model of a spin-glass},
  author={Sherrington, David and Kirkpatrick, Scott},
  journal={Physical review letters},
  volume={35},
  number={26},
  pages={1792},
  year={1975},
  publisher={APS}
}

@article{Strassen69,
    title = {Gaussian elimination is not optimal},
    author = {Volker Strassen},
    journal = {Numerische Mathematik},
    volume = {13},
    pages = {354--356},
    year = {1969}
}

@book{talagrand2010mean,
  title={Mean field models for spin glasses: Volume I: Basic examples},
  author={Talagrand, Michel},
  volume={54},
  year={2010},
  publisher={Springer Science \& Business Media}
}

@book{Talagrand11,
author = {Talagrand, Michel},
year = {2011},
month = {01},
pages = {},
title = {Mean Field Models for Spin Glasses: Volume II: Advanced Replica-Symmetry and Low Temperature},
volume = {55},
isbn = {978-3-642-22252-8},
journal = {Ergebnisse der Mathematik und ihrer Grenzgebiete},
doi = {10.1007/978-3-642-22253-5},
publisher={Springer Science \& Business Media}
}

@article{Tropp12,
  title={User-friendly tail bounds for sums of random matrices},
  author={Tropp, Joel A},
  journal={Foundations of computational mathematics},
  volume={12},
  number={4},
  pages={389--434},
  year={2012},
  publisher={Springer}
}

@book{Vershynin18,
  title={High-dimensional probability: An introduction with applications in data science},
  author={Vershynin, Roman},
  volume={47},
  year={2018},
  publisher={Cambridge university press}
}

@article{WainwrightJ08,
  title={Graphical models, exponential families, and variational inference},
  author={Wainwright, Martin J and Jordan, Michael I},
  journal={Foundations and Trends{\textregistered} in Machine Learning},
  volume={1},
  number={1-2},
  pages={1--305},
  year={2008},
  publisher={Emerald Publishing Limited}
}

@book{Wainwright19,
  title={High-dimensional statistics: A non-asymptotic viewpoint},
  author={Wainwright, Martin J},
  volume={48},
  year={2019},
  publisher={Cambridge university press}
}

@article{Weiss07,
  author  = {Weiss, Pierre},
  title   = {L'hypothèse du champ moléculaire et la propriété ferromagnétique},
  journal = {Journal de Physique Théorique et Appliquée},
  volume  = {6},
  number  = {1},
  pages   = {661--690},
  year    = {1907},
  doi     = {10.1051/jphystap:019070060066100}
}

@article{Wu06,
   title={Poincaré and transportation inequalities for Gibbs measures under the Dobrushin uniqueness condition},
   volume={34},
   ISSN={0091-1798},
   url={http://dx.doi.org/10.1214/009117906000000368},
   DOI={10.1214/009117906000000368},
   number={5},
   journal={The Annals of Probability},
   publisher={Institute of Mathematical Statistics},
   author={Wu, Liming},
   year={2006},
   }

@article{YangWJ16,
  title={On the computational complexity of high-dimensional Bayesian variable selection},
  author={Yang, Yun and Wainwright, Martin J and Jordan, Michael I},
  journal = {The Annals of Statistics},
  year={2016}
}

@article{yang1952,
  title={The spontaneous magnetization of a two-dimensional Ising model},
  author={Yang, Chen Ning},
  journal={Physical Review},
  volume={85},
  number={5},
  pages={808},
  year={1952},
  publisher={APS}
}

\newpage
\appendix

\section{Sampling Near the Almeida–Thouless Line}
\label{sec:sk_at_line}

In this section, we give an application to sampling near the \emph{Almeida-Thouless (AT) line}. The AT line is parameterized by  $\beta > 0$ and a \emph{field strength} $h > 0$, and considers the specialized SK model 
\begin{equation}
    \label{eq:sk_uniform_field}
    \pi_{\beta,h}(\vx)
    \propto
    \exp\Par{
        \frac{\beta}{2}\vx^\top\mj\vx
        +
        h\1_d^\top\vx
    },
    \quad
    \vx\in\calX^d,
\end{equation}
where $\mj\sim\operatorname{GOE}(d)$ follows
Model~\ref{model:sk}, and we set the diagonal of $\mj$ to zero without
loss of generality.
Thus, $h$ denotes the coefficient of the linear term $\propto \1_d$. 

The AT line delineates a region in $\R^2_{\ge 0}$, representing a pair of parameters $(\beta, h)$ in \eqref{eq:sk_uniform_field}. This line was originally derived in \cite{deAlmeidaThouless78} via the replica method, with the prediction that models induced by $h$ above the line (i.e., large enough as a function of $\beta$) are \emph{replica symmetric}, and that pairs below the line exhibit \emph{replica symmetry breaking}. This prediction was recently established rigorously for the SK model with a homogeneous external field \cite{Lopatto26}. However, our application requires a stronger quantitative concentration estimate, so our results hold above the slightly more restrictive \emph{weak AT line} (Definition~\ref{def:wat_cond}), leveraging bounds by \cite{Talagrand11, JagannathTobasco17} as presented by \cite{rajaraman2026markov}.

For large $h$, \eqref{eq:sk_uniform_field} favors $\vx$ with more positive spins. Following this convention, we write
\[
    N_-(\vx)
    \defeq
    \Abs{\Brace{i\in[d]:x_i=-1}},
    \qquad
    \calX_{\le k}^{d,-}
    \defeq
    \Brace{\vx\in\calX^d:N_-(\vx)\le k}.
\]
Our earlier results use the number of positive spins as the sparsity
parameter; the two conventions are equivalent under the global spin flip
$\vx\to-\vx$. Moreover, on every fixed-magnetization slice, $h\1_d^\top \vx$ is constant, so all fixed-size mixing analyses for \eqref{eq:sk_uniform_field} are
independent of $h$.

We now define the regions determined by the AT line.

\begin{definition}[AT condition, \cite{deAlmeidaThouless78}]
    \label{def:at_cond}
    For $\beta,h>0$, let $q=q(\beta,h)$ be the unique solution to
    \[
        q
        =
        \bbE\Brack{
            \tanh^2
            \Par{
                \beta\sqrt q Z+h
            }
        },
        \qquad
        Z\sim\mathcal N(0,1).
    \]
    Define
    \[
        \alpha_{\rm AT}(\beta,h)
        \defeq
        \beta^2
        \bbE\Brack{
            \operatorname{sech}^4
            \Par{
                \beta\sqrt q Z+h
            }
        }.
    \]
    We say that $(\beta,h)$ satisfies the AT condition if
    $\alpha_{\rm AT}(\beta,h)\le1$. The AT boundary
    $h_{\rm AT}(\beta)$ is characterized by
    $\alpha_{\rm AT}(\beta,h_{\rm AT}(\beta))=1$. We also define $q\AT(\beta) \defeq q(\beta, h\AT(\beta))$.
\end{definition}

\begin{definition}[Weak AT condition, \cite{brennecke2022replica}]
    \label{def:wat_cond}
    With $q=q(\beta,h)$ as in Definition~\ref{def:at_cond}, define
    \[
        \alpha_{\rm wAT}(\beta,h)
        \defeq
        \beta^2
        \bbE\Brack{
            \operatorname{sech}^2
            \Par{
                \beta\sqrt q Z+h
            }
        }
        =
        \beta^2(1-q).
    \]
    We say that $(\beta,h)$ satisfies the weak AT condition if
    $\alpha_{\rm wAT}(\beta,h) \le 1$. The weak AT boundary
    $h_{\rm wAT}(\beta)$ is characterized by
    $\alpha_{\rm wAT}(\beta,h_{\rm wAT}(\beta))=1$. We also define $q\wAT(\beta) \defeq q(\beta, h\wAT(\beta))$
\end{definition}

Since $\operatorname{sech}^4(u)\le\operatorname{sech}^2(u)$ for all $u$, we have that $\alpha_{\rm AT}(\beta,h)
    \le
    \alpha_{\rm wAT}(\beta,h)$ for all $(\beta, h)$. Thus, the weak AT region (above the weak AT line) is a subset of the AT region.

\subsection{Preliminaries}
\label{ssec:at_wat_scales}

In this section, we prove two key preliminary results. The first (Lemma~\ref{lem:at_line}) computes asymptotics of the boundaries $h\AT$, $h\wAT$ as a function of $\beta$. The second (Lemma~\ref{lem:gaussian_cut_concentration}) formalizes a reduction from large field strength $h$ to concentration on high-magnetization states.

\textbf{Boundary asymptotics.} We first recall a common Laplace estimate that will be used repeatedly. 

\begin{lemma}
    \label{lem:sech_laplace}
    Let \( \phi(t) \defeq \frac{1}{\sqrt{2\pi}} \exp(-\frac{t^2}{2})\) be the Gaussian density, and let $p\in\Brace{2,4}$. Suppose that there are positive $\{q_\beta, h_\beta\}_{\beta > 0}$ satisfying, as $\beta \to \infty$,
    $q_\beta\to1$, $b_\beta \defeq \frac{h_\beta} \beta \to\infty$, and
    $\frac{b_\beta} \beta\to0$. Then
    \begin{equation}
        \label{eq:sech_laplace}
        \bbE\Brack{
            \operatorname{sech}^p
            \Par{
                \beta\sqrt{q_\beta}Z+\beta b_\beta
            }
        }
        =
        \frac{
            I_p
        }{
            \beta\sqrt{q_\beta}
        }
        \phi\Par{
            \frac{
                b_\beta
            }{
                \sqrt{q_\beta}
            }
        }
        \Par{1+o(1)},
    \end{equation}
    where
    \[
        I_2
        =
        \int_{\R}\operatorname{sech}^2(u)\,\mathrm du
        =
        2,
        \qquad
        I_4
        =
        \int_{\R}\operatorname{sech}^4(u)\,\mathrm du
        =
        \frac{4}{3}.
    \]
\end{lemma}

\begin{proof}
We first perform a change of variables
$u=\beta(\sqrt{q_\beta}z+b_\beta)$, which gives
\begin{equation}\label{eq:cov_integral}
    \bbE\Brack{
        \operatorname{sech}^p
        \Par{
            \beta\sqrt{q_\beta}Z+\beta b_\beta
        }
    }
    =
    \frac{
        \phi\Par{b_\beta/\sqrt{q_\beta}}
    }{
        \beta\sqrt{q_\beta}
    }
    \int_{\R}
    \operatorname{sech}^p(u)
    \exp\Par{
        \frac{
            b_\beta u
        }{
            \beta q_\beta
        }
        -
        \frac{
            u^2
        }{
            2\beta^2q_\beta
        }
    }
    \mathrm du.
\end{equation}
Since $q_\beta\to1$ and $\frac{b_\beta} \beta\to0$, the factor $\exp(\frac{b_\beta u}{\beta q_\beta} - \frac{u^2}{2\beta^2q_\beta})$
converges pointwise to one. Moreover, for sufficiently large
$\beta$, $\frac{b_\beta}{\beta q_\beta} \to 0$, so we have
\[
    \exp\Par{
        \frac{
            b_\beta u
        }{
            \beta q_\beta
        }
        -
        \frac{
            u^2
        }{
            2\beta^2q_\beta
        }
    }
    \le
    \exp\Par{\frac{\Abs{u}}{2}}.
\]
For $p\in\Brace{2,4}$,
$\operatorname{sech}^p(u)\exp(\frac{|u|}{2})$ is integrable. Applying dominated
convergence then gives \eqref{eq:sech_laplace}.
\end{proof}

We now derive the (identical) asymptotics of $h\AT$ and $h\wAT$.

\begin{lemma}
    \label{lem:at_line}
    As $\beta\to\infty$,
    \[
        h_{\rm AT}(\beta)
        =
        \sqrt{2}\,\beta\sqrt{\log\beta}
        \Par{1+O\Par{\frac{1}{\log\beta}}},
        \quad
        h_{\rm wAT}(\beta)
        =
        \sqrt{2}\,\beta\sqrt{\log\beta}
        \Par{1+O\Par{\frac{1}{\log\beta}}}.
    \]
\end{lemma}

\begin{proof}
We begin by verifying the conditions of Lemma~\ref{lem:sech_laplace}. For $\star\in\Brace{\mathrm{AT},\mathrm{wAT}}$, write
\[
    q_\star\defeq q_\star(\beta),
    \quad
    b_\star\defeq \frac{h_\star(\beta)}{\beta}.
\]
We also drop the index $\beta$ from these two sequences for simplicity.
First, along the weak AT boundary,
$1-q_{\rm wAT}=\beta^{-2}$, so $q_{\rm wAT}\to1$. Similarly, along the AT boundary,
\[
    \bbE\Brack{
        \operatorname{sech}^4
        \Par{
            \beta\sqrt {q_{\rm AT}} Z+h_{\rm AT}(\beta)
        }
    }
    =
    \frac{1}{\beta^2}.
\]
Since
$1-q_{\rm AT}=\bbE[\operatorname{sech}^2(\beta\sqrt {q_{\rm AT}} Z+h_{\rm AT}(\beta))]$,
Cauchy-Schwarz gives $1-q_{\rm AT}\le\beta^{-1}$, so
$q_{\rm AT}\to1$.

Next, for either value of $\star$, the boundary equation also implies
$b_\star \to\infty$. Indeed, if $b_\star$
were bounded along a subsequence, then restricting the integral \eqref{eq:cov_integral} to
$u\in[-1,1]$ would already give a lower bound of order $\beta^{-1}$, because $q_\star \to 1$, $b_\star$ is bounded, and $\sech^p(u) = \Omega(1)$ for $u \in [-1, 1]$. This would
contradict the boundary equations, which say that \eqref{eq:cov_integral} evaluates to $\beta^{-2}$.

Finally, we claim that $\frac{b_\star}{\beta} \to0$.  Otherwise, $b_\star\ge\eps\beta$ along a subsequence
for some $\eps>0$. Now, split the integral \eqref{eq:cov_integral} along the events $Z \ge -b_\star / 2\sqrt{q_\star}$ or $Z \le -b_\star / 2\sqrt{q_\star}$. In the former region, the change of variables gives $u \ge \frac{\eps\beta^2}{2}$, so the corresponding integral is $\exp(-\Omega(\beta^2))$. Similarly, the latter region has a probability bounded by $\exp(-\Omega(\beta^2))$, and $\sech$ is pointwise bounded. Thus, the entire integral is $\exp(-\Omega(\beta^2))$, again contradicting that it equals $\beta^{-2}$ by definition.

Lemma~\ref{lem:sech_laplace} therefore applies and we obtain that 
\[
    \phi\Par{
        \frac{
            b_{\rm AT}
        }{
            \sqrt {q_{\rm AT}}
        }
    }
    =
    \frac{
        3\sqrt {q_{\rm AT}}
    }{
        4\beta
    }
    \Par{1+o(1)}, 
    \quad 
    \phi\Par{
        \frac{
            b_{\rm wAT}
        }{
            \sqrt {q_{\rm wAT}}
        }
    }
    =
    \frac{
        \sqrt {q_{\rm wAT}}
    }{
        2\beta
    }
    \Par{1+o(1)}.
\]
Expanding with the definition of $\phi$, and then taking logarithms, then gives for
$\star\in\Brace{\mathrm{AT},\mathrm{wAT}}$,
\[
    \frac{b_\star^2}{2q_\star}
    =
    \log\beta+O(1) \implies b_\star^2 = 2q_\star \log\beta + O(1) = 2\log\beta - 2(1-q_\star)\log\beta + O(1).
\]
Since $1-q_{\rm AT}\le\beta^{-1}$ and
$1-q_{\rm wAT}=\beta^{-2}$, the desired claims follow:
\[
    b_\star^2
    =
    2\log\beta+O(1) \implies
    b_\star
    =
    \sqrt{2\log\beta}
    \Par{1+O\Par{\frac{1}{\log\beta}}}.
\]
\end{proof}

\textbf{Magnetization from field strength.} To conclude the section, we show that taking $h$ large in \eqref{eq:sk_uniform_field} implies that $\pi_{\beta, h}$ is concentrated on high-magnetization states. 
For $s\in[0,1]$, we let
$H_2(s)
    \defeq
    -s\log s-(1-s)\log(1-s)$
denote the binary entropy, with the convention $0\log 0=0$.

\begin{lemma}\label{lem:gaussian_cut_concentration}
    Fix $\delta\in(0,\half)$. With probability at least $1-\delta$
    over $\mj$ in Model~\ref{model:sk}, the following holds simultaneously for every
    $\rho\in(0,\half)$ and $\gamma\ge0$. If
    \begin{equation}
        \label{eq:gaussian_cut_field_condition}
        h
        \ge
        \frac{H_2(\rho)}{2\rho}
        +
        \beta
        \sqrt{
            \frac{2(1-\rho)}{\rho}
            \Par{
                H_2(\rho)
                +
                \frac 1 d\log \frac d \delta
            }
        }
        +
        \gamma,
    \end{equation}
    then
    \[
       \Pr_{\vx \sim \pi_{\beta, h}}\Brack{N_-(\vx) > \rho d}
        \le
        d\exp\Par{-2\gamma\rho d}.
    \]
\end{lemma}

\begin{proof}
For a fixed $S\subseteq[d]$, define \( Y_S\defeq\sum_{\substack{i\in S\\j\notin S}}\mj_{ij}\). We claim that simultaneously for all $S \subseteq [d]$,
\begin{equation}\label{eq:y_lb}Y_S \ge -t\Par{\frac{|S|}{d}},\text{ where } t(s) \defeq d\sqrt{2s(1 - s)\Par{H_2(s) + \frac 1d \log \frac d \delta}},\end{equation}
with probability $\ge 1 - \delta$. To see this, fix some $r \in [d]$. 
Under Model~\ref{model:sk}, we have $Y_S \sim \mathcal N(0, ds(1-s))$ for all $|S| = r$ and $s \defeq \frac r d$. Then the standard Gaussian tail bound gives
\[
    \Pr[-Y_S\ge t(s)]
    \le
    \exp\Par{-dH_2(s)}
    \frac{\delta}{d}.
\]
Since $\binom dr\le\exp(dH_2(s))$, we conclude that \eqref{eq:y_lb} holds for all $|S| = r$ with probability $\ge 1 - \frac \delta d$, and then a union bound over all nontrivial layers $r \in [d]$ gives the claim.

Next, letting $w$ be the unnormalized weight in \eqref{eq:sk_uniform_field}, a direct calculation gives
$\log\frac{w(S)}{w(\varnothing)}
    =
    -2h|S|-2\beta Y_S$.
Thus, again letting $s = \frac r d$ for some $r \in [d]$,
\begin{equation}\label{eq:one_layer_weight}
    \frac{\sum_{S:|S|=r}w(S)}{w(\varnothing)}
    \le
    \exp\Par{
        2sd
        \Par{
            \frac{H_2(s)}{2s}
            +
            \frac{\beta t(s)}{sd}
            -
            h
        }
    }.
\end{equation}
It remains to bound the right-hand side uniformly over $s \ge \rho$. First, we can directly check that $\frac{\mathrm d}{\mathrm ds}(\frac{H_2(s)}{s}) = \frac{\log(1 - s)}{s^2} < 0$, so $\frac{H_2(s)}{s}$ is decreasing in $s$. This implies
\[\frac{t(s)}{sd} = \sqrt{\frac{2(1-s)}{s}\Par{H_2(s) + \frac 1 d \log \frac d \delta}}\]
is decreasing in $s$, because every term in the square root is decreasing. Hence for every $s \ge \rho$, the assumed lower bound \eqref{eq:gaussian_cut_field_condition} implies
\[\frac{H_2(s)}{2s} + \frac{\beta t(s)}{sd} - h \le \frac{H_2(\rho)}{2\rho} + \frac{\beta t(\rho)}{\rho d} -h \le -\gamma.\]
Summing \eqref{eq:one_layer_weight} over all $r>\rho d$ and using that the
partition function is $\ge w(\varnothing)$ proves the claim.
\end{proof}

\subsection{Sampling around $\1_d$}
\label{ssec:all_one_sampling}
\label{ssec:gaussian_cut_concentration}

In this section, we give the basic variant of our result for sampling from \eqref{eq:sk_uniform_field}. This variant shows that as $\beta \to \infty$, when the field strength $h$ exceeds the threshold $h\AT(\beta)$ in Lemma~\ref{lem:at_line} by roughly a $\sqrt{2}$ factor, we can sample from $\pi_{\beta, h}$ in polynomial time.

\begin{theorem}
\label{thm:unrestricted_high_field_sk}
    Let $\beta\ge1$ and $\delta\in(0,\half)$, and suppose that
    \eqref{eq:delta_condition} holds with $\delta\gets \frac \delta 2$.
    If     \begin{equation}
        \label{eq:high_field_sampling_threshold}
        h
        \ge
        \frac{H_2(\rho_\beta)}{2\rho_\beta}
        +
        \beta
        \sqrt{
            2\frac{1-\rho_\beta}{\rho_\beta}
            \Par{
                H_2(\rho_\beta)
                +
                \frac 1 d\log \frac{2d}{\delta}
            }
        }
        +
        \frac{1}{2\rho_\beta d}
        \log\frac{2d}{\delta},
    \end{equation}
    where $\rho_\beta \defeq \frac{c}{\beta^2 \log(e\beta)}$ for a universal constant $c > 0$,
    then with
    probability at least $1-\delta$ over the SK model (Model~\ref{model:sk}), there is a
    polynomial-time algorithm that outputs $\vx$ satisfying $        \TV{\Law(\vx),\pi_{\beta,h}}
        \le
        \delta$.
    Furthermore, for $\delta = \poly(\frac 1 d)$ and fixed $\beta$, as $d \to \infty$, the right-hand side
    of~\eqref{eq:high_field_sampling_threshold} converges to
    \begin{equation}
        \label{eq:thermodynamic_sampling_curve}
        h_{\rm HM}(\beta) \defeq \frac{H_2(\rho_\beta)}{2\rho_\beta}
    +
    \beta
    \sqrt{
        2(1-\rho_\beta)
        \frac{H_2(\rho_\beta)}{\rho_\beta} 
    } = 2\beta\sqrt{\log\beta}(1 + o(1)).
    \end{equation}
\end{theorem}

\begin{proof}
Throughout this proof, set
\[k\defeq\lfloor\rho_\beta d\rfloor,\quad \gamma \defeq \frac 1 {2\rho_\beta d}\log \frac {2d} \delta.\] 
Applying
Lemma~\ref{lem:gaussian_cut_concentration} with failure probability
$\frac \delta 2$, sparsity lower bound parameter $\rho_\beta$, and $\gamma$ as defined above then implies that with probability $\ge 1 - \frac \delta 2$ over $\mj$, 
\begin{equation*}
    \Pr_{\vx \sim \pi_{\beta, h}}\Par{N_-(\vx) > k} \le d\exp(-2\gamma\rho_\beta d) \le \frac \delta 2.
\end{equation*}
Next, to sample from $\pi_{\beta, h}$ conditioned on $N_-(\vx) \le k$, assume $k \ge 1$, else it suffices to output $\1_d$. By choosing $c$ sufficiently small,
$k\le \rho_\beta d$ is in the range required by
Corollary~\ref{cor:fast_mixing_sparse_set_dobrushin_gen}. Then applying Corollary~\ref{cor:fast_mixing_sparse_set_dobrushin_gen} with
accuracy and failure probability $\frac \delta 2$ gives the desired sample $\vx$, upon flipping $1$s and $-1$s consistently with our convention in this section. A union bound then gives both the total failure probability of $\delta$ over the draw $\mj$, and the overall accuracy $\delta$ to the target $\pi_{\beta, h}$.

Finally, we prove the asymptotic claims.  For fixed $\beta$ and $\delta = \poly(\frac 1 d)$, the finite-$d$ corrections in \eqref{eq:high_field_sampling_threshold} vanish as $d\to\infty$, so the limit is
\begin{equation*}
    h_{\rm HM}(\beta)
    =
    \frac{H_2(\rho_\beta)}{2\rho_\beta}
    +
    \beta
    \sqrt{
        2(1-\rho_\beta)
        \frac{H_2(\rho_\beta)}{\rho_\beta}
    }.
\end{equation*}
The conclusion follows because as $\beta \to \infty$, we have 
\[
    \frac{H_2(\rho_\beta)}{\rho_\beta}
    =
    \log\frac1{\rho_\beta}
    +
    1
    +
    O(\rho_\beta),
    \quad
    \log\frac1{\rho_\beta}
    =
    2\log\beta+o(\log\beta).
\]
\end{proof}

\subsection{Sampling around the mean}
\label{ssec:oracle_centering}

We next give a stronger variant of Theorem~\ref{thm:unrestricted_high_field_sk} that removes the $\sqrt{2}$ factor overhead in our lower bound on $h$, assuming the ability to compute the signs of the mean magnetization vector. Concretely, Theorem~\ref{thm:unrestricted_high_field_sk} is not adapted to the actual realization of $\mj$. Instead, we show that knowledge of
\[
    \vtau_i^\star
    \defeq
    \operatorname{sign}(\vm_i),\text{ where } 
    \vm_i
    \defeq
    \bbE_{\pi_{\beta,h}}\Brack{\vx_i}\text{ for all } i \in [d],
\]
for a fixed $\mj$, where $\operatorname{sign}(0) \defeq 1$, allows us to recenter the algorithm and improve our $h$ range. Roughly speaking, the idea is to use overlap concentration to redefine our notion of sparsity as disagreement with $\vtau_i^\star$, rather than disagreement with $\1_d$ as used in Section~\ref{ssec:all_one_sampling}.

We next set up some notation for this section. For
$\vtau\in\calX^d$, let
\[
    N_{\vtau}(\vx)
    \defeq
    \Abs{
        \Brace{
            i\in[d]:
            \vx_i\neq\vtau_i
        }
    }.
\]
Also, for two i.i.d.\ draws $(\vx^{(1)}, \vx^{(2)}) \sim \pi_{\beta, h}^{\otimes 2}$, we define the overlap quantity
\[R_{1, 2} \defeq \frac{\inprod{\vx^{(1)}}{\vx^{(2)}}}{d}.\]
Observe that independence gives $\E[R_{1, 2}] = \frac 1 d \norm{\vm}_2^2$, which is a quantity depending on the realized $\mj$. We next state a stronger result from \cite{rajaraman2026markov} which shows that above the weak AT line, $R_{1, 2}$ concentrates around the deterministic quantity $q(\beta, h)$, independent of $\mj$.

\begin{proposition}[Lemma~4.14, \cite{rajaraman2026markov}]
    \label{prop:wat_overlap_concentration}
    Suppose that $(\beta,h)$ satisfies the weak AT condition in
    Definition~\ref{def:wat_cond}. Then there exists
    $C_{\beta,h}>0$ such that
    \begin{equation}
        \label{eq:wat_exponential_overlap_concentration}
        \bbE_{\mj}
        \bbE_{\pi_{\beta,h}^{\otimes2}}
        \Brack{
            \exp\Par{
                \frac{
                    d
                    \Par{
                        R_{1,2}-q(\beta,h)
                    }^2
                }{
                    C_{\beta,h}
                }
            }
        }
        \le
        2.
    \end{equation}
\end{proposition}

As a corollary, we upgrade Proposition~\ref{prop:wat_overlap_concentration} into a high-probability distance bound to $\vtau^\star$, replacing the direct sparsity notion from Lemma~\ref{lem:gaussian_cut_concentration}. Our strategy is to first relate the random quantity $\norm{\vm}_2^2$ to $q(\beta, h)$ using Proposition~\ref{prop:wat_overlap_concentration}, and then to relate $\norm{\vm}_2^2$ to the Hamming distance $N_{\vtau^\star}$ using \eqref{eq:oracle_hamming_expectation}.

\begin{lemma}\label{lem:oracle_center_condition}
    Suppose that $(\beta,h)$ satisfies the weak AT condition. Then
    there exists $C_{\beta,h}>0$ such that for
    every $\rho,\delta\in(0,\half)$, with probability at least
    $1-\delta$ over $\mj$,
    \begin{equation}
        \label{eq:oracle_center_condition_delta}
        \pi_{\beta,h}
        \Par{
            N_{\vtau^\star}>\rho d
        }
        \le
        \frac{
            1-q(\beta,h)
        }{
            2\rho
        }
        +
        \frac{1}{2\rho}
        \sqrt{
            \frac{
                C_{\beta,h}
                \log\frac{2}{\delta}
            }{d}
        }.
    \end{equation}
\end{lemma}

\begin{proof}
Set \( q_d(\mj) \defeq \frac{\twonorm{\vm}^2}{d}\), and recall that $\bbE_{\pi_{\beta,h}^{\otimes2}}[R_{1,2}] = q_d(\mj)$. We first derive
\begin{equation}
    \label{eq:oracle_hamming_expectation}
    \frac{1}{d}
    \bbE_{\pi_{\beta,h}}
    \Brack{
        N_{\vtau^\star}(\vx)
    }
    =
    \frac{1}{2d}
    \sum_{i=1}^d
    \Par{
        1-\Abs{m_i}
    }
    \le
    \frac{1}{2}
    \Par{
        1 - q_d(\mj)
    }.
\end{equation}
Applying Jensen's inequality conditionally on $\mj$ to
\eqref{eq:wat_exponential_overlap_concentration} then gives
\[
    \bbE_{\mj}
    \Brack{
        \exp\Par{
            \frac{
                d
                \Par{
                    q_d(\mj)-q(\beta,h)
                }^2
            }{
                C_{\beta,h}
            }
        }
    }
    \le
    2.
\]
Hence, for every $t>0$,
\[
    \Pr_{\mj}
    \Brack{
        q_d(\mj)
        <
        q(\beta,h)-t
    }
    \le
    2
    \exp\Par{
        -\frac{dt^2}{C_{\beta,h}}
    }.
\]
On the complementary event,
\eqref{eq:oracle_hamming_expectation} and Markov's inequality under
$\pi_{\beta,h}$ give
\[
    \pi_{\beta,h}
    \Par{
        N_{\vtau^\star}>\rho d
    }
    \le
    \frac{
        1-q(\beta,h)+t
    }{
        2\rho
    }.
\]
Taking \( t = \sqrt{ \frac{ C_{\beta,h} \log\frac{2}{\delta} }{d}}\) proves the claim.
\end{proof}

Lemma~\ref{lem:oracle_center_condition} shows that to apply our bounded-magnetization SK sampler (Corollary~\ref{cor:fast_mixing_sparse_set_dobrushin_gen}), we have reduced the problem to making $1 - q(\beta, h)$ smaller than the sparsity parameter $\rho_\beta \approx (\beta^2 \log \beta)^{-1}$. The weak AT condition only gives $1 - q(\beta, h) \le \beta^{-2}$, so we choose a slightly larger field strength, still asymptotic to the weak AT scale in Lemma~\ref{lem:at_line}:
\begin{equation}
    \label{eq:oracle_field}
    h_{\rm OC}(\beta)
    \defeq
    \beta
    \sqrt{
        2\log\beta
        +
        2\log\log\beta
        +
        2\log\log\log\beta
    },
    \quad
    q\OC(\beta)
    \defeq
    q\Par{
        \beta,
        h_{\rm OC}(\beta)
    }.
\end{equation}

Lemma~\ref{lem:oracle_field_scale} next shows that the field strength in \eqref{eq:oracle_field} satisfies the weak AT condition, and thus we can bound the sparsity of its induced model using Lemma~\ref{lem:oracle_center_condition}.

\begin{lemma}
    \label{lem:oracle_field_scale}
    There exist universal constants $C,\beta_0>0$ such that, for every
    $\beta\ge\beta_0$, and $h \ge h\OC(\beta)$,
    \begin{equation}
        \label{eq:oracle_q_rate}
        1-q(\beta, h)
        \le
        \frac{C}{
            \beta^2
            \log\beta
            \log\log\beta
        }.
    \end{equation}
    Consequently, $(\beta,h)$ satisfies the weak AT
    condition.
\end{lemma}

\begin{proof}
We first claim that $q(\beta, h)$ is nondecreasing in $h$ if $(\beta, h)$ strictly satisfies the weak AT condition, so it suffices to prove the result for $h = h\OC(\beta)$. Let $T(q, h) \defeq \E \tanh^2(\beta\sqrt{q}Z + h)$, so that by definition, $q = T(q, h)$. Then, performing a Gaussian integration by parts gives
\[\frac{\partial}{\partial q} T(q, h) = \beta^2 \E\Brack{\sech^2\Par{\beta\sqrt{q}Z + h} \Par{1 - 3\tanh^2\Par{\beta\sqrt{q}Z + h}}} \le \beta^2(1 - q) < 1,\]
where the last inequality used the weak AT condition. Similarly,
\[\frac{\partial}{\partial h} T(q, h) = 2\E\Brack{\tanh\Par{\beta\sqrt{q}Z + h}\sech^2\Par{\beta\sqrt{q}Z + h}} > 0,\]
because $\tanh(\cdot)\sech^2(\cdot)$ is odd and positive on $\R_{> 0}$, and $\beta\sqrt{q}Z + h$ is centered around the positive value $h > 0$.
Now, implicit differentiation gives
\[q'(h) = \frac{\partial}{\partial q} T(q(h), h) q'(h) + \frac{\partial}{\partial h} T(q(h), h) \implies q'(h) = \frac{\frac{\partial}{\partial h} T(q(h), h)}{1 - \frac{\partial}{\partial q} T(q(h), h)} > 0,\]
as desired. For the rest of the proof we take $h = h\OC(\beta)$.

Write $b_\beta\defeq \frac{h_{\rm OC}(\beta)}{\beta}$.
We first verify Lemma~\ref{lem:sech_laplace}'s hypotheses. From \eqref{eq:oracle_field}, the conditions $b_\beta \to \infty$ and $\frac{b_\beta}{\beta} \to 0$ are immediate. Moreover, from the fixed-point equation defining $q_\beta \defeq q\OC(\beta)$,
\[1 - q_\beta = \E\Brack{\sech^2\Par{\beta\sqrt{q_\beta} Z + \beta b_\beta}} \le \sech^2\Par{\frac{\beta b_\beta}{2}} + \Pr\Brack{Z < -\frac{b_\beta}{2}} \to 0,\]
where the inequality split the expectation based on whether $Z \ge -\frac{b_\beta}{2}$ or not, and used that $\sech^2 \le 1$ pointwise and $q_\beta \le 1$. Thus,
Lemma~\ref{lem:sech_laplace} applies with $p = 2$, yielding
\[
    1-q_\beta
    =
    \frac{2}{\beta\sqrt{q_\beta}}
    \phi\Par{
        \frac{b_\beta}{\sqrt{q_\beta}}
    }\Par{1 + o(1)} \le \frac{C}{\beta} \phi(b_\beta),
\]
for a universal constant $C$,
where the last inequality used that $q_\beta \le 1$, $\phi$ is decreasing on $\R_{\ge 0}$, and $q_\beta \to 1$ as $\beta \to \infty$.
Finally, substituting the definition \eqref{eq:oracle_field} gives
\[
    \phi(b_\beta)
    =
    \frac{1}{
        \sqrt{2\pi}
        \beta
        \log\beta
        \log\log\beta
    },
\]
and combining the above two displays proves \eqref{eq:oracle_q_rate}. Finally,
$\beta^2(1-q_\beta)\le \frac{C}{\log\beta\log\log\beta} <1$ as $\beta \to \infty$, which verifies the weak AT condition with strict inequality. Since we earlier showed $q'(h) > 0$ whenever $\beta^2(1 - q(h)) < 1$, all $h > h\OC(\beta)$ also strictly satisfy the weak AT condition.
\end{proof}

We are finally ready to give our main result.

\begin{theorem}
\label{thm:oracle_centered_sampling}
Let $\beta \ge 1$ and $\delta \in (0, \half)$, and suppose that \eqref{eq:delta_condition} holds with $\delta\gets \frac \delta 2$.
    Assume that $h \ge h\OC(\beta)$ defined in \eqref{eq:oracle_field}, and define $\rho_\beta \defeq \frac{c}{\beta^2\log(e\beta)}$ for a universal constant $c > 0$. Also assume
    \begin{equation}\label{eq:parameter_range}\log\log\beta\ge \frac{2C}{\delta},
    \quad
    d\ge \frac{4C_{\beta, h}\log \frac 4 \delta}{\delta^2 \rho_\beta^2},\end{equation}
    where $C$ is a universal constant and $C_{\beta, h}$ is from Lemma~\ref{lem:oracle_center_condition}, and that we are given
    \[\vtau^\star_i \defeq \sign\Par{\E_{\pi_{\beta, h}}\Brack{\vx_i}} \text{ for all } i \in [d].\]
    Then with probability at least $1-\delta$ over the SK model (Model~\ref{model:sk}), there is a polynomial-time algorithm that outputs $\vx$ satisfying \(\TV{\Law(\vx),\pi_{\beta,h}}\le\delta\). Further,
    \[h\OC(\beta) = \sqrt{2} \beta\sqrt{\log \beta}\Par{1 + o(1)}.\]
\end{theorem}
\begin{proof}
Let $\md \defeq \diag{\vtau^\star}$ throughout. We first claim that the fixed-magnetization sampler in Theorem~\ref{thm:sk_trickle} (and hence, the bounded-magnetization sampler in Corollary~\ref{cor:fast_mixing_sparse_set_dobrushin_gen}) is invariant to replacing $\mj \gets \md\mj\md$. To see this, all but one of the estimates in Lemma~\ref{lem:conditions_sk} remain unchanged under this replacement, because $\md \vu$ has the same $\ell_2$ and $\ell_1$ norms as $\vu$ for any vector $\vu$, and $(\md \mj \md) \circ (\md \mj \md) = \mj \circ \mj$. The only difference is that $\norm{\md\mj\md\1_d}_\infty$ could be larger because $\md \1_d$ is no longer independent of $\mj$, but the same bound holds up to a $\sqrt{d}$ factor. This factors into Theorem~\ref{thm:sk_trickle}'s initial $\chi^2$ divergence bound, which only affects the claimed runtime by a polynomial factor after taking logarithms. 

To complete the proof, take $k \defeq \lfloor \rho_\beta d\rfloor$. Lemmas~\ref{lem:oracle_center_condition} and \ref{lem:oracle_field_scale}, and our assumed bounds \eqref{eq:parameter_range}, imply
\[\Pr_{\vx \sim \pi_{\beta, h}}\Brack{N_{\vtau^\star}(\vx) > k} \le \frac \delta 2.\]
The rest of the proof is identical to Theorem~\ref{thm:unrestricted_high_field_sk}, under the change of variables $\vx \gets \md \vx$.
\end{proof}

In summary, Theorem~\ref{thm:oracle_centered_sampling} covers a range of $h$ with a lower bound $h\OC(\beta)$ roughly a $\sqrt{2}$ factor smaller than Theorem~\ref{thm:unrestricted_high_field_sk}'s $h\HM(\beta)$. In particular, up to a $1 + o(1)$ factor, $h\OC(\beta)$ matches the AT and weak AT thresholds $h\AT(\beta)$ and $h\wAT(\beta)$ derived in Lemma~\ref{lem:at_line}. In comparison, a recent work by \cite{BandeiraElAlaouiRodder26} derives a similar polynomial-time sampling result, but examining their proof (particularly, Corollary 1.2 combined with the improvement in Remark 3.3) implies their threshold on the field strength $h$ scales as $h = \Omega(\beta^2 \sqrt{\log \beta})$, i.e., roughly a $\beta$ factor larger than $h\AT(\beta)$.

We remark that Theorem~\ref{thm:oracle_centered_sampling} is a conditional result that requires access to $\vtau^\star$, the signs of the mean magnetization vector; we leave the efficient computation of $\vtau^\star$ as an important open problem. Additionally, our definition of $h\OC(\beta)$ includes potentially unnecessary low-order terms, arising due to a discrepancy between the weak AT region's definition (which gives $1 - q(\beta, h) \le \beta^{-2}$) with the requirements of our sampler in Corollary~\ref{cor:fast_mixing_sparse_set_dobrushin_gen} (which requires a sparsity parameter $\rho \approx (\beta^2 \log \beta)^{-1}$). It is also an interesting open problem to improve the thresholds imposed by our approach, with the goal of efficient sampling across the entire AT region.
\section{Scaling of Infinite $\Delta$-Regular Tree Threshold}\label{app:critical}

In this section, we provide a calculation that explains asymptotics induced by the ``tree threshold'' $\beta_c(\Delta)$, which parameterizes the results of  \cite{Carlson2022, KuchukovaPappikPerkinsYap2025}. Concretely, $\beta_c(\Delta) = \log(\frac \Delta {\Delta - 2})$ is a critical inverse temperature under which the Ising model on the infinite $\Delta$-regular tree undergoes a phase transition. When $\beta < \beta_c(\Delta)$, there is a unique fixed point of a certain message-passing recursion on the infinite tree, and when $\beta > \beta_c(\Delta)$, two new fixed points appear. We defer an overview to \cite{Lyons89, Mossel06}; here we calculate consequences of this threshold asymptotically.

The main algorithmic result of \cite{Carlson2022} at low temperatures $\beta > \beta_c(\Delta)$ is their Theorem 2(a). They show that there is a corresponding \emph{critical magnetization} $\eta_{\Delta, \beta, 1}^+$ at which the fixed-magnetization Ising model undergoes a computational phase transition. This $\eta_{\Delta, \beta, 1}^+$ is the expected spin at the root of the infinite $\Delta$-regular tree, at one of the new fixed points emerging when $\beta > \beta_c(\Delta)$ (see Section 4, \cite{Carlson2022} for these calculations).
Note that \cite{Carlson2022} defines the magnetization $\eta \in [-1, 1]$ as (in our notation) $-1 + \frac{2k}{d}$, so that $k = d$ corresponds to $\eta = 1$ (and $k = 0$ corresponds to $\eta = -1$). Rearranging, Theorem 2(a) of \cite{Carlson2022} applies at sparsity levels
\[\eta < -\eta_{\Delta, \beta, 1}^+ \implies k < \frac d 2 \Par{1 - \eta_{\Delta, \beta, 1}^+}. \]
It remains to understand $1 - \eta_{\Delta, \beta, 1}^+$. From Section 1.1, \cite{Carlson2022}, letting $L$ be the largest root of
\[L = \Par{\Delta - 1}\textup{arctanh} \Par{\tanh(L)\tanh\Par{\frac \beta 2 }},\]
denoting $\eta_\beta \defeq \eta_{\Delta, \beta, 1}^+$ for short,
\[\eta_\beta = \tanh\Par{L + \textup{arctanh}\Par{\tanh(L)\tanh\Par{\frac \beta 2 }}}.\]
Let $A \defeq \textup{arctanh} (\tanh(L)\tanh(\frac \beta 2)) = \frac 1 {\Delta - 1} L$. Then,
\[\tanh A = \tanh\Par{(\Delta - 1)A} \tanh\Par{\frac \beta 2} \implies \tanh\Par{\frac \beta 2} = \frac{\tanh A}{\tanh\Par{(\Delta - 1)A}}.\]
Now using the approximation $\tanh(c) = 1 - \Theta(\exp(-2c))$ for large $c$, we obtain
$A = \frac \beta 2 + O_\Delta(1)$.
Finally, plugging this back into our definition of $\eta_\beta$, we have
\[1 - \eta_\beta = 1 - \tanh\Par{\Delta A} = \Theta\Par{\exp\Par{-2\Delta A}} = \exp\Par{-\Delta\beta + O_\Delta(1)}.\]
Therefore, for the sparsity regime
\begin{equation}\label{eq:logd_sparsity_regime}k < d\exp\Par{-\Delta\beta + O_\Delta(1)}\end{equation}
to be meaningful (i.e., the inequality above does not hold only when $k = 0$), the result of \cite{Carlson2022} is limited to inverse temperatures of $\beta = O(\log d)$ for constant-degree graphs. 

The main algorithmic result of the subsequent work \cite{KuchukovaPappikPerkinsYap2025} is their Theorem 1.1, which is parameterized at a slightly different \emph{critical magnetization} $\eta_{\beta, a}$, which is at least as large as the $\eta_\beta$ from before. Therefore, Theorem 1.1 of \cite{KuchukovaPappikPerkinsYap2025} is also restricted to the regime \eqref{eq:logd_sparsity_regime}. 

Finally, we note that the comparison in this section is purely a statement about the allowable temperatures that our sparsity-aware framework tolerates (vs.\ the prior works \cite{Carlson2022, KuchukovaPappikPerkinsYap2025}), for our models of interest. For example, directly applying our results to the graph-based Ising models in these prior works only permits rapid mixing in the high-temperature regime $\beta = O(\frac 1 k)$. However, for other well-studied models, e.g., Models~\ref{model:sk} and~\ref{model:hopfield}, our results allow taking inverse temperatures as large as $\beta = \poly(d)$ when $k$ is sufficiently small.
\end{document}